\documentclass[a4paper]{article} 
\usepackage{a4wide}

\usepackage{amsmath,amsfonts,amsthm,bbm}
\usepackage{graphicx}
\usepackage{xspace}
\usepackage{url,doi}
\usepackage{subcaption}
\usepackage[export]{adjustbox}

\usepackage[linesnumbered,lined,boxed]{algorithm2e}

\makeatletter
\def\cramped                           % "Cramped" list style.
 {
\parskip0pt\@topsep0pt       % put after \begin{itemize}
  \itemsep0pt\parsep0pt
  \settowidth{\labelwidth}{\@itemlabel}         
  \advance\leftmargin-\labelsep                 
  \advance\leftmargin-\labelwidth               
  \parshape1\@totalleftmargin\linewidth
}        
\makeatother
\graphicspath{{figures/}}

\DeclareMathOperator	{\arcosh}   {arccosh}
\DeclareMathOperator	{\arsinh}   {arcsinh}

\DeclareMathOperator	{\arcot}	{arccot}

\DeclareMathOperator	{\sys}   	{sys}

\DeclareMathOperator	{\real}   	{Re}

\DeclareMathOperator	{\DT}   	{DT}
\DeclareMathOperator	{\area}		{area}

\newcommand{\R}                 {\mathbb{R}\xspace}
\renewcommand{\H}               {\mathbb{H}\xspace}
\newcommand{\C}			        {\mathbb{C}\xspace}
\newcommand{\E}			        {\mathbb{E}\xspace}

\newcommand{\Z}					{\mathbb{Z}\xspace}

\newcommand{\mD}					{\mathbb{D}}
\newcommand{\mQ}				{\mathbb{Q}}

\newcommand{\br}		[1]		{\mathopen{}\left(#1\right)\mathclose{}\xspace}
\newcommand{\sqbr}		[1]		{\mathopen{}\left[#1\right]\mathclose{}\xspace}
\newcommand{\braces}	[1]		{\mathopen{}\left\{#1\right\}\mathclose{}\xspace}
\newcommand{\angular}	[1]		{\mathopen{}\left<#1\right>\mathclose{}\xspace}

\newcommand{\inv}		[1]		{#1^{-1}\xspace}
\newcommand{\nbt}    {{T^{\text{\tiny nbr}}}}
\newcommand{\ltr}    {{T_p^{\text{\tiny loc}}}}

\newcommand{\Do} 				{\widetilde{D}_g}
\newcommand{\Db} 				{\widetilde{D}_2}
\newcommand{\I}					{\mathcal{I}\xspace}

\renewcommand{\P}				{\mathcal{P}\xspace}
\renewcommand{\S}				{\mathcal{S}\xspace}
\newcommand{\diam} 		[1]		{\delta({#1})\xspace}
\newcommand{\dummy}				{\mathcal{Q}_g\xspace}
\newcommand{\dummyq}			{\mathcal{Q}\xspace}
\newcommand{\hlen}		[1]		{\operatorname{length}(#1)\xspace}

\newcommand{\seg}		[1]		{\sqbr{#1}\xspace}
\newcommand{\Gg}                {\Gamma_g}
\newcommand{\Gb}                {\Gamma_2}

\newcommand{\eg}                  {\mathbbm{1}\xspace}
\newcommand{\M}				    {\mathbb{M}\xspace}
\newcommand{\Mg}				{\M_g\xspace}
\newcommand{\proj} 				{\pi_g\xspace}
\newcommand{\Bolza}				{\M_2\xspace}
\newcommand{\sysg}              {\sys\br{\Mg}\xspace}
\newcommand{\N}                 {\mathcal{N}_g\xspace}
\newcommand{\Nnude}                 {\mathcal{N}\xspace}
\newcommand{\Nb}                 {\mathcal{N}_2\xspace}
\newcommand{\W}                 {\mathcal{W}_g\xspace}

\newcommand{\orientation} 		{\textsc{Orientation}\xspace}
\newcommand{\incircle} 			{\textsc{InCircle}\xspace}

\newcommand{\h}[1]{{\boldsymbol{#1}}} 
\newcommand{\can}[1]  	        {\h{{#1}^c}\xspace} % the canonical repsentative
\newcommand{\canindex}[2]   {\h{#1}^\h{c}_{#2}\xspace} % the canonical repsentative

\newcommand{\ph}                {\h{p}} 
\newcommand{\qh}                {\h{q}} 
\newcommand{\rh}                 {\h{r}} 
\newcommand{\mh}                 {\h{m}} 
\newcommand{\zh}				{\h{z}}
\newcommand{\wh}				{\h{w}}
\newcommand{\xh}				{\h{x}}
\renewcommand{\th}                 {\h{t}}
\newcommand{\vh}                 {\h{v}}
\newcommand{\uh}                 {\h{u}}
\newcommand{\eh}                 {\h{e}}
\newcommand{\sh}                 {\h{s}}
\newcommand{\lh}                 {\h{l}}
\newcommand{\gammah}			{\h{\gamma}}

\newcommand{\ch}				{\h{c}}
\newcommand{\Uh}				{\h{U}}
\newcommand{\Ch}				{\h{C}}

\newcommand{\first}[1]       {{#1}^{\star}\xspace}

\newcommand{\len}		[1]		{\ell\br{#1}\xspace}

\DeclareMathOperator	{\tr}   	{Tr}
\newcommand{\trSqr}[1]{\tr^2(#1)}
\newcommand{\abs}		[1]		{\mathopen{}\left|\,#1\,\right|\mathclose{}\xspace}
\newcommand{\conj}      [1]     {\bar{#1}\xspace}

\newcommand{\hp}                        {\mathbb{D}\xspace}
\newcommand{\dth}       [1]             {\DT_{\hp}\br{#1}\xspace}
\newcommand{\dts}	[1]		{\DT_{\mathbb{M}}\br{#1}\xspace}
\newcommand{\dtsg}	[1]		{\DT_{\Mg}\br{#1}\xspace}

\newcommand{\dte}	[1]		{\DT_{\E}\br{#1}\xspace}

\newcommand{\cgal} 				{\textsc{Cgal}\xspace}
\newcommand{\core} 				{\textsc{Core}\xspace}
\newcommand{\expr}              {\texttt{CORE::Expr}\xspace}

\newtheorem{lemma}{Lemma}

\newtheorem{proposition}[lemma]{Proposition}
\newtheorem{theorem}[lemma]{Theorem}

\theoremstyle{definition}
\newtheorem{definition}[lemma]{Definition}

\newtheorem{remark}[lemma]{Remark}

\usepackage{xcolor}

\title{Delaunay triangulations of generalized Bolza surfaces \thanks{This work was partially supported by the grant(s)
    ANR-17-CE40-0033 of the French National Research Agency ANR (project SoS)
    and
    INTER/ANR/16/11554412/SoS of the Luxembourg National Research fund FNR.}}

\author{
  Matthijs Ebbens\thanks{Bernoulli Institute for Mathematics, Computer
    Science and Artificial Intelligence, University of Groningen,
    Netherlands. \texttt{y.m.ebbens@rug.nl}}
  \and
  Iordan Iordanov\thanks{Université de Lorraine, CNRS, Inria, LORIA,
    F-54000 Nancy, France. \texttt{i.m.iordanov@gmail.com}}
  \and 
  Monique Teillaud\thanks{Université de Lorraine, CNRS, Inria, LORIA,
    F-54000 Nancy, France. \texttt{Monique.Teillaud@inria.fr}}
  \and
  Gert Vegter\thanks{Bernoulli Institute for Mathematics, Computer
    Science and Artificial Intelligence, University of Groningen,
    Netherlands. \texttt{g.vegter@rug.nl}}
}

\begin{document}
\maketitle
 
\begin{abstract}
  The Bolza surface can be seen as the quotient of the hyperbolic
  plane, represented by the Poincar\'e disk model, under the action of the group generated by the hyperbolic
  isometries identifying opposite sides of a regular octagon 
  centered at the origin. We consider
  \emph{generalized} Bolza surfaces~$\Mg$, where the octagon is
  replaced by a regular $4g$-gon, leading to a genus $g$ surface.
  We propose an extension of Bowyer's algorithm to these surfaces. In
  particular, we compute the value of the systole of~$\Mg$. We also
  propose algorithms computing small sets of points on~$\Mg$ that are used to initialize Bowyer's algorithm. 
\end{abstract}

%%%%%%%%%%%%%%%%%%%%%%%%%%%%%%%%%%%%%%%%%%%%%%%%%%%%%%%%%%%%%

\section{Introduction}

Lawson's well-known incremental
algorithm that computes Delaunay
triangulations using edge flips in the Euclidean
plane~\cite{l-scsi-77} has recently been proved to generalize on
hyperbolic surfaces~\cite{dst-fgths-20}. However, the experience
gained in the \cgal\ project for many years has shown that Bowyer's
algorithm~\cite{Bowyer81} leads to a cleaner code, much easier to
maintain; there is actually work in progress in \cgal\ to replace Lawson's
flip algorithm, in triangulation packages that are still using it, by Bowyer's
algorithm. In the context of quotient spaces 
Bowyer's algorithm was used already in the \cgal\ packages for
3D flat quotient spaces~\cite{cgal:ct-p3t3-perm} and for the
Bolza surface~\cite{cgal-it-p4ht2-perm}. To the best of our knowledge, the latter
package is the only available software for a hyperbolic
surface. The advantages of Bowyer's algorithm largely
compensate the constraint that it intrinsically requires that the
Delaunay triangulation be a simplicial complex.

In this paper, we study the
extension of this approach to what we call \emph{generalized Bolza
	surfaces}. A closed orientable hyperbolic surface $\M$ is isometric
to a quotient $\mD/\Gamma$, where $\Gamma$ is a discrete group of
orientation preserving
isometries acting on the hyperbolic plane, represented here as the
Poincar\'e disk $\mD$.
See Section~\ref{sec:mathematical-preliminaries} for some mathematical background on the hyperbolic plane and
hyperbolic surfaces.
The universal cover of such a surface is the hyperbolic
plane, with associated projection map
$\pi:\mD\rightarrow\mD/\Gamma$.
The \emph{generalized Bolza group} $\Gamma_g$, $g \geq 2$, is the
(discrete) group generated by the orientation preserving isometries that pair opposite sides of
the regular $4g$-gon $D_g$, centered at the origin and with angle
sum $2\pi$ (unique up to rotations).
The \emph{generalized Bolza surface} $\Mg$, of genus $g$, is defined as
the hyperbolic surface $\mD/\Gamma_g$, with projection map
$\proj:\mD\rightarrow\mD/\Gamma_g$. In particular, $\Bolza$ is the
usual Bolza surface.

We denote by $\sys\br{\M}$ the \emph{systole} of a closed hyperbolic surface $\M$, i.e., the
length of a shortest non-contractible curve on~$\M$, and, for a set of
points $\dummyq\subset \M$, by $\diam{\dummyq}$
the diameter of the largest disks in $\mD$ that do not contain any
point in $\pi^{-1}(\dummyq)$. 
In Section~\ref{sec:computingDT} we recall the following
\emph{validity condition}~\cite{ct-dtceo-16,btv-dtosl-16}:
If a finite set $\dummyq$ of points on the surface $\M$
satisfies the inequality
\begin{flushright}
	$\diam{\dummyq} < \frac 12 \sys\br{\M}$ \hspace*{2cm} (condition~\eqref{condition} in Proposition~\ref{thm:condition})
\end{flushright}
then Bowyer's algorithm can be extended to the computation of the
Delaunay triangulation of any finite set of points $\S$ on $\M$ containing
$\dummyq$. Before actually inserting the input points, the algorithm performs a preprocessing
step consisting of computing the Delaunay triangulation of an appropriate
(but small) set $\dummyq$ satisfying this validity condition; following the
terminology of previous papers~\cite{btv-dtosl-16,ct-dtceo-16}, we
refer to the points of $\dummyq$ as \emph{dummy points}. When
sufficiently many and well-distributed input points have been
inserted, the validity condition is satisfied without the dummy
points, which can then be removed. This approach was used in
the \cgal\ package
for the Bolza surface $\Bolza$~\cite{it-idtbs-17,cgal-it-p4ht2-perm}. 

Other practical approaches for (flat) quotient spaces start by
computing in a finite-sheeted covering space~\cite{ct-dtceo-16} or in
the universal covering space~\cite{ort-gcpdt-20}, thus requiring the
duplication of some input points.  In contrast to this approach, our
algorithm proceeds directly on the surface, thus circumventing the
need for duplicating any input points.  While the number of copies of
duplicated points in approaches using covering spaces is small, the
number of duplicated input points is always much larger than the
number of dummy points that could instead be added to the set of input
points in our approach.  Moreover, to the best of our knowledge, the
number of required copies in the case of hyperbolic surfaces is
largely unknown; first bounds have been obtained recently~\cite{ben-2020}.

\paragraph*{Results.}~\\
We describe the extension of Bowyer's algorithm to the case of the generalized Bolza
surface~$\Mg$ in Section~\ref{sec:computingDT}, and we 
derive bounds on the number of dummy points necessary to satisfy the
validity condition 
(Propositions~\ref{thm:upperbounddummypoints}
and~\ref{thm:minimumnumberofpoints} in Section~\ref{sec:bound-dummy}), yielding the following result:

\begin{theorem}
	The number of dummy points that must be added on $\Mg$ to satisfy the
	validity condition~\eqref{condition} grows as $\Theta(g)$. 
	\label{th:bounds_dummy}
\end{theorem}

In Section~\ref{sec:symmetric-hyperbolic-surfaces}, we give an explicit value for the systole of
$\Mg$: 
\begin{theorem}\label{thm:systolevalue}
	The systole of the surface $\Mg$ is given by $\varsigma_g$, where $\varsigma_g$ is defined as 
	\[ \varsigma_g:=2\arcosh\left(1+2\cos(\tfrac{\pi}{2g})\right).\]
\end{theorem}
This generalizes a result of Aurich and Steiner~\cite{poshs-as-88}
who derived the identity $\cosh\tfrac{1}{2}\sys(\M_2)= 1 + \sqrt{2}$
for the systole of the Bolza
surface ($g=2$), with a method that is quite different from our proof.

Then, in Section~\ref{sec:combinatoricsofsmallDT}, we propose two algorithms to compute dummy points. 
The first algorithm is based on the well-known Delaunay
refinement algorithm for mesh generation~\cite{ruppert1995}. Using a
packing argument, we prove that it provides an asymptotically optimal
number of dummy points (Theorem~\ref{thm:simplealgorithm}). 
The second
algorithm modifies the refinement algorithm so as to yield a symmetric dummy point set, at the expense of a
slightly larger output size $\Theta(g\log g)$
(Theorem~\ref{thm:symmetricalgorithm}); this symmetry may be
interesting for some applications~\cite{cff-bhp-11}. The two algorithms
have been implemented and we quickly present results for small genera $g=2$ and $g=3$.

Finally, in Section~\ref{sec:representation-of-delaunay-triangulations}, we
describe the data structure that we are using to support the extension
of Bowyer's algorithm to generalized Bolza surfaces. We also discuss
the algebraic degree of the predicates used in the
computations and present experimental results. 

%%%%%%%%%%%%%%%%%%%%%%%%%%%%%%%%%%%%%%%%%%%%%%%%%%%%%%%%%%%%%

\section{Mathematical preliminaries}
\label{sec:mathematical-preliminaries}

In this section we define our notation and present a short introduction on hyperbolic geometry and hyperbolic surfaces~\cite{Beardon1983,r-fhm-06}. 

\subsection{The Poincar{\'e} disk}\label{sec:intropoincaredisk}
The model of the hyperbolic plane we use is the Poincar\'e disk, the open unit disk $\hp$ in the complex
plane equipped with a Riemannian metric of constant Gaussian curvature $K = -1$ \cite{Beardon1983}.
The Euclidean boundary $\hp_\infty$ of the unit disk consists of the points at
infinity or \textit{ideal points} of the hyperbolic plane (which do not belong to $\hp$). 
The geodesic segment $\sqbr{\zh,\wh}$ between points $\zh,\wh\in\hp$ is the (unique) shortest curve
connecting $\zh$ and $\wh$.
A hyperbolic line (i.e., a geodesic for the given metric) in this model is a curve which contains the
geodesic segment between any two of its points. These geodesics are diameters of $\hp$ or circle
arcs whose supporting lines or circles intersect $\hp_\infty$ orthogonally (see
the leftmost frame of Figure~\ref{fig:hyp-transformation}).
A circle in the hyperbolic plane is a Euclidean circle in the Poincar\'e disk, in general with a hyperbolic center and radius that are different from their Euclidean counterparts. 

\begin{figure}[ht]
	\centering
	\includegraphics[height=4.8cm]{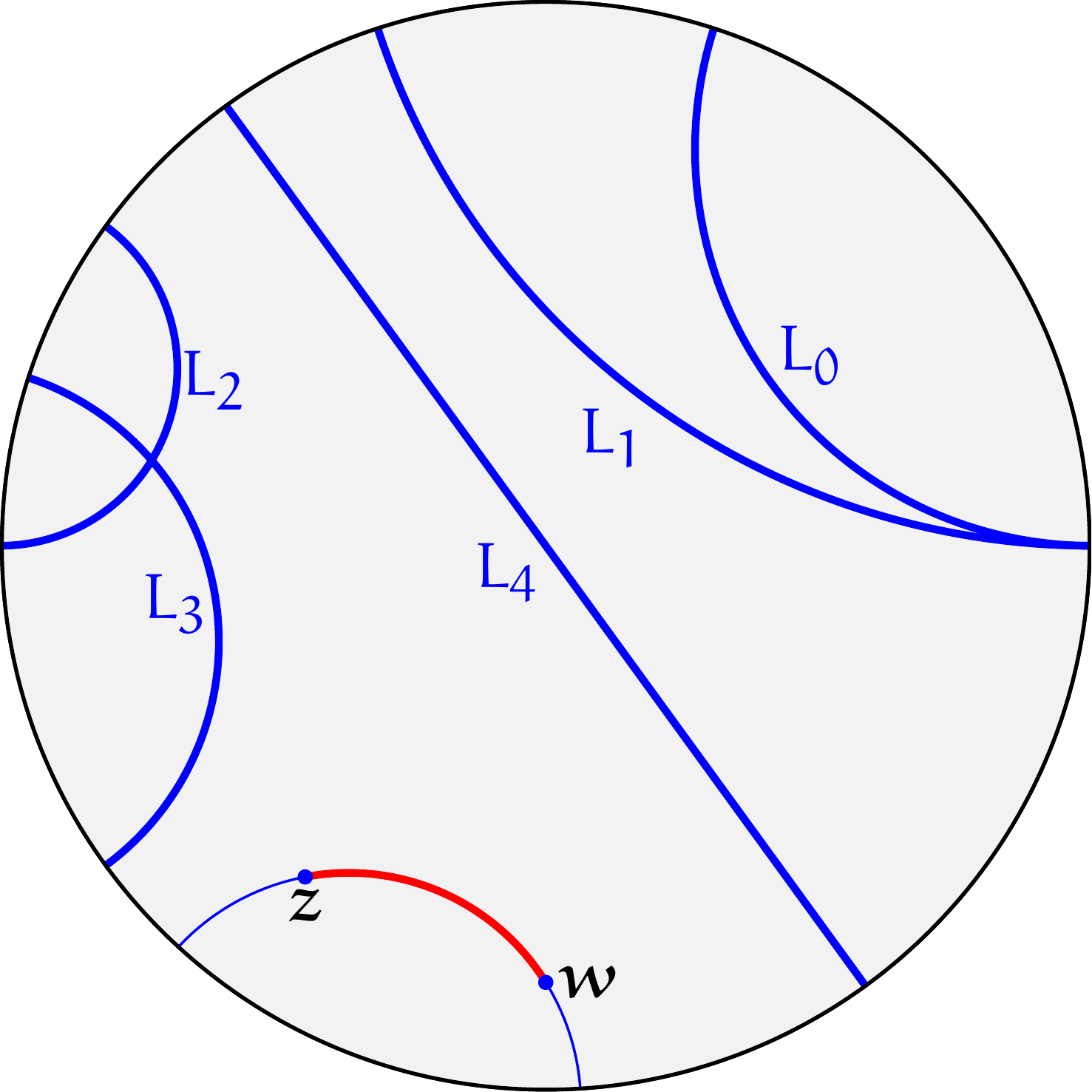}~\qquad~%
	\includegraphics[height=5cm]{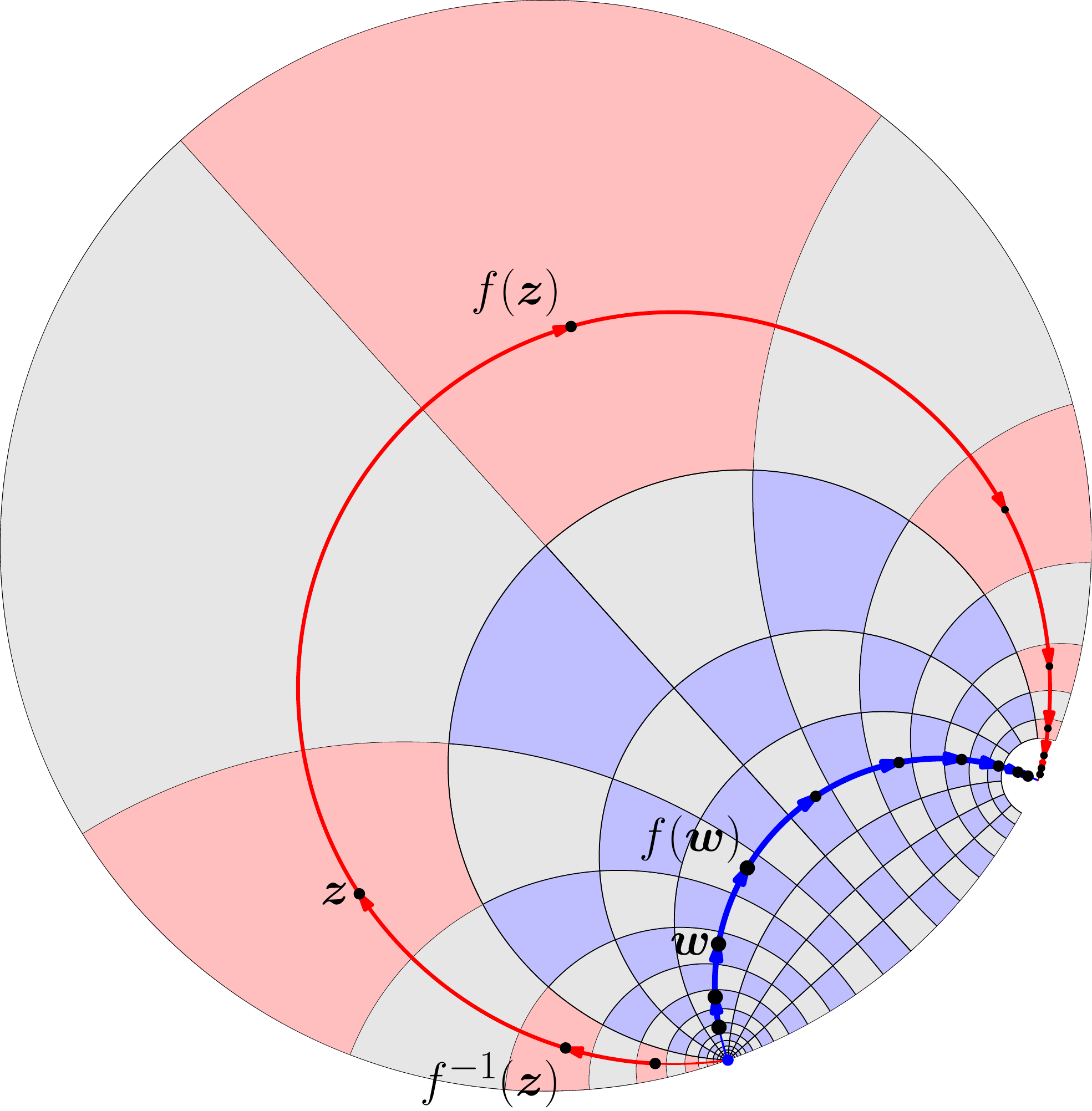}~\qquad\qquad
	
	\caption{Left: the Poincar{\'e} disk model $\mD$ of the hyperbolic plane, with some geodesics. 
		The boundary $\mD_\infty$ does not belong to $\mD$, but consists  of ideal points of $\mD$.
		Geodesics $L_0$ and $L_1$ are parallel (have an ideal point in common), $L_2$ and $L_3$ are
		intersecting and $L_4$ is disjoint from the other geodesics. The points $\boldsymbol{z}$
		and $\boldsymbol{w}$ are connected by a hyperbolic segment.
		\\
		Right: A hyperbolic translation $f$ has two fixed points on the boundary
		$\hp_\infty$ of the Poincar{\'e} disk $\hp$. The axis of $f$ is the (unique) geodesic
		connecting the fixed points of $f$.
		The orbit of point $\wh$ is contained in the axis of $f$.
		The orbit of point $\zh$, which does not lie on the axis of $f$, is contained in an equidistant
		of the axis (an arc of a Euclidean circle through the fixed points).
		The red region containing the point $\zh$ is mapped by $f$ to the red region containing $f(\zh)$.
	}
	\label{fig:hyp-transformation}
\end{figure}

We only consider orientation-preserving isometries of $\hp$, called \textit{isometries} from now on,
which are linear fractional transformations of the form 
\begin{equation}
	\label{eq:hypplaneisometry}
	f(\zh)=\dfrac{a\zh+b}{\conj{b}\zh+\conj{a}},	
\end{equation}
with $a,b\in\C$ such that $\abs{a}^2-\abs{b}^2=1$. 
By representing isometries of the form~\eqref{eq:hypplaneisometry} by either of the two matrices
\begin{equation}
	\label{eq:mat-f}
	\pm
	\begin{pmatrix}
		a & b \\
		\conj{b} & \conj{a}
	\end{pmatrix},
\end{equation}
with $\abs{a}^2-\abs{b}^2 = 1$, composition of isometries corresponds to multiplication of either of
their representing matrices.
The only non-identity isometries we consider are \textit{hyperbolic translations},
which are characterized by having two distinct fixed points on $\mD_\infty$.
An isometry of the form~\eqref{eq:hypplaneisometry} is a hyperbolic translation if and only if
$\trSqr{f} > 4$, where the trace-squared $\trSqr{f}$ of $f$ is the square of the trace $\pm 2
\real{a}$ of the matrices representing $f$.

The $f$-orbit $\{f^n(\zh) \mid n \in\Z\}$ of a point $\zh\in\hp$ is contained in a Euclidean circle
through the fixed points of the hyperbolic translation $f$. 
See Figure~\ref{fig:hyp-transformation}.
Let $d$ denote the distance in the hyperbolic plane.
The \textit{translation length} $\len{f}$ of a hyperbolic translation $f$ is the minimal value of the
displacement function $\zh \mapsto d(\zh,f(\zh))$, which is attained at all points $\zh$ on
the geodesic connecting the two fixed points of $f$. This geodesic is the \textit{axis of $f$}.
The translation length is given by
$  \cosh^2(\tfrac{1}{2}\len{f}) = \tfrac{1}{4}\trSqr{f}$, or, in terms of the matrices~\eqref{eq:mat-f} representing $f$:
\begin{equation}
	\cosh(\tfrac{1}{2}\len{f}) = \abs{\real{a}}.
\end{equation}

\subsection{Hyperbolic surfaces, closed geodesics and systoles}\label{sec:introhyperbolicsurfaces}
In our setting a \textit{hyperbolic surface} is a two-dimensional orientable manifold without boundary which
is locally isometric to the hyperbolic plane. In particular, it has constant Gaussian curvature -1. We will always assume our hyperbolic surfaces to be compact. 
By the uniformization theorem~\cite{Abikoff1981} a hyperbolic surface $\M$ has $\hp$ as its universal
covering space. The surface $\M$ is isometric to the quotient surface $\hp/\Gamma$ of the
hyperbolic plane $\hp$ under the action of a Fuchsian group, i.e., a discrete group $\Gamma$ of
orientation preserving isometries of $\hp$.
The covering projection $\pi : \hp \to \hp/\Gamma$ is a local
isometry.
The orbit $\Gamma \zh$ of a point $\zh \in \hp$ is a discrete subset of $\hp$. Note that
$\Gamma \zh = \pi^{-1}(z)$, with $z = \pi(\zh) \in \M$.
We emphasize that points in the hyperbolic plane $\hp$ are denoted by $\zh,\wh,\ph,\qh$ and so on,
whereas the corresponding points on the surface $\hp/\Gamma$ are denoted by $z,w,p,q$ and so on.
Since $\hp/\Gamma$ is a smooth hyperbolic surface all non-identity elements of $\Gamma$ are hyperbolic
translations.

The distance between points $p$ and $q$ on a hyperbolic surface is given by $\min\{d(\ph,\qh )\;|\;\ph\in\pi^{-1}(p),\qh\in\pi^{-1}(q) \}$ and, abusing notation, is denoted by $d(p,q)$. The projection $\pi$ maps
(oriented) geodesics of $\hp$ to (oriented) geodesics of $\M = \hp/\Gamma$, and it maps the axis of
a hyperbolic translation $f\in\Gamma$ to a \textit{closed} geodesic of $\M$.
Every (oriented) closed geodesic $\gamma$ on $\M$ arises in this way, i.e., there is a hyperbolic translation
$f\in\Gamma$ such that $\gamma$ \textit{lifts} to the axis of $f$.
The axes of two hyperbolic translations $f, f'\in\Gamma$ project to the same closed geodesic of $\M$
if and only if $f'$ is conjugate to $f$ in $\Gamma$
(i.e., iff there is an $h\in\Gamma$ such that $f' =
h^{-1}fh$).

A \textit{simple closed geodesic} of $\M$ is the $\pi$-image of a hyperbolic segment  $[\zh,f(\zh)]$
on the axis of a hyperbolic translation $f$ such that $\pi$ is injective on the open segment $(\zh,f(\zh))$.
The length of this geodesic is equal to the translation length $\len{f}$ of $f$.
For every $L > 0$ the number of simple closed geodesics of $\M$ with length less than $L$ 
is finite, so there is at least one with minimal length. This minimal length is the \textit{systole}
of $\M$, denoted by $\sys\br{\M}$. It is known that 
\begin{equation}\label{eq:systoleinequality}
	\sys(\M)\leq 2\log(4g-2)
\end{equation}
for every hyperbolic surface $\M$ of genus $g$ \cite[Theorem 5.2.1]{b-gscrs-92}. 

A \emph{triangle} $t$ on a hyperbolic surface is the $\pi$-image of a triangle $\th$ in $\hp$ such that $\pi$ is injective on $\th$. Clearly, the vertices of $t$ are the projections of the vertices of $\th$ and the edges of $t$ are geodesic segments.  
A \emph{circle} on a hyperbolic surface is the $\pi$-image of a circle in the hyperbolic plane. In this case, we do not require $\pi$ to be injective on the circle, so the image may have self-intersections.

\subsection{Fundamental domain for the action of a Fuchsian group}
The \textit{Dirichlet region} $D_\ph(\Gamma)$ of a point $\ph\in\hp$ with respect to the Fuchsian group
$\Gamma$ is the closure of the open cell of $p$ in the Voronoi diagram of the infinite discrete set of points $\Gamma \ph$ in $\hp$.
In other words, $D_\ph(\Gamma)=\{\xh\in\hp \mid d(\xh,\ph)\leq d(\xh,f(\ph))\text{ for all }f\in\Gamma
\}$.
Since $\hp/\Gamma$ is compact, every Dirichlet region is a compact convex hyperbolic polygon
with finitely many sides.
Each Dirichlet region $D_\ph(\Gamma)$ is a \emph{fundamental domain} for
the action of $\Gamma$ on $\hp$, i.e., \textit{(i)}~$D_\ph(\Gamma)$ contains
at least one point of the orbit $\Gamma\ph$, and
\textit{(ii)}~if $D_\ph(\Gamma)$ contains more than one point of $\Gamma\ph$
then all these points of $\Gamma\ph$ lie on its boundary.  

\begin{figure}[ht]
	\centering
	\includegraphics[height=5cm]{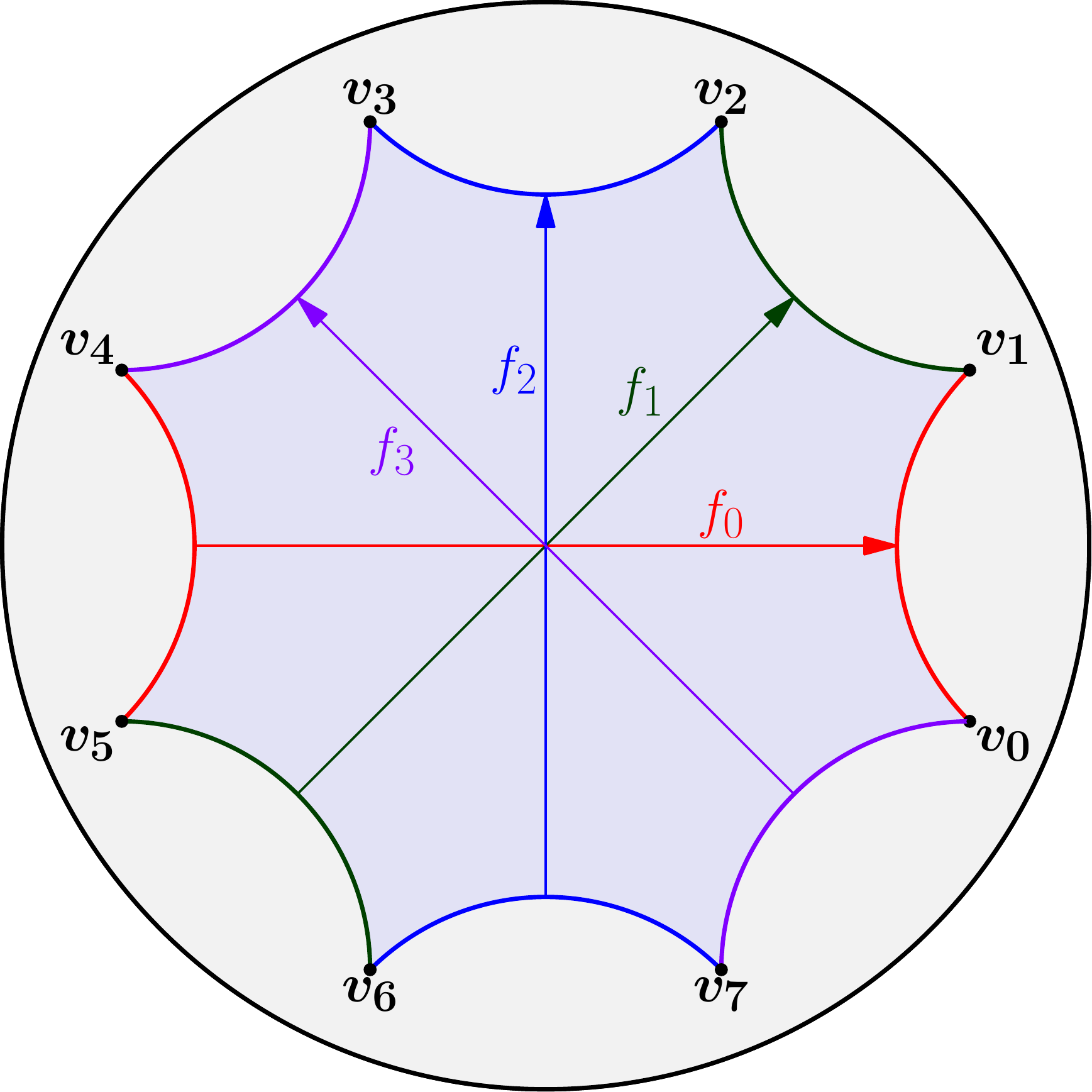}~\qquad~%
	\includegraphics[height=5cm]{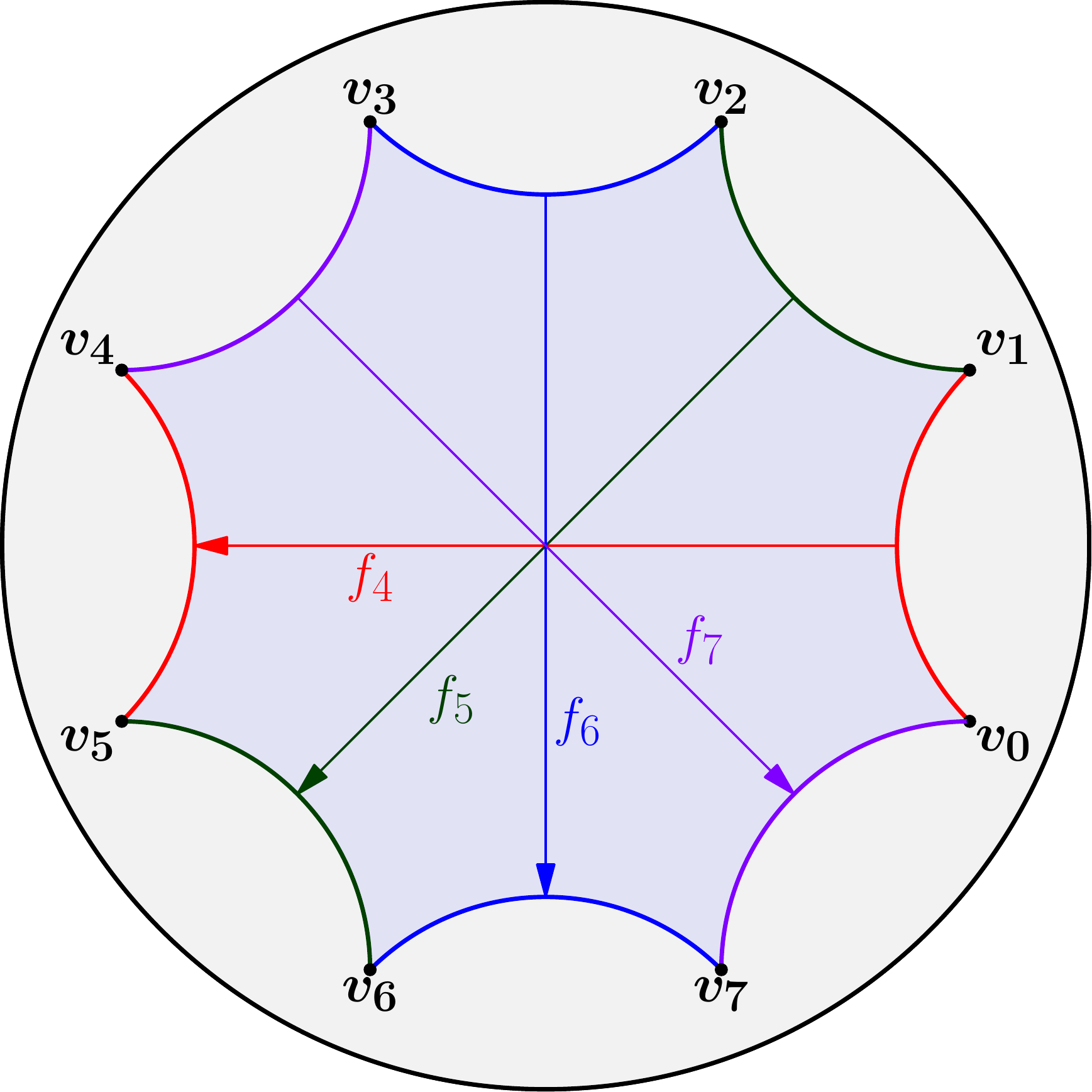}
	\caption{
		The side-pairings $f_0,\ldots,f_3$ of the Bolza surface (of genus 2) pair opposite
		edges of the fundamental octagon (a regular octagon in $\mD$ with angles $\tfrac{1}{4}\pi$).
		Their inverses $f_4,\ldots,f_7$ satisfy $f_{k+4} = f_k^{-1}$.
		The side-pairings generate the Fuchsian group $\Gamma_2$. 
		All vertices are in the same $\Gamma_2$-orbit.
		The composition $f_0 f_5 f_2 f_7 f_4 f_1 f_6f_3$ is the identity $\eg\in\Gamma_2$.
	}
	\label{fig:side-pairings}
\end{figure}

\subsection{Generalized Bolza surfaces\label{sec:generalized}}
\paragraph{The Fuchsian group $\Gamma_g$.}
The generalized Bolza group of genus $g$, $g \geq 2$, is the Fuchsian group $\Gamma_g$ defined in the
following way.
Consider the regular hyperbolic $4g$-gon $D_g$ with angle-sum $2\pi$.
The counterclockwise sequence of vertices is $\vh_0,\ldots,\vh_{4g-1}$, where the midpoint of edge
$[\vh_0,\vh_1]$ lies on the positive real axis.
See Figure~\ref{fig:side-pairings} for $g=2$.
The sides of $D_g$ are $\sh_j, j=0,\ldots,4g-1$, where $\sh_j$ is the side with vertices
$\vh_j$ and $\vh_{j+1}$ (counting indices modulo $4g$).
The orientation preserving isometries $f_0,\ldots,f_{4g-1}$
pair opposite sides of $D_g$.
More precisely, $f_j$ maps $\sh_{j+2g}$ to $\sh_j$, and $\sh_j = f_j(D_g) \cap D_g$.
According to~\eqref{eq:mat-f} the side-pairing $f_j$, $j=0,1,\ldots,4g-1$, is represented by
any of two matrices $\pm A_j$ with determinant $1$. Using some elementary hyperbolic geometry it
can be seen that $A_j$ is given by~\cite{poshs-as-88}
\begin{equation}
	\label{eq:generator-matrix}
	A_j=\begin{pmatrix}
		\cot(\tfrac{\pi}{4g}) & \exp(\tfrac{ij\pi}{2g})\sqrt{\cot^2(\tfrac{\pi}{4g})-1}\,\,
		\\[1.2ex]
		\exp(-\tfrac{ij\pi}{2g})\sqrt{\cot^2(\tfrac{\pi}{4g})-1} & \cot(\tfrac{\pi}{4g})
	\end{pmatrix}.
\end{equation}
By Poincar{\'e}'s Theorem (\cite[Chapter 9.8]{Beardon1983} and \cite[Chapter 11.2]{r-fhm-06})
these side-pairings generate a Fuchsian group, the generalized Bolza group
$\Gamma_g$, all non-identity elements of which are hyperbolic translations.
The polygon $D_g$ is a fundamental domain for the action of this group, and
it is even the Dirichlet region of the origin.

Since $\vh_j = f_jf_{j+1}^{-1}(\vh_{j+2})$, we see that the
element
$f_0f_1^{-1}f_2f_3^{-1}\cdots f_{4g-2}f_{4g-1}^{-1}$ of $\Gamma_g$ maps $\vh_{4g}$ to $\vh_0$.
In other words, $\vh_0$ is a fixed point of this element. Since all non-identity elements of
$\Gamma_g$ are hyperbolic translations, and, hence, without fixed points in $\mD$, this element is the
identity $\eg$ of $\Gamma_g$:
\begin{equation}
	\label{eq:relation-with-inverses}
	f_0f_1^{-1}f_2f_3^{-1}\cdots f_{4g-2}f_{4g-1}^{-1} = \eg.
\end{equation}
For even $j$ we have $f_j = f_{j(2g+1)}$, since we are counting indices modulo $4g$.
Similarly, $f_j^{-1} = f_{j+2g} = f_{j(2g+1)}$, for odd $j$.
Therefore, we can rewrite~\eqref{eq:relation-with-inverses} as
\begin{equation}
	\label{eq:relation}
	\prod_{j=0}^{4g-1}f_{j(2g+1)} =f_{0}f_{2g+1}f_{2(2g+1)}\cdots f_{(4g-1)(2g+1)} = \eg.
\end{equation}
The order of the factors in this product does matter since the group $\Gamma$ is not abelian. 
Equation~\eqref{eq:relation} is usually called the \emph{relation} of $\Gamma_g$.
In addition to~\eqref{eq:relation-with-inverses} and~\eqref{eq:relation}, there are many other ways
to write the relation. By rotational symmetry of $D_g$, conjugating
$\prod_{j=0}^{4g-1}f_{j(2g+1)}$ with the rotation by angle $k\pi/2g$ around the origin yields the
relation $\prod_{j=0}^{4g-1}f_{k+j(2g+1)}=\eg$. 
The latter expression can be rewritten as
\begin{equation}
	\label{eq:relation-k}
	f_k f_{k+1}^{-1} f_{k+2}f_{k+3}^{-1} \cdots f_{k+4g-2} f_{k+4g-1}^{-1} = \eg.
\end{equation}

\paragraph{Neighbors of vertices of the fundamental polygon.}
In the clockwise sequence of Dirichlet regions $h_1(D_g), h_2(D_g),\cdots,h_{4g}(D_g)$
around vertex $\vh_k$ the element $h_j \in \Gamma_g$ is the prefix of length $j$ in the left-hand
side of~\eqref{eq:relation-k}:
\begin{equation}
	\label{eq:neighbors-vk}
	h_j =
	\begin{cases}
		f_kf_{k+1}^{-1}\cdots f_{k+j-2}f_{k+j-1}^{-1}, & \text{ if $j$ is even}, \\[1.2ex]    
		f_kf_{k+1}^{-1}\cdots f_{k+j-2}^{-1}f_{k+j-1}, & \text{ if $j$ is odd}.
	\end{cases}
\end{equation}
Prefixes $h_j$ of length $j \geq 2g$ can be reduced to a word of length $4g-j$
in $f_k, f_{k+1}^{-1},\ldots,f_{k+2g-1}^{-1}$ using relation~\eqref{eq:relation-k} (where the empty
word -- of length zero -- corresponds to $\eg$) and the fact that $f_{j} = f_{j-2g}^{-1}$ for
$j \geq 2g$. More precisely, $h_{4g-j}$ is the prefix of length $j$ in
$f_{k+2g-1}^{-1}f_{k+2g-2}\cdots f_{k+1}^{-1}f_k$, for $j = 0,\ldots,2g$.
Figure~\ref{fig:neighbors-vk} depicts the neighbors of $\vh_k$ for the case $g=2$.
\begin{figure}[htb]
	\centering
	\includegraphics[height=5cm]{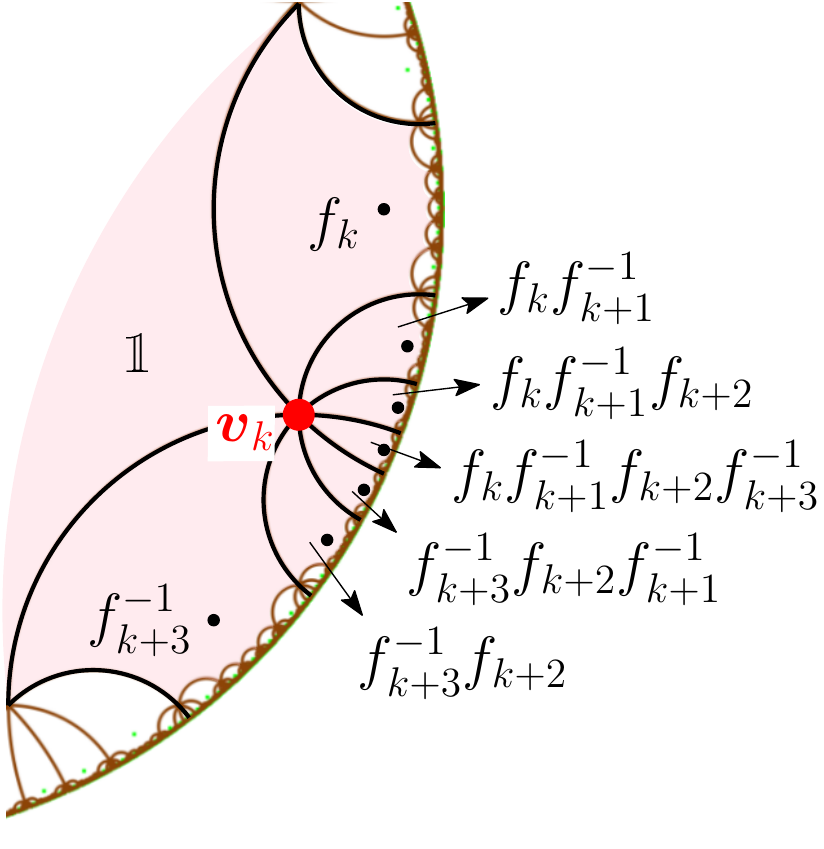}
	\caption{\label{fig:neighbors-vk}%
		Enumeration of the regions around a vertex
		$\vh_k$, $k=0,1,\ldots,7$, for the case $g=2$.
		Note that $f_{k+3}^{-1}f_{k+2}f_{k+1}^{-1}f_k=f_kf_{k+1}^{-1}f_{k+2}f_{k+3}^{-1}$
		by the group relation.} 
\end{figure}

The ordering of neighbors of the vertices of $D_g$ yields an ordering of all regions around $D_g$,
which will play an important role in the data structure for the representation of
Delaunay triangulations in Section~\ref{sec:representation-of-delaunay-triangulations}.
More precisely, we define the set $\N$ of \emph{neighboring translations} as:
\[
\N = \braces{ f \in \Gg \; \big| \; f (D_g) \cap D_g \neq \emptyset }.
\]
Each Dirichlet region sharing an edge or a vertex with the (closed)
domain $D_g$ is the image of $D_g$ under the action of a translation
in $\N$, which is used to label the region.
Also see Figure~\ref{fig:neighbors-vk}.
We denote the union of these \emph{neighboring regions} of 
$D_g$ by $D_{\N}$, so
\begin{equation*}
	D_{\N} = \bigcup_{ f \in \N } f (D_g).  
\end{equation*}
Note that we slightly abuse terminology in the sense that the identity is an element of $\N$, and,
therefore, a neighboring translation, even though it is not a hyperbolic translation.
Also note that $D_g$ is a neighboring region of itself.

\paragraph{The hyperbolic surface $\M_g$.}
The \textit{generalized Bolza surface} of genus $g$ is the hyperbolic surface $\hp/\Gamma_g$,
denoted by $\Mg$. The projection map is $\pi_g:\hp \to\hp/\Gamma_g$. The surface $\Bolza$ is
the classical Bolza surface~\cite{Bolza1887,potrho-abs-91}.

The \emph{original domain} $\Do$ is a subset of $D_g$ containing exactly
one representative of each point on the surface $\Mg$, i.e., of each
orbit under $\Gamma_g$.
The original domain $\Do$ is constructed from the fundamental domain $D_g$ as
follows (see Figure~\ref{fig:original-domain}): $\Do$ and $D_g$ have the same
interior; the only vertex of $D_g$ belonging to $\Do$ is the vertex
$\vh_0$; the $2g$ sides $\sh_{2g},\ldots,\sh_{4g-1}$ of $D_g$ preceding $\vh_0$ (in
counter-clockwise order) are in $\Do$, while the subsequent $2g$ sides
are not.
\begin{figure}[ht]
	\centering
	\includegraphics[]{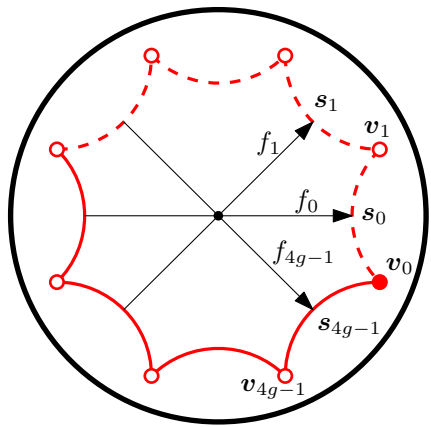}
	\caption{Original domain $\Do$ for $g = 2$. Only vertex
		$\vh_0$ and the solid sides are included in $\Do$.}
	\label{fig:original-domain}
\end{figure}
For a point $p$ on $\Mg$, the \emph{canonical} representative
of $p$ is the unique point of the orbit $\proj^{-1}(p)$ that lies in
$\Do$. 

%%%%%%%%%%%%%%%%%%%%%%%%%%%%%%%%%%%%%%%%%%%%%%%%%%%%%%%%%%%%%

\section{Computing Delaunay triangulations}\label{sec:computingDT} 
%\auth{Monique}
%  \mt{This first version is essentially meant to test whether the order
%    in which things are introduced makes sense. TO BE CHECKED. I would
%    like to introduce the condition that DT on M is a simplicial
%    complex in a smooth and natural way, i.e. make it come from the
%    algorithm (conflict region = disk). However an object cannot be
%    computed before being defined\ldots}
% \me{Could we give a definition of (simplicial) triangulation in the
%   plane, then introduce Bowyer's algorithm, and then generalize to
%   hyperbolic surfaces, commenting that now we really need
%   simpliciality?} 
% \mt{That's precisely what I am trying to do, but it is not that
%   easy. In the plane, a triangulation just cannot have loops or
%   multiple edges because edges of triangles are line segments, that's
%   it. It is a trivial fact, not a constraint of a definition.}

\subsection{Bowyer's algorithm in the Euclidean plane}
There exist various algorithms to compute Delaunay triangulations in
Euclidean spaces. Bowyer's algorithm~\cite{Bowyer81,Watson81} has proved its efficiency in \cgal~\cite{cgal:pt-t3}.
% I removed the citation of the 2D triangulations as they are actually
% using a flip algorithm

Let us focus here on the two-dimensional case. Let $\P$ be a set of points
points in the Euclidean plane $\E$ and $\dte{\P}$ its Delaunay triangulation. Let $\ph \not\in
\P$ be a point in the plane to be inserted in the triangulation. Bowyer's 
algorithm performs the insertion as follows. 
\begin{enumerate}\cramped
	\item Find the set of triangles of $\dte{\P}$ that are
	\textit{in conflict} with $\ph$, i.e., whose
	open circumcribing disks contain $p$;
	\item Delete each triangle in conflict with $\ph$; this creates
	a ``hole'' in the triangulation;
	\item Repair the triangulation by forming new triangles with
	$\ph$ and each edge of the hole boundary to obtain
	$\dte{\P\cup\{\ph\}}$. 
\end{enumerate}
Degeneracies can be resolved using a symbolic
perturbation technique, which actually works in any
dimension~\cite{dt-pdwdt-11}. 

An illustration is given in Figure~\ref{fig:bowyer}.
\begin{figure}[htbp]
	\centering
	\includegraphics[width=0.98\textwidth]{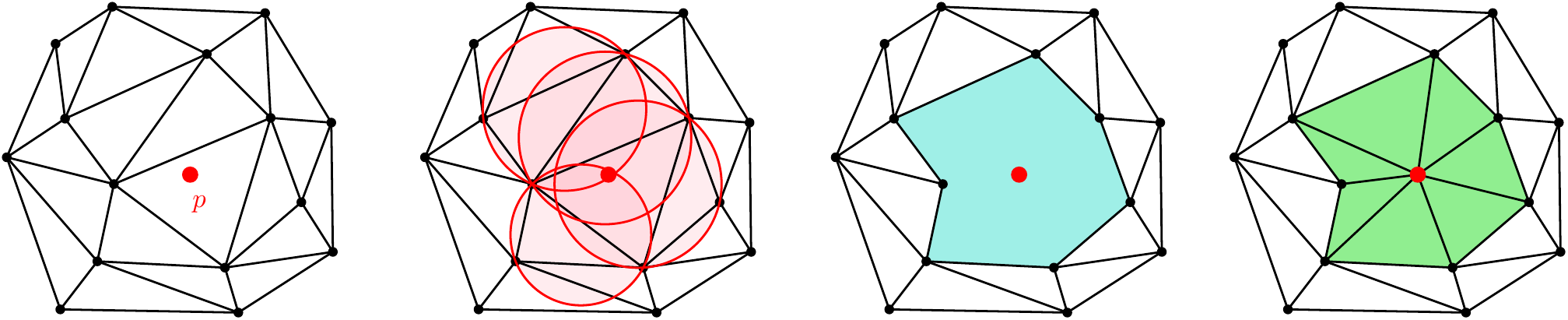}
	\caption[Insertion with Bowyer's incremental algorithm]{Insertion of a point in a Delaunay triangulation with Bowyer's incremental algorithm.}
	\label{fig:bowyer}
\end{figure}
The first step of the insertion of $\ph$ uses geometric
computations, whereas the next two are purely combinatorial. This is
another reason why this algorithm is favored in \cgal: it allows for a
clean separation between combinatorics and geometry, as opposed to an
insertion by flips, in which geometric computations and combinatorial updates
would alternate. 

Note that the combinatorial part heavily relies on the fact that \textit{the
	union of the triangles of $\dte{\P}$ in conflict with $\ph$ is a
	topological disk}. We will discuss this essential property in the
next section.

\subsection{Delaunay triangulations of points on hyperbolic surfaces}
\label{sec:bower-surf}

Let $\M = \mD/\Gamma$ be a hyperbolic surface, as introduced in Section~\ref{sec:introhyperbolicsurfaces}, with the associated projection map $\pi: \mD \rightarrow \M$, and
$F\subset\mD$ a fundamental domain.

Let us consider a triangle $t$ and a point $p$ on $\M$. The
triangle $t$ is said to be \emph{in conflict} with $p$ if the
circumscribing disk of one of the triangles in $\pi^{-1}(t)$ is in
conflict with a point of $\pi^{-1}(p)$ in the fundamental domain. As
noticed in the literature~\cite{bdt-hdcvd-14}, the notion of conflict
in $\mD$ is the same as in $\E$, since for the Poincar\'e disk model,
hyperbolic circles are Euclidean circles (see
Section~\ref{sec:intropoincaredisk}).

Let us now consider a finite set $\P$ of points on the surface $\M$ and a
partition of $\M$ into triangles with vertex set $\P$. 
Assume that the triangles of the partition have no conflict with any of the
vertices. Let $p\not\in\P$ be a point on
$\M$. The region $C_p$ formed by the union of the triangles of the
partition that are in conflict with $p$ might not be a topological
disk (see Figure~\ref{fig:bowyer-handle}). In such a case, the last
step of Bowyer's algorithm could not directly apply, as there are
multiple geodesics between $p$ and any given point on the boundary of~$C_p$.

\begin{figure}[htbp]
	\centering
	\includegraphics{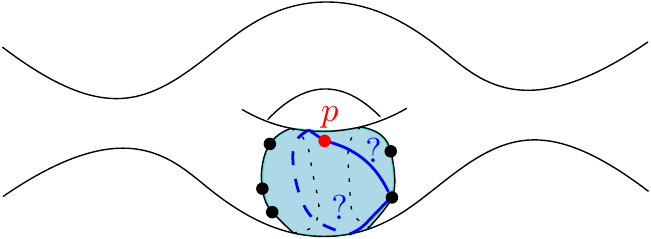}
	\caption[]{Bowyer's insertion is not well defined when
		the conflict region is not a topological disk.}
	\label{fig:bowyer-handle}
\end{figure}

In order to be able to use Bowyer's algorithm on $\M$, the triangles
on $\M$ without conflict with any vertex, together with their edges and their
vertices, should form a \emph{simplicial complex}.
Here, a collection
$\mathcal{K}$ of vertices, edges, and triangles (together called \emph{simplices}) is called a simplicial complex
if it satisfies the following two conditions
(cf~\cite[Chapter 6]{Armstrong2013} and \cite[Chapter 1]{Munkres1984}):

\begin{itemize}\cramped
	\item each face of a simplex of $\mathcal{K}$ is also an element of $\mathcal{K}$; \label{def:sc1}
	\item the intersection of two simplices of $\mathcal{K}$ is either empty or is a
	simplex of $\mathcal{K}$. \label{def:sc2}
\end{itemize}
In other words, the graph of edges of the triangles
should have no loops (1-cycles) or multiple edges
(2-cycles). Note
that, as the set $\P$ is finite, all triangulations considered in this
paper are locally finite, so, we can skip the local finiteness in
the above conditions (see also the discussion
in~\cite[Section~2.1]{ct-dtceo-16}). 

For a set of points $\dummyq\subset \M$ we denote by $\diam{\dummyq}$
the diameter of the largest disks in $\mD$ that do not contain any
point in $\pi^{-1}(\dummyq)$. We will reuse the following result. 

\begin{proposition}[Validity condition \cite{btv-dtosl-16}]
	\label{thm:condition}
	Let $\dummyq \subset \M$ be a set of points such that 
	\begin{equation}
		\diam{\dummyq} < \frac 12 \sys\br{\M}.
		\label{condition}
	\end{equation} 
	Then, for any set of points $\S\subset\M$ such that
	$\dummyq\subseteq\S$, the graph of edges of the projection
	$\pi\br{\dth{\pi^{-1}(\S)}}$ has no 1- or 2-cycles.
\end{proposition}
This condition is stronger than just requiring that the Delaunay
triangulation of $\dummyq$ be a simplicial complex: 
if only the latter condition holds, inserting more
points could create cycles in the
triangulation~\cite[Figure~3]{ct-dtceo-16}; see also Remark~\ref{rem:condition}
below.

The proof is easy, we include it for completeness. 
\begin{proof}
	Assume that condition~(\ref{condition}) holds. For each edge $\eh$ of the
	(infinite) Delaunay triangulation $\dth{\pi^{-1}(\dummyq)}$ in $\mD$, there exists an empty ball
	having the endpoints of $\eh$ on its boundary, so, the length of $\eh$ is
	not larger than $\diam{\dummyq}$. Assume now that there is a 2-cycle
	formed by two edges $\pi(\eh_1)$ and $\pi(\eh_2)$ in
	$\pi\br{\dth{\pi^{-1}(\dummyq)}}$, then the length of the
	non-contractible loop that they form is the sum of the lengths of
	$\eh_1$ and $\eh_2$, which is at most $2 \diam{\dummyq}$ and smaller than
	$\sys(\M)$. This is impossible by definition of $\sys(\M)$, so, there is
	no 2-cycle in $\pi\br{\dth{\pi^{-1}(\dummyq)}}$. 
	
	As the diameter of the largest empty disks does not increase with the
	addition of new points, the same holds for any set $\S\supseteq\dummyq$.
\end{proof}

The most obvious example of a set that does not satisfy the validity
condition is a single point: each edge of the projection is a
1-cycle. The condition is satisfied when the set contains sufficiently
many and well-distributed points. 

\begin{definition} Let $\S\subset \M$ be a set of points satisfying the validity
	condition~\eqref{condition}. The \emph{Delaunay triangulation of $\M$ defined by $\S$}
	is then defined as $\pi\br{\dth{\pi^{-1}(\S)}}$ and denoted
	by~$\dts{\S}$.
\end{definition}

As for the Bolza surface~\cite{btv-dtosl-16},
Proposition~\ref{thm:condition} naturally suggests a way to adapt 
Bowyer's algorithm to compute $\dts{\P}$ for a given set $\P$ of
$n$ points on $\M$: 
\begin{itemize}
	\item Initialize the triangulation as the Delaunay triangulation
	$\dts{\dummyq}$ of $\M$
	defined by an artificial set of \emph{dummy points} $\dummyq$
	that satisfies condition~\eqref{condition}; 
	\item Compute incrementally the Delaunay triangulation
	$\dts{\dummyq\cup\P}$ by inserting the points
	$p_1,p_2,\ldots,p_k,\ldots,p_n$ of $\P$
	one by one, i.e., for each new point $p_k$:
	\begin{itemize}
		\item find all triangles of the Delaunay triangulation
		$\dts{\dummyq\cup\{p_1,\ldots,p_{k-1}\}}$ that are in conflict with $p_k$; let $C_{p_k}$ 
		denote their union; since $\dummyq$ satisfies the validity
		condition, $C_{p_k}$ is a topological disk;
		\item delete the triangles in $C_{p_k}$; 
		\item repair the triangulation by forming new triangles with $p_k$
		and each edge of the boundary of $C_{p_k}$;
	\end{itemize}
	\item Remove from the triangulation all points of $\dummyq$ whose
	removal does not violate the validity condition.
\end{itemize}
We ignore degeneracies
here; they can be resolved as in the case of flat orbit
spaces~\cite{ct-dtceo-16}. 
Depending on the size and distribution of the input set $\P$, the
final Delaunay triangulation of $\M$ might have some or all of the
dummy points as vertices. If $\P$ already satisfies the validity
condition then no dummy point is left.

\subsection{Bounds on the number of dummy points}\label{sec:bound-dummy}

In the following proposition we show that a dummy point set exists
and give an upper bound for its cardinality. The proof is
non-constructive, but we will construct dummy point sets for generalized Bolza surfaces in Section~\ref{sec:combinatoricsofsmallDT}. 

\begin{proposition}\label{thm:upperbounddummypoints}
	Let $\M$ be a hyperbolic surface of genus $g$ with
	systole $\sys(\M)$. Then there exists a point set
	$\dummyq\subset\M$ satisfying the validity condition~\eqref{condition} with cardinality
	$$ |\dummyq| \leq \dfrac{2(g-1)}{\cosh(\tfrac{1}{8}\sys(\M))-1}.$$
\end{proposition}
\begin{proof}
	Let $\dummyq$ be a maximal set of points such that for all
	distinct $p,q\in\dummyq$ we have
	$d(p,q)\geq\tfrac{1}{4}\sys(\M)$. By maximality, we know that
	for all $x\in\M$ there exists $p\in\dummyq$ such that $d(x,p)<
	\tfrac{1}{4}\sys(\M)$: if this is not the case, i.e., if there
	exists $x\in\M$ such that $d(x,p)\geq\tfrac{1}{4}\sys(\M)$ for
	all $p\in\dummyq$, then we can add $x$ to $\dummyq$, which contradicts maximality of $\dummyq$. 
	Hence, for any $x\in\M$ the largest disk centered at $x$ and not containing any points of $\dummyq$ has diameter less than $\tfrac{1}{2}\sys(\M)$, which implies $\diam{\dummyq}<\tfrac{1}{2}\sys(\M)$.\\ 
	To prove the statement on the cardinality of $\dummyq$, denote
	the open disk centered at $p\in\dummyq$ with radius $R$ by
	$B_p(R)$. The disk $B_p(\tfrac{1}{8}\sys(\M))$ for $p\in\dummyq$ is embedded in $\M$, because its radius is smaller than $\tfrac{1}{2}\sys(\M)$. Furthermore, by construction of $\dummyq$, 
	$$B_p(\tfrac{1}{8}\sys(\M))\cap B_q(\tfrac{1}{8}\sys(\M))=\emptyset$$
	for all distinct $p,q\in\dummyq$. Hence, the cardinality of $\dummyq$ is bounded from above by the number of disjoint embedded disks of radius $\tfrac{1}{8}\sys(\M)$ that we can fit in $\M$. We obtain
	$$ |\dummyq| \leq \dfrac{\area(\M)}{\area(B_p(\tfrac{1}{8}\sys(\M)))}=\dfrac{4\pi(g-1)}{2\pi(\cosh(\tfrac{1}{8}\sys(\M))-1)}=\dfrac{2(g-1)}{\cosh(\tfrac{1}{8}\sys(\M))-1}.$$
\end{proof}

Similarly, in the next proposition we state a lower bound for the cardinality of a dummy point set.

\begin{proposition}\label{thm:minimumnumberofpoints}
	Let $\M$ be a hyperbolic surface of genus $g\geq 2$. Let
	$\dummyq$ be a set of points in $\M$ such that the validity
	condition~\eqref{condition} holds. Then
	$$|\dummyq|>\left( \dfrac{\pi}{\pi-6\arcot(\sqrt{3}\cosh(\tfrac{1}{4}\sys(\M)))}-1\right)\cdot 2(g-1).$$
\end{proposition}

\begin{proof}
	Denote the number of vertices, edges and triangles in the (simplicial) Delaunay triangulation $\dts{\dummyq}$ of $\M$ by $k_0,k_1$ and $k_2$, respectively. We know that $3k_2=2k_1$, since every triangle consists of three edges and every edge belongs to two triangles. By Euler's formula, 
	$$k_0-k_1+k_2=2-2g,$$
	so 
	\begin{equation}\label{eq:eulertriangle}
		k_2=4g-4+2k_0.
	\end{equation}
	Consider an arbitrary triangle $t$ in $\dts{\dummyq}$. Because
	the validity condition holds, the circumradius of $t$ is
	smaller than $\tfrac{1}{2}\sys$. It can be shown
	that $\area(t)<\pi - 6\arcot(\sqrt{3}\cosh(\tfrac{1}{4}\sys))$; this is Lemma~\ref{lem:triangleareabound} in Appendix~\ref{sec:appendixtriangulations}.
	Because $\M$ has area $4\pi(g-1)$, it follows that 
	\begin{equation}\label{eq:lowerboundfaces}
		k_2>\dfrac{4\pi(g-1)}{\pi - 6\arcot(\sqrt{3}\cosh(\tfrac{1}{2}\sys))}.
	\end{equation}
	Combining~\eqref{eq:eulertriangle} and~\eqref{eq:lowerboundfaces} yields the result. 
\end{proof}

To show that our lower and upper bounds are meaningful,
we consider the asymptotics of these bounds for a
family of surfaces of which the systoles are 1.\ contained in a compact subset of $\R_{>0}$,
2.\ arbitrarily close to zero, or 3.\ arbitrarily large. 
\begin{enumerate}
	\item If the systoles of the family of surfaces are contained in a compact subset of $\R_{>0}$, which is the case for the generalized Bolza surfaces, then the upper bound is of order $O(g)$ and the lower bound of order $\Omega(g)$. Hence, a minimum dummy point set has cardinality $\Theta(g)$.
	\item If $\sys(\M)\rightarrow 0$, then $\cosh(\tfrac{1}{8}\sys(\M))-1\sim \tfrac{1}{2}(\tfrac{1}{8}\sys(\M))^2$, so
	our upper bound is of order $g\cdot O(\sys(\M)^{-2})$. In a similar way, it can be shown that 
	$$ \left( \dfrac{\pi}{\pi-6\arcot(\sqrt{3}\cosh(\tfrac{1}{4}\sys(M)))}-1\right)\sim \dfrac{64\pi}{3\sqrt{3}\sys(\M)^2},$$
	which means that our lower bound is of order $g\cdot\Omega(\sys(\M)^{-2})$. It follows that in this case a minimum dummy point set has cardinality $g\cdot\Theta(\sys(\M)^{-2})$.
	\item Finally, consider the case when
	$\sys(\M)\rightarrow\infty$ when $g\rightarrow
	\infty$. Since $\sys(\M)\leq 2\log(4g-2)$ for all hyperbolic
	surfaces of genus $g$ (see Equation~\eqref{eq:systoleinequality} in
	Section~\ref{sec:introhyperbolicsurfaces}), we only consider the case where
	$\sys(\M)\sim C\log g$ for some $C$ with $0<C\leq 2$. In fact, there exist families
	of surfaces for which $\sys(\M)>\tfrac{4}{3}\log g-c$ for some constant $c\in\R$ for infinitely many genera $g$.
	See \cite[page 45]{bs-pmrslg-94} and \cite{katz2007}.
	In this case, we can use $\cosh x \sim \tfrac{1}{2}e^x$ to deduce
	that our upper bound reduces to an expression of order
	$O(g^{1-\tfrac{1}{8}C})$. Similarly, by considering the Taylor
	expansion of the coefficient in the lower bound we see that
	the lower bound has cardinality $\Omega (g^{1-\tfrac{1}{4}C})$. Of our three cases, this is the only case in which there is a gap between the stated upper and lower bound.
\end{enumerate}

\begin{remark}\label{rem:condition}
	Note that the validity condition~\eqref{condition} is stronger
	than just requiring that the Delaunay triangulation of
	$\dummyq$ be a simplicial complex. This can also be seen in
	the following way. It has been shown that every hyperbolic
	surface of genus $g$ has a simplicial Delaunay triangulation
	with at most $151g$ vertices~\cite{ebbens2020}. In particular,
	this upper bound does not depend on $\sys(\M)$. Since the
	coefficient of $g-1$ in the lower bound given in
	Proposition~\ref{thm:minimumnumberofpoints} goes to infinity
	as $\sys(\M)$ goes to zero, the minimal number of vertices of
	a set $\dummyq$ satisfying the validity
	condition is strictly larger than the number of
	vertices needed for a simplicial Delaunay triangulation of a
	hyperbolic surface with sufficiently small systole. Moreover,
	in the same work it was shown that for infinitely many genera
	$g$ there exists a hyperbolic surface $\M$ of genus $g$ which
	has a simplicial Delaunay triangulation with
	$\Theta(\sqrt{g})$ vertices. Hence, the number of vertices
	needed for a simplicial Delaunay triangulation and a point set
	satisfying the validity
	condition differs asymptotically as well.  
\end{remark}

%%%%%%%%%%%%%%%%%%%%%%%%%%%%%%%%%%%%%%%%%%%%%%%%%%%%%%%%%%%%%

\section{Proof of Theorem~\ref{thm:systolevalue}: Systole of generalized Bolza surfaces}
\label{sec:symmetric-hyperbolic-surfaces}

In the previous section we have recalled the validity
condition~\eqref{condition}, allowing us to define the Delaunay
triangulation $\dts{\S}$. To be able to verify that this condition
holds, we must know the value of the systole for the given hyperbolic
surface.
This section is devoted to proving Theorem~\ref{thm:systolevalue} stated in the introduction, which
gives the value of the systole for the generalized Bolza surfaces $\Mg$ defined in
Section~\ref{sec:generalized}.

As a preparation for the proof we show in Section~\ref{sec:representationofclosedgeodesic}
how to represent a simple closed geodesic $\gamma$ on $\Mg$ by a sequence $\gammah$ of pairwise
disjoint hyperbolic line segments between sides of the fundamental domain~$D_g$. The length of
$\gamma$ is equal to the sum of the lengths of the line segments in $\gammah$.

The proof consists of two parts. In Section~\ref{sec:upperboundsystole} we show that
$\sysg\leq \varsigma_g$ by constructing a simple (non-contractible) closed geodesic of length
$\varsigma_g$.
In Section~\ref{sec:caseanalysis} we show that $\hlen{\gamma}\geq\varsigma_g$ for all closed
geodesics $\gamma$ by a case analysis based on the line segments contained in the sequence
$\gammah$ representing $\gamma$. This shows that $\sysg\geq \varsigma_g$.

\subsection{Representation of a simple closed geodesic by a sequence of segments}
\label{sec:representationofclosedgeodesic}
Consider a simple closed geodesic $\gamma$ on the generalized Bolza surface $\Mg$. Because $D_g$ is
compact, there is a finite number, say $m$, of pairwise disjoint hyperbolic lines intersecting $D_g$
in the preimage $\proj^{-1}(\gamma)$ of $\gamma$. See the leftmost panel in
Figure~\ref{fig:findingsequenceofsegments}.
These hyperbolic lines are the axes of conjugated elements of $\Gamma_g$.
Therefore, the intersection of $\proj^{-1}(\gamma)$ with $D_g$ consists of $m$ pairwise disjoint
hyperbolic line segments between the sides of $D_g$, the union of which we denote by $\gammah$.
See the rightmost panel in Figure~\ref{fig:findingsequenceofsegments}.
These line segments are oriented and their orientations are compatible with the orientation of
$\gamma$. In particular, every line segment has a starting point and an endpoint. Since $D_g$ is a fundamental domain for $\Gamma_g$, the $\proj$-images of these line segments form
a covering of the closed geodesic $\gamma$ by $m$ closed subsegments with pairwise disjoint
interiors. In other words, these projected segments lie side-by-side on $\gamma$, so they form a
(cyclically) ordered sequence.
This cyclic order lifts to an order $\gammah_1, \ldots, \gammah_m$ of the $m$ segments in $D_g$,
which together represent the simple closed geodesic $\gamma$. More precisely:
\begin{definition}
	\label{def:representation-closed-geodesic}
	An oriented simple closed geodesic $\gamma$ on $\Mg$ is \emph{represented by a sequence of
		oriented geodesic segments} $\gammah_1, \ldots, \gammah_m$ in $D_g$ if
	(i) the starting point and endpoint of each segment lie on different sides of $\partial D_g$,
	and (ii) the projections $\proj(\gammah_1), \ldots, \proj(\gammah_m)$ are oriented closed
	subsegments of $\gamma$ that cover $\gamma$, have pairwise disjoint interiors, and lie
	side-by-side on $\gamma$ in the indicated order.
\end{definition}
\begin{figure}[ht]
	\centering
	\includegraphics[width=.45\textwidth,valign=t]{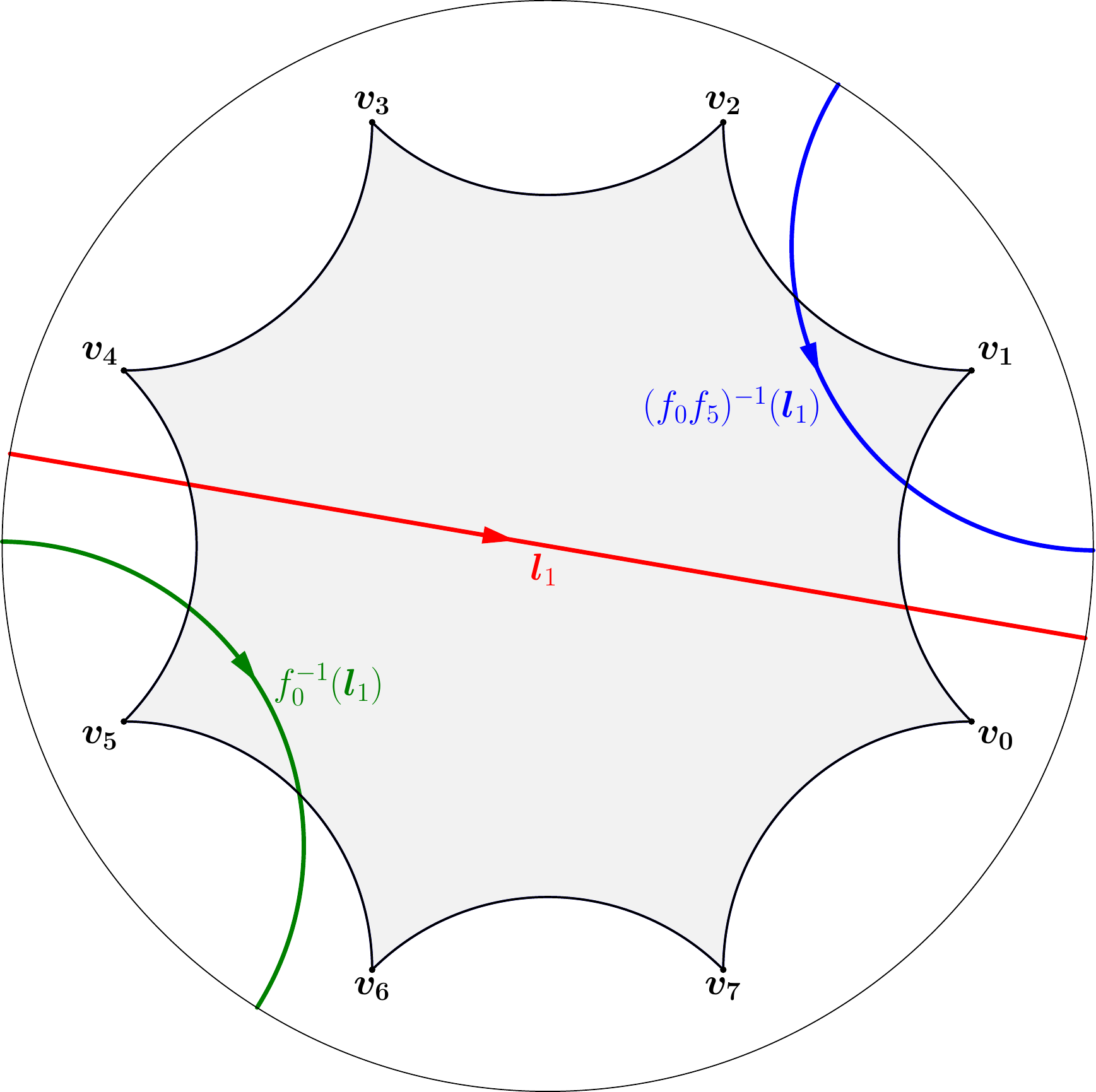}~~%
	\includegraphics[width=.45\textwidth,valign=t]{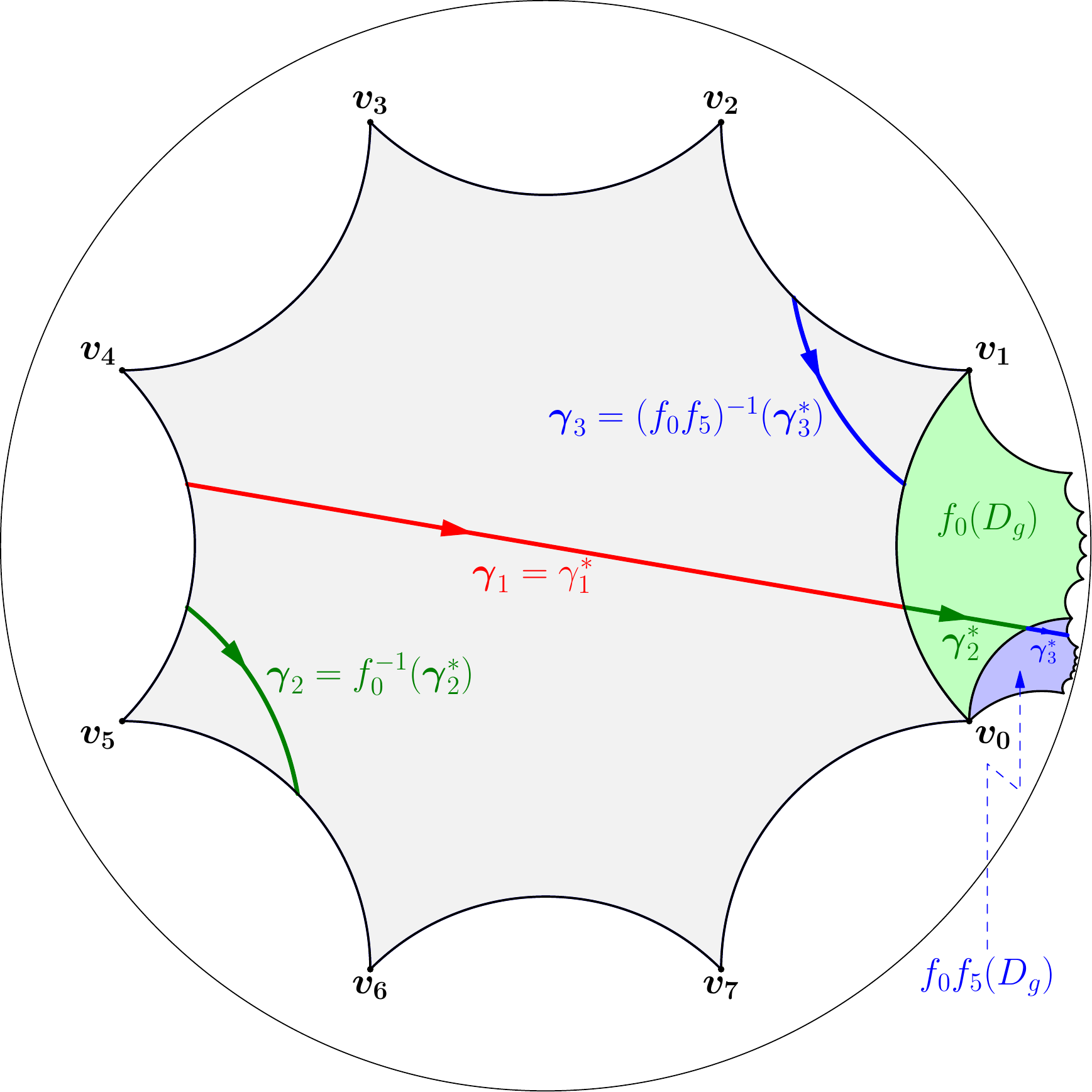}
	\caption{	\label{fig:findingsequenceofsegments}
		Left: Connected components of the preimage of an oriented simple closed geodesic
		$\gamma$ on the Bolza surface intersecting the fundamental octagon $D_2$.
		The geodesic is covered by $f = f_0f_5f_0$ (and all its conjugates in the Bolza group),
		which has axis $\lh_1$.
		\\
		Right: The cyclic sequence of geodesic segments $\gammah_1, \gammah_2, \gammah_3$
		in $D_g$ represents $\gamma$. The segments $\gammah_1^\ast$, $\gammah_2^\ast$ and
		$\gammah_3^\ast$ are the successive intersections of the axis of $f$ with $F_0(D_g)$,
		$F_1(D_g)$ and $F_2(D_g)$, where $F_0 = \eg$, $F_1 = f_0$ and $F_2 = f_0f_5$.
		The endpoint of $\gammah_3^\ast$ is $F_3(\ph_0^\ast)$, 
		where $F_3 = f$ and $\ph_0^\ast$ is the starting point of $\gammah^\ast_1$.
		The endpoint of $\gammah_k$ is paired with the starting point of $\gammah_{k+1}$ by the 
		side-pairing $F_k^{-1}F_{k-1}$.
		In this example $F_1^{-1}F_0 = f_0^{-1} = f_4$, $F_2^{-1}F_1 = f_5^{-1} = f_1$, and
		$F_3^{-1}F_2 = f_0^{-1} = f_4$. 
	}
\end{figure}
We now discuss in more detail how such a sequence is obtained from a hyperbolic isometry
the axis of which intersects $D_g$ and projects onto the simple closed geodesic.
Let $\lh_1$ be an arbitrary oriented geodesic in the set of $m$ connected components of
$\proj^{-1}(\gamma)$ that intersect $D_g$. The oriented segment $\gammah_1$ is the intersection
$\lh_1 \cap D_g$.
Let $f \in \Gamma_g$ be the hyperbolic isometry that covers $\gamma$ and has axis $\lh_1$.
More precisely, if $\ph_0^\ast$ is the starting point of $\gammah_1$, then the segment
$[\ph_0^\ast,f(\ph_0^\ast)]$ projects onto $\gamma$, and $\proj$ is injective on the interior of this
segment.

Let $F_0(D_g), F_1(D_g),\ldots,F_{m-1}(D_g)$, be the \textit{sequence} of
successive Dirichlet domains intersected by the segment $[\ph_0^\ast,f(\ph_0^\ast)]$.
Here $F_0 = \eg$ and $F_0,F_1,\ldots,F_{m-1}$ are distinct elements of $\Gamma_g$.
This sequence consists of $m$ regions, since
$F_0^{-1}(\lh_1), F_1^{-1}(\lh_1),\ldots,F_{m-1}^{-1}(\lh_1)$ are the (pairwise disjoint) geodesics
in $\proj^{-1}(\gamma)$ that intersect $D_g$.
This implies that the segment $[\ph_0^\ast,f(\ph_0^\ast)]$ is covered by the \textit{sequence} of
closed segments $\gammah_1^\ast, \ldots, \gammah_m^\ast$ in which $\lh_1$ intersects these $m$
regions,
i.e., $\gammah_k^\ast = F_{k-1}(D_g) \cap \lh_1$, for $k = 1,\ldots, m$.
(Note that $\gammah_1^\ast = \gammah_1$.)
The segments $\gammah_k = F_{k-1}^{-1}(\gammah_k^\ast)$, $k = 1,\ldots,m$, lie in $D_g$ and project
onto the same subsegment of $\gamma$ as $\gammah_k^\ast$.
In other words, $\proj(\gammah_1),\ldots,\proj(\gammah_m)$ lie side-by-side on the closed
geodesic $\gamma$ and cover $\gamma$. Therefore, the simple closed geodesic $\gamma$ is
represented by the sequence $\gammah_1,\ldots,\gammah_m$.

It is convenient to consider $f(\ph_0^\ast)$ as the starting point of the segment
$\lh_1 \cap f(D_g)$, which we denote by $\gammah_{m+1}^\ast$.
Taking $F_m = f$, we see that $\gammah_{m+1}^\ast = \lh_1 \cap F_m(D_g)$.
Extending our earlier definition $\gammah_{k} = F_{k-1}^{-1}(\gammah_{k}^\ast)$ to $k = m+1$,
we see that $\gammah_{m+1}$ is the subsegment of $\proj^{-1}(\gamma) \cap D_g$
with starting point $F_m^{-1}(f(\ph_0^\ast)) = \ph_0^\ast$, so $\gammah_{m+1} = \gammah_1$.

Finally, we show that the endpoint of $\gammah_{k}$ is mapped to the starting point of
$\gammah_{k+1}$ by a side-pairing transformation of $D_g$, for $k = 1,\ldots m$.
Since $F_{k-1}(D_g)\cap F_k(D_g)$ is a side of $F_k(D_g)$, for $k = 1,\ldots,m$,
the intersection $F_k^{-1}F_{k-1}(D_g) \cap D_g$ is a side of $D_g$, say $\sh_{j_k}$.
Then $F_k^{-1}F_{k-1} = f_{j_k}$, since $\sh_{j_k} = f_{j_k}(D_g)\cap D_g$.
Let $\gammah_k^\ast = [\ph_{k-1}^\ast,\ph_k^\ast]$, then
$\gammah_k = [F_{k-1}^{-1}(\ph_{k-1}^\ast),F_{k-1}^{-1}(\ph_k^\ast)]$.
Therefore, $f_{j_k}$ maps the endpoint $F_{k-1}^{-1}(\ph_k^\ast)$ of $\gammah_{k}$
to the starting point $F_{k}^{-1}(\ph_k^\ast)$ of $\gammah_{k+1}$, since
$f_{j_k}F_{k-1}^{-1} = F_{k}^{-1}$.
See the rightmost panel in Figure~\ref{fig:findingsequenceofsegments}.

\subsection{Upper bound for the systole}\label{sec:upperboundsystole}

To show that $\sysg\leq \varsigma_g$ it is sufficient to prove the following lemma.

\begin{lemma}\label{lem:systole}
	There is a simple closed geodesic on $\M_g$ of length $\varsigma_g$.
\end{lemma}

\begin{proof}
	The axis of the hyperbolic translation $f_{2g+1}f_0$ projects onto a simple closed geodesic
	$\gamma$ on $\M_g$ with length equal to the translation length $\len{f_{2g+1}f_0}$.
	See Section~\ref{sec:introhyperbolicsurfaces}.
	Since $f_{2g+1}f_0$ is represented by the matrix $A_{2g+1}A_0$, with $A_j$ given
	by~\eqref{eq:generator-matrix}, we see that
	\begin{equation*}
		\cosh\tfrac{1}{2}\len{f_{2g+1}f_0} = \tfrac{1}{2}\,\abs{\tr(A_{2g+1}A_0)} = 1 + \cos(\tfrac{\pi}{2g}).
	\end{equation*}
	Since $1 + \cos(\tfrac{\pi}{2g}) = \cosh\tfrac{1}{2}\varsigma_g$ we conclude that $\gamma$ has length
	$\varsigma_g$.
\end{proof}

\begin{remark}
	Two connected components of the pre-image $\proj^{-1}(\gamma)$ of
	the simple closed geodesic $\gamma$, appearing in the proof of Lemma~\ref{lem:systole},
	intersect the fundamental polygon $D_g$:
	the axis $\lh_1$ of $f = f_{2g+1}f_0$, and the geodesic $\lh_2 = f_{2g+1}^{-1}(\lh_1)$, which is
	the axis of $f_0f_{2g+1}$.
	The geodesic $\gamma$ is represented by the segments $\gammah_1 = \lh_1
	\cap D_g$ and $\gammah_2 = \lh_2 \cap D_g$.
	The first segment connects the midpoint $\mh_{2g}$ of $\sh_{2g}$
	and the midpoint $\mh_{2g+1}$ of $\sh_{2g+1}$, whereas the second segment connects the midpoints
	of $\sh_0$ and $\sh_1$.
	See Figure~\ref{fig:lower-bound-2}.
	
	This can be seen as follows. 
	Since $f = f_{2g+1}f_{2g}^{-1}$ and the axes of $f_{2g}$ and $f_{2g+1}$ intersect at the origin
	$O$, (the proof of) Theorem 7.38.6 of~\cite{Beardon1983} implies
	that the axis of $f$ passes through the midpoint of the segment $[O,f_{2g}(O)]$ and the
	midpoint of $[O,f_{2g+1}(O)]$.
	But these midpoints coincide with $\mh_{2g}$ and $\mh_{2g+1}$, respectively, so
	$[\mh_{2g},\mh_{2g+1}] = \gammah_1$.
	This theorem also implies that the length of the latter segment 
	is half the translation length of $f$, i.e., $\tfrac{1}{2}\varsigma_g$.
	A similar argument shows that the length of $\gammah_2$ is $\tfrac{1}{2}\varsigma_g$.
\end{remark}

\begin{figure}[ht]
	\centering
	\includegraphics[width=.45\textwidth,valign=t]{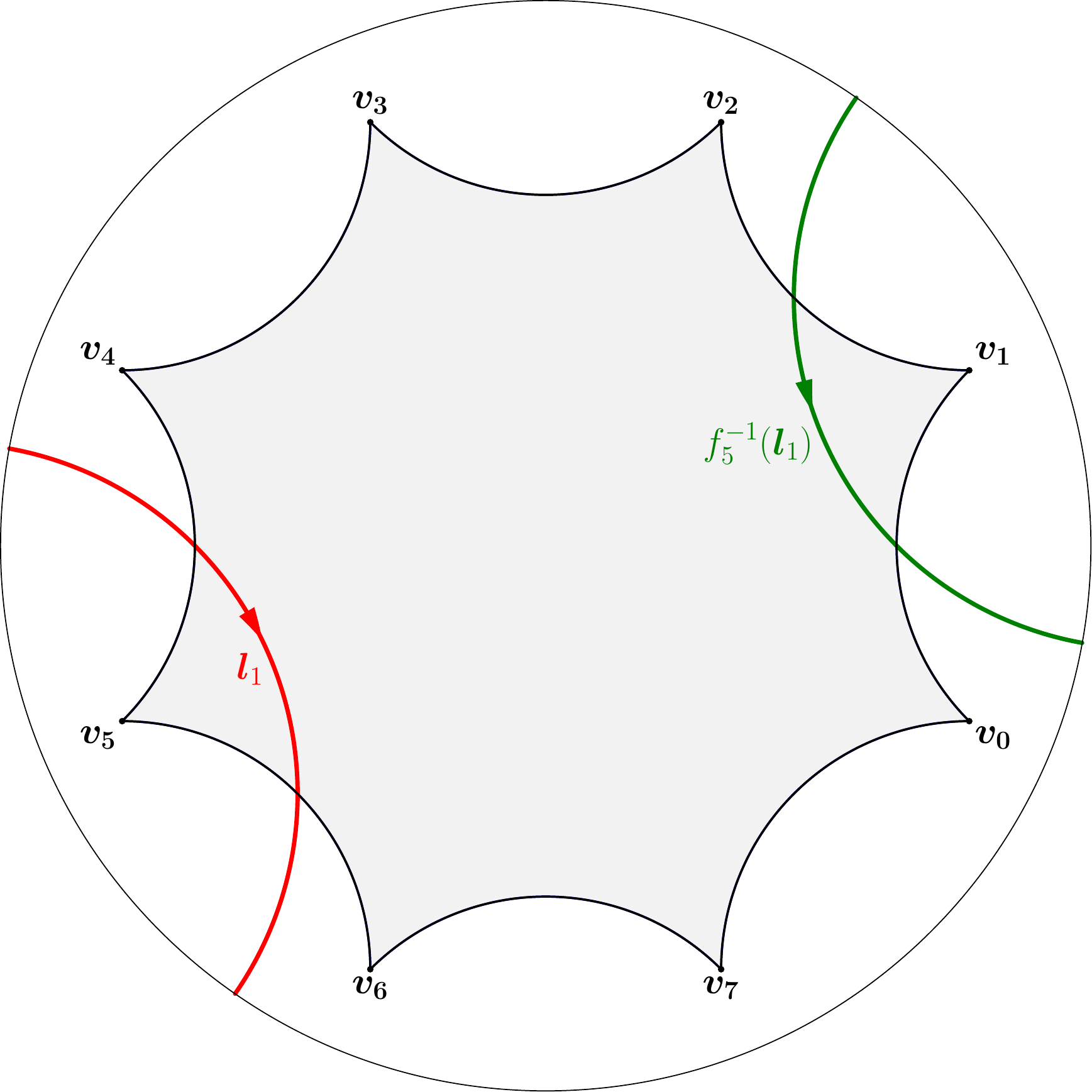}~~%
	\includegraphics[width=.45\textwidth,valign=t]{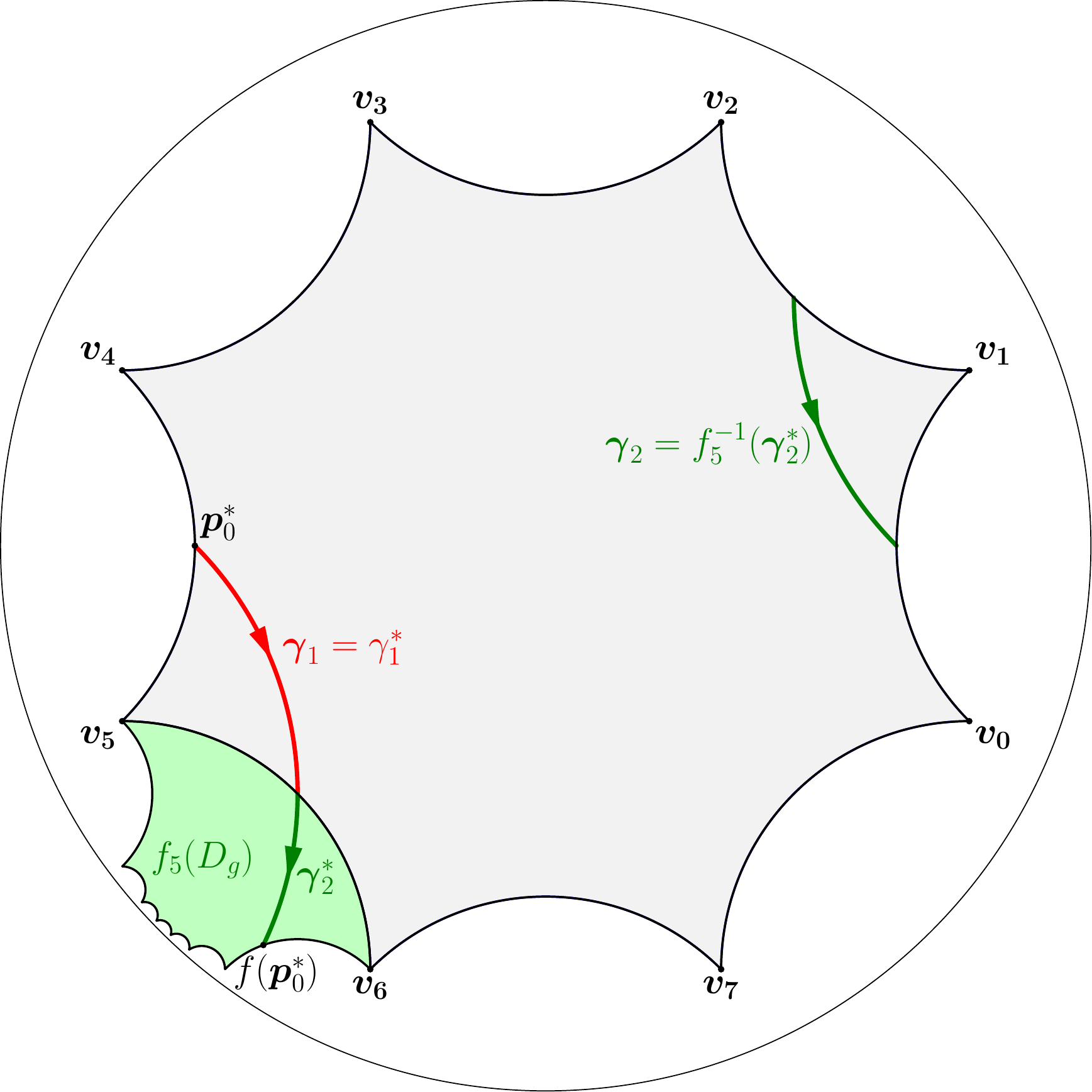}
	\caption{	\label{fig:lower-bound-2}
		Two geodesics in the pre-image of $\gamma$ intersect the fundamental polygon $D_g$ (left)
		The intersections are the segments $\gammah_1$ and $\gammah_2$, which represent $\gamma$
		(right). The figure illustrates the situation for the Bolza surface ($g=2$).}
\end{figure}

\subsection{Lower bound for the systole}\label{sec:caseanalysis}
We now prove that the length of every simple closed geodesic of $\Mg$ is at least $\varsigma_g$,
or, equivalently, that the total length of the segments representing such a geodesic is at least
$\varsigma_g$.
To this end we consider different types of closed geodesics based on which ``kind'' of segments are
contained in the sequence. We say that an oriented hyperbolic line segment between two sides of
$D_g$ is a \emph{$k$-segment}, $1 \leq k \leq 4g-1$, if its starting point and endpoint are
contained in $\sh_j$ and $\sh_{j+k}$, respectively, for some $j$ with $0\leq j\leq 4g-1$,
where indices are counted modulo $4g$. Furthermore, we
say that the segment is \emph{$k$-separated} or \emph{has separation $k$}, $1\leq k\leq 2g$,
if either the segment itself or the segment with the opposite orientation is a $k$-segment.
Equivalently, a $k$-separated segment is either a $k$-segment or a $(4g-k)$-segment.
For example, both segments in Figure~\ref{fig:lower-bound-2} - Right are
1-separated, but $\gammah_1$ is a $1$-segment while $\gammah_2$ is a $7$-segment.

In the derivation of the lower bound for the systole we will use the following lemma. This lemma
will be used in the proof of Proposition~\ref{prop:inclusionproperty} in
Section~\ref{sec:representation-of-delaunay-triangulations} as well.

\begin{lemma}\label{lem:separationlowerbound}
	For geodesic segments between the sides of $D_g$ the following properties hold:
	\begin{enumerate}\cramped
		\item The length of a segment that has separation at least 4 is at least $\varsigma_g$.\label{item:sep4}
		\item The length of a segment that has separation at least 2 is at least $\tfrac{1}{2}\varsigma_g$.\label{item:sep2}
		\item Every pair of consecutive 1-separated segments consists of exactly one 1-segment and one $(4g-1)$-segment.\label{item:sep1alternate} 
		\item The length of two consecutive 1-separated segments is at least $\tfrac{1}{2}\varsigma_g$.\label{item:sep11}
		\item A sequence of segments consisting of precisely two 1-separated segments has length $\varsigma_g$.\label{item:2sep1} 
	\end{enumerate}
\end{lemma}

The proof can be found in Appendix~\ref{sec:appendixsymmetrichyperbolicsurfaces}.
The lower bound for the systole follows from the following result.

\begin{lemma}\label{lem:lowerboundsystole}
	Every closed geodesic on $\Mg$ has length at least $\varsigma_g$.
\end{lemma}

\begin{proof}
	It is sufficient to show that every sequence of segments representing a closed geodesic on $\Mg$ has length at least $\varsigma_g$.	Let $\gammah$ be a sequence of segments. We distinguish between the following four types:
	
	\begin{enumerate}\cramped
		\item $\gammah$ contains at least one segment that has separation at least 4,
		\item $\gammah$ contains at least two segments that have separation 2 or 3 and all other segments are $1$-separated,
		\item $\gammah$ contains exactly one segment that has separation 2 or 3 and all other segments are $1$-separated,
		\item all segments of $\gammah$ are 1-separated.
	\end{enumerate}
	
	It is straightforward to check that every sequence of segments is of precisely one type.
	
	First, suppose that $\gammah$ is of Type 1 or 2. Then, it follows directly from Part~\ref{item:sep4} and~\ref{item:sep2} of Lemma~\ref{lem:separationlowerbound} that $\hlen{\gammah}\geq\varsigma_g$.
	
	Second, suppose that $\gammah$ is of Type 3. It is not possible to form a closed geodesic with a segment of separation 2 or 3 and just one segment of separation 1, so we can assume that there are at least two 1-separated segments. In the cyclic ordering of the segments, these 1-separated segments are consecutive, so it follows from Part~\ref{item:sep11} of Lemma~\ref{lem:separationlowerbound} that their combined length is at least $\tfrac{1}{2}\varsigma_g$. By Part~\ref{item:sep2} of Lemma~\ref{lem:separationlowerbound} the length of the segment of separation 2 or 3 is at least $\tfrac{1}{2}\varsigma_g$ as well, so we conclude that $\hlen{\gammah}\geq\varsigma_g$.
	
	Finally, suppose that $\gammah$ is of Type 4. By Part~\ref{item:sep1alternate} of Lemma~\ref{lem:separationlowerbound} every 1-segment is followed by a $(4g-1)$-segment and reversely, so in particular the number of 1-segments and $(4g-1)$-segments is identical. Therefore, the number of 1-separated segments in $\gammah$ is even (and at least two). If the number of 1-separated segments is exactly two, then $\hlen{\gammah}=\varsigma_g$ by Part~\ref{item:2sep1} of Lemma~\ref{lem:separationlowerbound}. If the number of 1-separated segments is at least four, then $\hlen{\gammah}\geq\varsigma_g$, since every pair of consecutive 1-separated segments has combined length at least $\tfrac{1}{2}\varsigma_g$ by Part~\ref{item:sep11} of Lemma~\ref{lem:separationlowerbound}.
	
	This finishes the proof. 
\end{proof}

%%%%%%%%%%%%%%%%%%%%%%%%%%%%%%%%%%%%%%%%%%%%%%%%%%%%%%%%%%%%%

\section{Computation of dummy points}
\label{sec:combinatoricsofsmallDT}

In this section we present two algorithms for constructing a dummy point set ${\dummy}$ satisfying the validity condition~\eqref{condition} for $\Mg$ and give the growth rate of the cardinality of ${\dummy}$ as a function of $g$.

Both algorithms use the set $\W$ of the so-called Weierstrass points of $\Mg$. In the fundamental domain $D_g$, the Weierstrass points are represented by the origin, the vertices and the midpoints of the sides. In the original domain $\Do$, where there is only one point of each orbit under the action of $\Gg$, this reduces to $2g+2$ points: the origin, the midpoint of each of the $2g$ closed sides, and the vertex $\vh_0$. Some special properties of Weierstrass points are known in Riemann surface theory \cite{farkas1992}, however we will not use them in this paper.

Each of the algorithms has its own advantages and drawbacks. The \emph{refinement algorithm} (Section~\ref{sec:simplealgorithm}) yields a point set with optimal asymptotic cardinality $\Theta(g)$ (Proposition~\ref{thm:minimumnumberofpoints}). The idea is borrowed from the well-known Delaunay refinement algorithm for mesh generation~\cite{ruppert1995}. 
The \emph{symmetric algorithm} (Section~\ref{sec:symmetricalgo}) uses the Delaunay refinement algorithm as well. However, instead of inserting one point in each iteration, we insert its images by all rotations around the origin by angle $k\pi/2g$ for $k=1,\ldots,4g$. In this way, we obtain a dummy point set that preserves the symmetries of $D_g$, at the cost of increasing the asymptotic cardinality to $\Theta(g\log g)$. 

\medskip

Let us now elaborate on the refinement algorithm. The set $\dummy$ is initialized as $\W$ and the triangulation as $\dtsg{\W}$. Then, all non-admissible triangles in $\dth{\proj^{-1}(\dummy)}$ are removed by inserting the projection onto $\Mg$ of their circumcenter, while updating the set $\dummy$ of vertices of the triangulation. The following proposition shows that $\dth{\proj^{-1}(\dummy)\cap D_{\N}}$ contains at least one representative of each face of $\dth{\proj^{-1}(\dummy)}$, thus providing the refinement algorithm with a finite input.

\begin{proposition}
	For any finite set of points $\dummy$ on $\Mg$ containing $\W$, each face in $\dth{\proj^{-1}(\dummy)}$ with at least one vertex in $\Do$ is contained in $D_{\N}$.
	\label{prop:dummy-points-inclusion}
\end{proposition}
The proof is given in Appendix~\ref{sec:appendixstructuredalgorithm}. 

The set $\proj^{-1}(\dummy)\cap D_{\N}$is obtained as follows: we first consider the set of canonical representatives (as defined in Section~\ref{sec:generalized}) of the points of $\dummy$, which is $\proj^{-1}(\dummy)\cap\Do$. Then, we obtain $\proj^{-1}(\dummy)\cap D_{\N}$ by computing the images of $\proj^{-1}(\dummy)\cap\Do$ under the elements in $\N$. In other words, $\proj^{-1}(\dummy)\cap D_{\N}$ can be computed as $\h{\dummyq}_{\N}=\{ f(\proj^{-1}(\dummy)\cap \Do), f\in\N \}$.

Apart from the two algorithms, detailed below, we have also looked at the \emph{structured algorithm}~\cite{iordanov-tel-02072155}, which can be found in Appendix~\ref{sec:appendixstructuredalgorithmdescription}. Its approach is fundamentally different from the refinement and symmetric algorithms: the dummy point set and the corresponding Delaunay triangulation are exactly described. As in the symmetric algorithm, the resulting dummy point set preserves the symmetries of $D_g$ and is of order $\Theta(g\log g)$. 

\subsection{Refinement algorithm}\label{sec:simplealgorithm}
Following the refinement strategy introduced above and using Proposition~\ref{prop:dummy-points-inclusion}, we insert the circumcenter of each triangle in $\dth{\h{\dummyq}_{\N}}$ having a non-empty intersection with the domain $\Do$ and whose circumradius is at least $\tfrac{1}{2}\sysg$ (see Algorithm~\ref{simplealgo}).
Figure~\ref{fig:simplealgorithm} illustrates the computation of $\dth{\dummyq_{\Nnude_3}}$.

\IncMargin{1em}	
\begin{algorithm}[H]
	\SetKwInOut{Input}{Input}
	\SetKwInOut{Output}{Output}
	\DontPrintSemicolon
	\Input{hyperbolic surface $\Mg$}
	\Output{finite point set ${\dummy}\subset\Mg$ such that $\diam{\dummy}<\tfrac{1}{2}\sysg$}
	\BlankLine
	Initialize: let ${\dummy}$ be the set $\W$ of Weierstrass points of $\Mg$.\;
	Compute $\dth{\h{\dummyq}_{\N}}$.\; 
	\While{there exists a triangle $\Delta$ in $\dth{\h{\dummyq}_{\N}}$ with circumdiameter at least $\tfrac{1}{2}\sysg$ and $\Delta\cap D_g\neq\emptyset$}{
		Add to ${\dummy}$ the projection onto $\Mg$ of the circumcenter of $\Delta$\;		Update $\dth{\h{\dummyq}_{\N}}$		
	}
	\caption{Refinement algorithm}\label{simplealgo}
\end{algorithm}

\begin{figure}
	\centering
	\begin{tabular}{cc}
		\includegraphics[width=.45\textwidth]{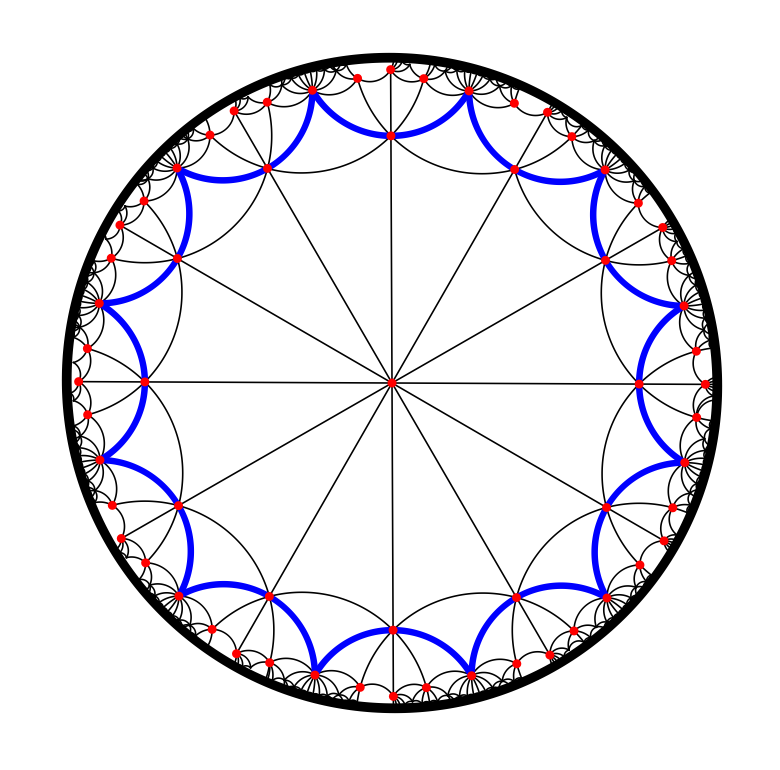} &
		\includegraphics[width=.45\textwidth]{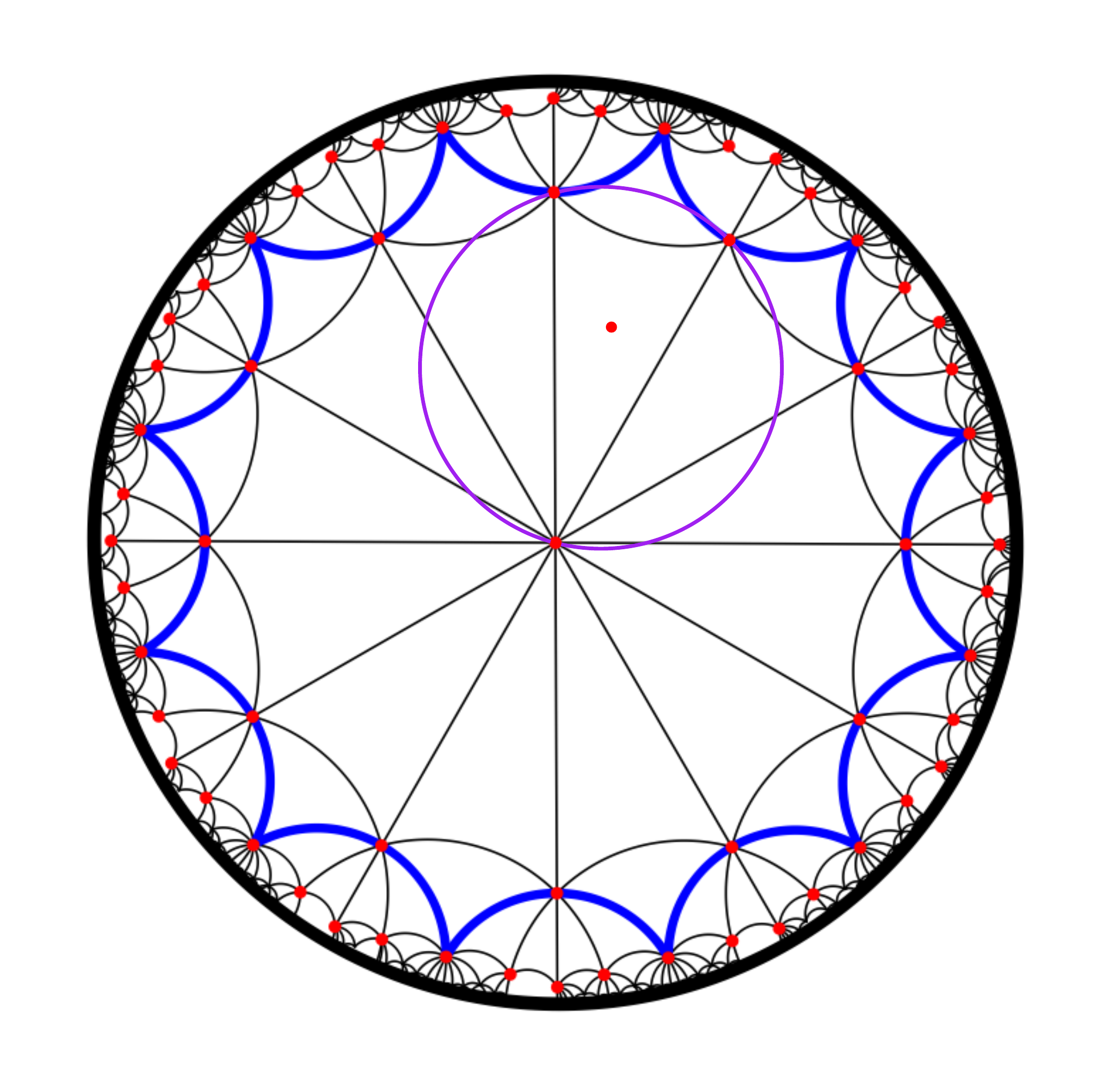} \\
		After initialization & First insertion \\
		\includegraphics[width=.45\textwidth]{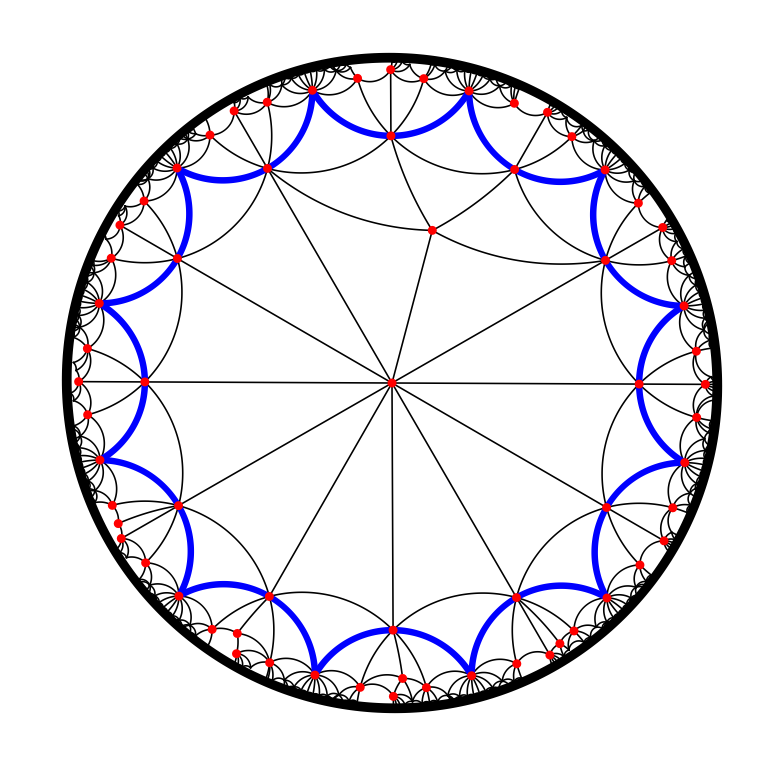} &
		\includegraphics[width=.45\textwidth]{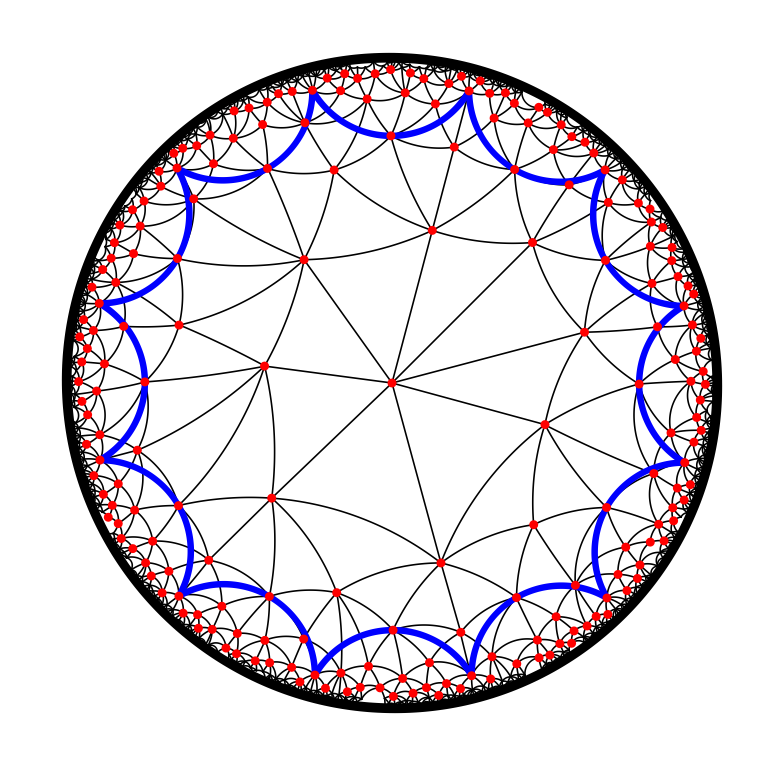} \\
		After first insertion & After last insertion
	\end{tabular}
	\caption{Several steps in the refinement algorithm (genus~3)}
	\label{fig:simplealgorithm}
\end{figure}

We can now show that the cardinality of the resulting dummy point set is linear in the genus~$g$.

\begin{theorem}\label{thm:simplealgorithm}
	The refinement algorithm terminates and the resulting dummy point set ${\dummy}$ satisfies the validity condition~\eqref{condition}. The cardinality $|{\dummy}|$ is bounded as follows
	$$ 5.699(g-1)<|\dummy|<27.061(g-1).$$
\end{theorem}
\begin{proof}
	We first prove that the hyperbolic distance between two distinct points of $\dummy$ is greater than $\tfrac{1}{4}\sysg$. The distance between any pair of Weierstrass points is larger than $\tfrac{1}{4}\sysg$ (see Lemma~\ref{lem:distanceweierstrasspoints} in Appendix~\ref{sec:appendixstructuredalgorithm}). 
	
	Furthermore, every point added after the initialization is the projection of the circumcenter of an empty disk in $\mD$ of radius at least $\tfrac{1}{4}\sysg$, so the distance from the added point to any other point in ${\dummy}$ is at least $\tfrac{1}{4}\sysg$. For arbitrary $p\in{\dummy}$, consider the disk $D_p$ in $\Mg$ of radius $\tfrac{1}{8}\sysg$ centered at $p$, i.e., the set of points in $\Mg$ at distance at most $\tfrac{1}{8}\sysg$ from $p$. Every disk of radius at most $\tfrac{1}{2}\sysg$ is embedded in $\Mg$, so in particular $D_p$ is an embedded disk. Because the distance between any pair of points of ${\dummy}$ is at least $\tfrac{1}{4}\sysg$, the disks $D_p$ and $D_q$ of radius $\tfrac{1}{8}\sysg$ centered at $p$ and $q$, respectively, are disjoint for every distinct $p,q\in{\dummy}$. For fixed $g$, the area of such disks is fixed, as is the area of $\Mg$, so only a finite number of points can be added. Hence, the algorithm terminates.
	
	Observe that the algorithm terminates if and only if the while loop ends, i.e.\ ${\dummy}$ satisfies the validity condition. 
	
	Finally, we bound for the cardinality of ${\dummy}$. From the above argument we know that the cardinality of ${\dummy}$ is bounded above by the number of disjoint disks $D$ of radius $\tfrac{1}{8}\sysg$ that fit inside $\Mg$. Hence,
	$$\abs{{\dummy}}\leq \dfrac{\area(\Mg)}{\area(D)}= \dfrac{4\pi(g-1)}{2\pi\left(\cosh\left(\tfrac{1}{8}\sysg\right)-1\right)}=\dfrac{2(g-1)}{\cosh\left(\tfrac{1}{8}\sysg\right)-1}.$$
	Proposition~\ref{thm:minimumnumberofpoints} gives a lower bound.
	The coefficients of $g-1$ in these upper and lower bounds decrease as a function of $g$, so the announced bounds can be obtained by plugging in the value of $\sysg$ (see Theorem~\ref{thm:systolevalue}) for $g\rightarrow\infty$ and $g=2$ respectively. This finishes the proof.
\end{proof}

\subsection{Symmetric algorithm}\label{sec:symmetricalgo}
This algorithm is similar to the refinement algorithm. However, instead of adding one point at every step in the while loop, it uses the $4g$-fold symmetry of the fundamental polygon $D_g$ to add $4g$ points at every step (see Algorithm~\ref{symmetricalgo}). Figure~\ref{fig:symmetricalgorithm} illustrates the computation of $\dth{\dummyq_{\Nnude_3}}$. 

\IncMargin{1em} 
\begin{algorithm}[H]
	\SetKwInOut{Input}{Input}
	\SetKwInOut{Output}{Output}
	\DontPrintSemicolon
	\Input{hyperbolic surface $\Mg$}
	\Output{finite point set ${\dummy}\subset\Mg$ such that $\diam{\dummy}<\tfrac{1}{2}\sysg$}
	\BlankLine
	Initialize: let ${\dummy}$ be the set $\W$ of Weierstrass points of $\Mg$.\;
	Compute $\dth{\h{\dummyq}_{\N}}$.\; 
	\While{there exists a triangle $\Delta$ in $\dth{\h{\dummyq}_{\N}}$ with circumdiameter at least $\tfrac{1}{2}\sysg$}{
		\For{$k=0,\ldots,4g-1$}{
			Let $\ph_k$ be the circumcenter of $\Delta$ rotated around the origin by angle $\tfrac{k\pi}{2g}$.\;
			Add $\proj(\ph_k)$ to ${\dummy}$.}
		Update $\dth{\h{\dummyq}_{\N}}$.		
	}
	\caption{Symmetric algorithm}\label{symmetricalgo}
\end{algorithm}\medskip

\begin{figure}
	\centering
	\begin{tabular}{cc}
		\includegraphics[width=.45\textwidth]{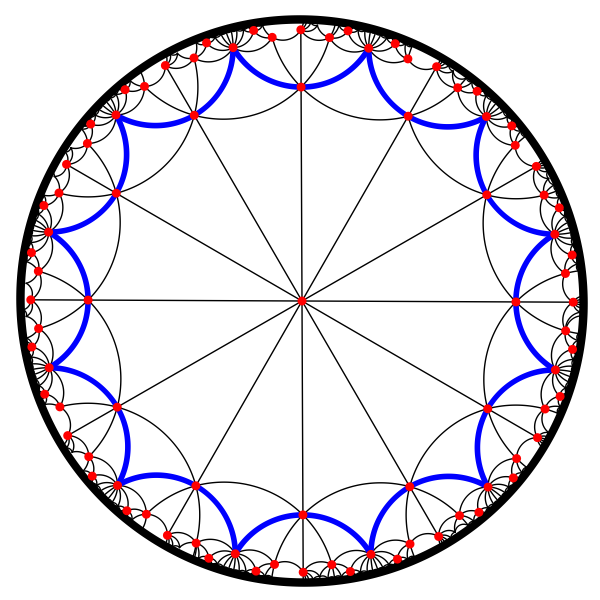} &
		\includegraphics[width=.45\textwidth]{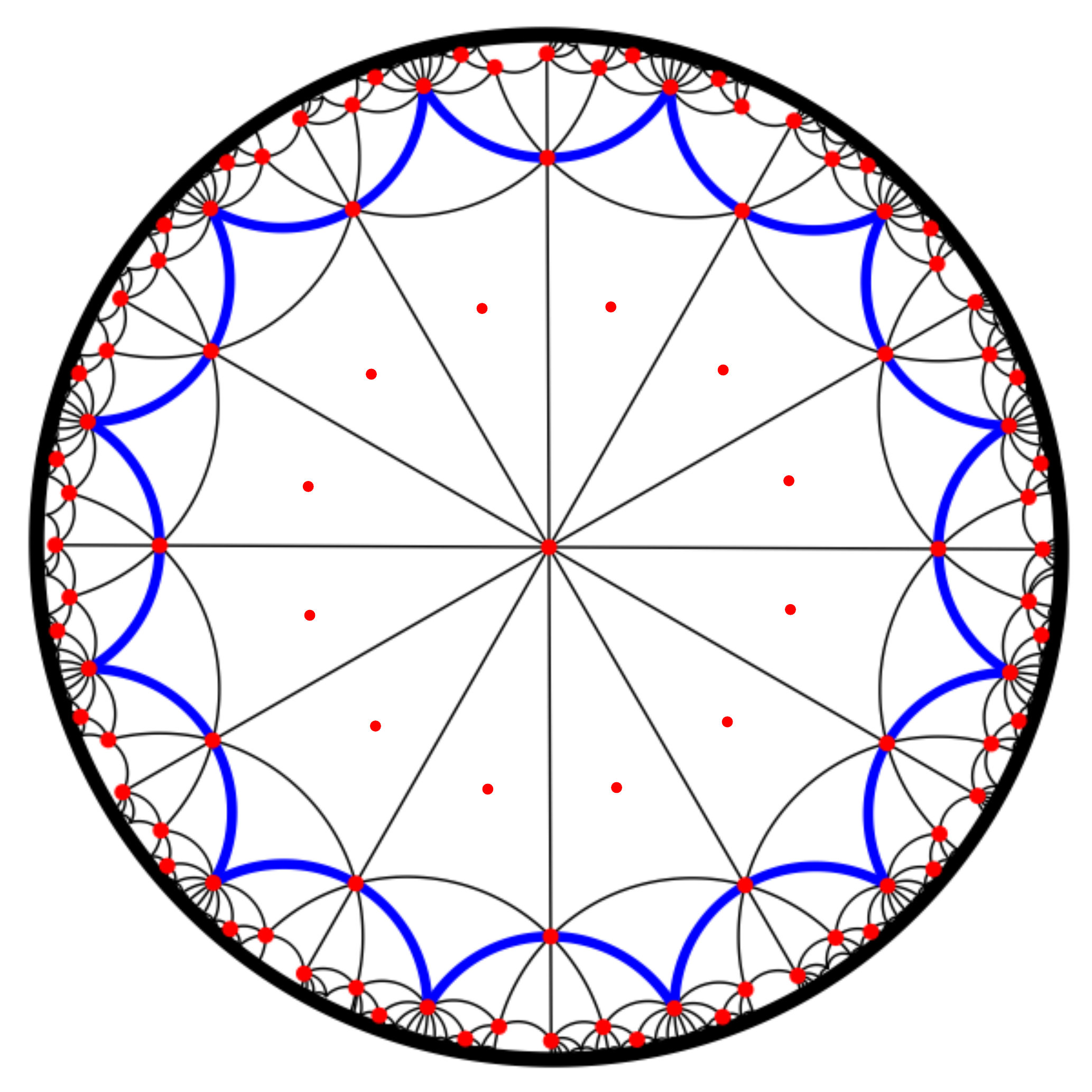} \\
		After initialization & First insertion \\
		\includegraphics[width=.45\textwidth]{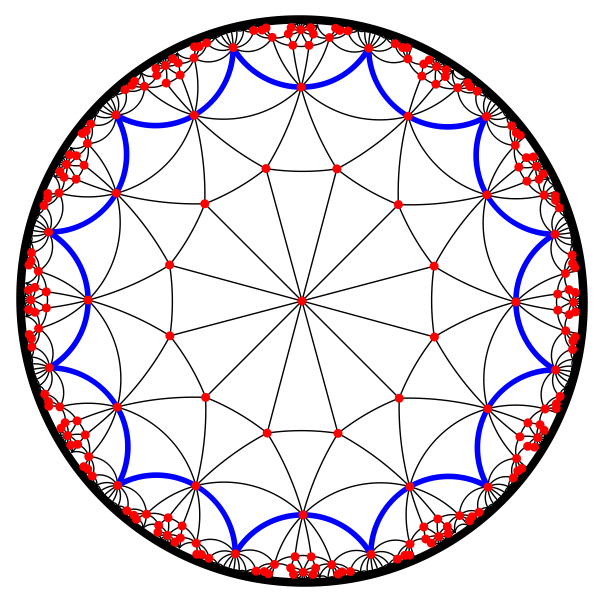} &
		\includegraphics[width=.45\textwidth]{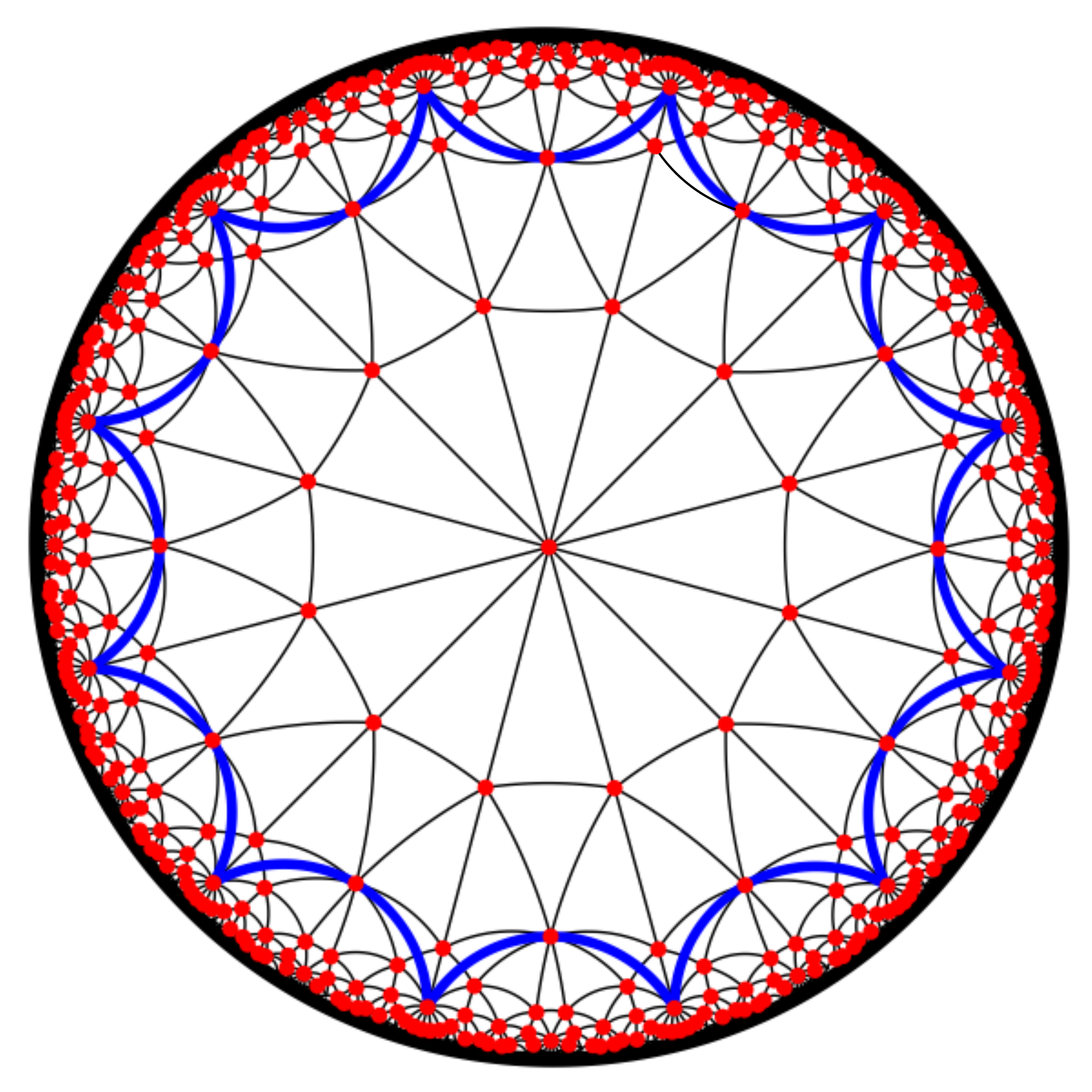} \\
		After first insertion & After second (also last) insertion
	\end{tabular}
	\caption{Several steps in the symmetric algorithm (genus~3)}
	\label{fig:symmetricalgorithm}
\end{figure}

By using the symmetry of the regular $4g$-gon we obtain a more symmetric dummy point set, which may be interesting for some applications~\cite{cff-bhp-11}. However, asymptotically the resulting point set is larger than the point set obtained from the refinement algorithm.

\begin{theorem}\label{thm:symmetricalgorithm}
	The symmetric algorithm terminates and the resulting dummy point set satisfies the validity condition~\eqref{condition}. Its cardinality is of order $\Theta(g\log g)$.
\end{theorem}
\begin{proof}
	The first two statements follow directly from the proof of Theorem~\ref{thm:simplealgorithm}, so we only have to prove the claim on the cardinality of ${\dummy}$.\\ 
	First, we prove that $|{\dummy}|$ is of order $O(g\log g)$. Again, the distance between the Weierstrass points is more than $\tfrac{1}{4}\sysg$. We claim that the distance between points that are added in different iterations of the while loop is at least $\tfrac{1}{4}\sysg$. Namely, by the same reasoning as in the proof of Theorem~\ref{thm:simplealgorithm}, the distance between the circumcenter of an empty disk of radius at least $\tfrac{1}{4}\sysg$ and any other point in $\dummy$ is at least $\tfrac{1}{4}\sysg$. Because $\dummy$ is invariant under symmetries of $D_g$, it follows that the distance between an image of the circumcenter under a rotation around the origin and any other point in $\dummy$ is at least $\tfrac{1}{4}\sysg$ as well. 
	
	However, the distance between points in ${\dummy}$ can be smaller than $\tfrac{1}{4}\sysg$ if they are added simultaneously in some iteration of the while loop. Denote the points added to ${\dummy}$ in iteration $j$ by $\proj(\ph^j_k)$ where $k=0,\ldots,4g-1$. In particular, $\ph^j_k$ is the circumcenter of a triangle in $\dth{\h{\dummyq}_{\N}}$, i.e, in the hyperbolic plane. 
	
	Let $D(p,r)$ be the hyperbolic disk with center $p$ and radius $r$, where $p$ is either a point in $\H^2$ or in $\Mg$. For each iteration $j$, define 
	$$\Uh_j=\bigcup_{k=0}^{4g-1}D\left(\ph^j_k,\tfrac{1}{8}\sysg\right)$$
	and let $U_j=\proj(\Uh_j)$. Let $a_j=\area(U_j)$. Denote the area of a hyperbolic circle of radius $\tfrac{1}{8}\sysg$ by $a$, i.e.
	$$ a:=2\pi\left(\cosh\left(\tfrac{1}{8}\sysg\right)-1\right).$$
	Observe that $a\leq a_j\leq 4ga$, where the lower bound is in the limiting case where all disks are equal and the upper bound in the case where all disks are disjoint. 
	
	Define
	$$ \I=\{j\;|\;a_j< 2ga \}$$
	and denote its complement by $\I^c$. We give upper bounds for $|\I|$ and $|\I^c|$. To see for which $j$ the inequality $a_j<2ga$ holds, we first look at the area of $\Uh_j$ (see Figure~\ref{fig:symmetric-algo-case1}). The amount of overlap between $D(\ph^j_k)$ and $D(\ph^j_{k+1})$ can be written as a strictly decreasing function of $d(\ph^j_k,\ph^j_{k+1})$, which can be written as a strictly increasing function of $d(O,\ph^j_k)$. Therefore, there exists a constant $d_g>0$ such that $\area(\Uh_j)<2ga$ if and only if $d(O,\ph^j_k)<d_g$ for all $k=0,\ldots,4g-1$. 
	
	We claim that $j\in\I$ if and only if there exists $k\in\{0,\ldots,4g-1\}$ such that either $d(O,\ph^j_k)< d_g$ or $d(\vh_0,\ph^j_k)< d_g$. First, assume that such a $k$ exists. If $d(O,\ph^j_k)< d_g$ (Figure~\ref{fig:symmetric-algo-case1}), then $\area(\Uh_j)<2ga$ by definition of $d_g$, so $j\in\I$. Now, assume that $d(\vh_0,\ph^j_k)< d_g$ (Figure~\ref{fig:symmetric-algo-case2}). By symmetry $d(\vh_\ell,\ph^j_{k+\ell})=d(\vh_0,\ph^j_k)$ for all $\ell=0,\ldots,4g-1$ (counting modulo $4g$). Recall that $f_0$ is the side-pairing transformation that maps $\sh_{2g}$ to $\sh_0$. Then
	\begin{align*}
		d(f_0(\ph^j_{k+2g+1}),\vh_0)&=d(f_0^{-1}(f_0(\ph^j_{k+2g+1})),f_0^{-1}(\vh_0)),\\
		&=d(\ph^j_{k+2g+1},\vh_{2g+1}),\\
		&=d(\ph^j_k,\vh_0).
	\end{align*}
	Therefore, the circle $\Ch_j$ centered at $\vh_0$ and passing through $\ph^j_k$ passes through $f_0(\ph^j_{k+2g+1})$ as well. By induction, for every pair of adjacent fundamental regions $f(D_g)$ and $f'(D_g)$ that contain $\vh_0$ there exists an $\ell\in\{0,\ldots,4g-1\}$ such that $f(\ph^j_\ell)$ and $f'(\ph^j_{\ell+2g+1})$ are equidistant from $\vh_0$. There are $4g$ fundamental regions that have $\vh_0$ as one of their vertices. Because $2g+1$ and $4g$ are co-prime, it follows that $\Ch_j$ contains exactly one translate of $\ph^j_\ell$ for every $\ell=0,\ldots,4g-1$. Hence, if we translate the union of disks of radius $\tfrac{1}{8}\sysg$ centered at the translates of $\ph^j_\ell, \ell=0,\ldots,4g-1$ on $\Ch_j$ by the hyperbolic translation that maps $\vh_0$ to the origin, we obtain a union of disks of radius $\tfrac{1}{8}\sysg$ at distance $d(\vh_0,\ph^j_k)<d_g$ from the origin. By definition of $d_g$, it follows that $a_j<2ga$.
	
	\begin{figure}[htbp]
		\centering
		\begin{subfigure}{0.5\textwidth}
			\centering
			\includegraphics[width=0.9\textwidth]{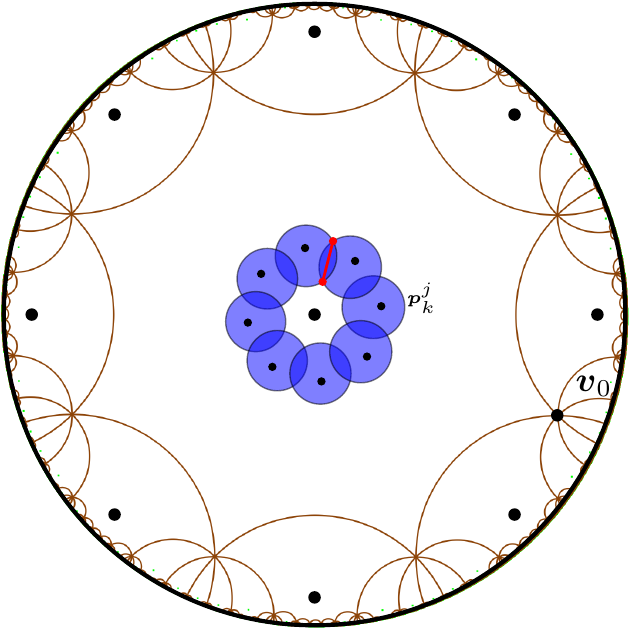}
			\caption{$d(O,\ph^j_k)< d_g$ for all $k=0,\ldots,4g-1$.}
			\label{fig:symmetric-algo-case1}
		\end{subfigure}%
		\begin{subfigure}{0.5\textwidth}
			\centering
			\includegraphics[width=0.9\textwidth]{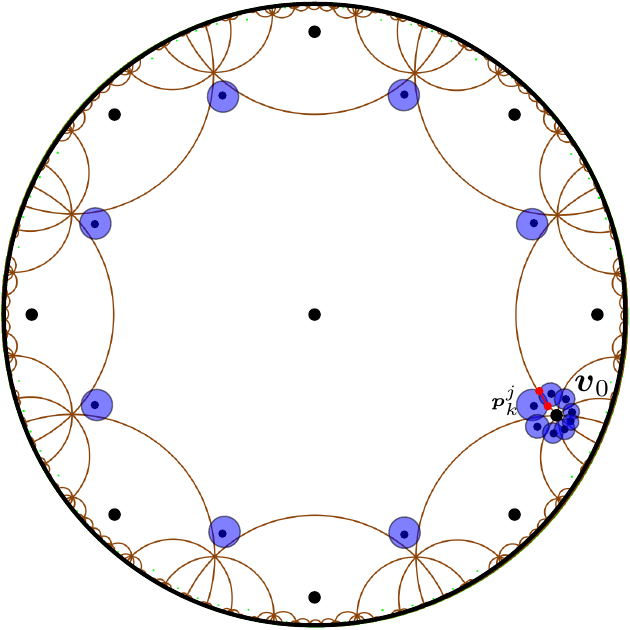}
			\caption{$d(\vh_0,\ph^j_k)< d_g$ for some $k\in\{0,\ldots,4g-1\}$.}
			\label{fig:symmetric-algo-case2}
		\end{subfigure}
		\\
		\begin{subfigure}{0.5\textwidth}
			\centering
			\includegraphics[width=0.9\textwidth]{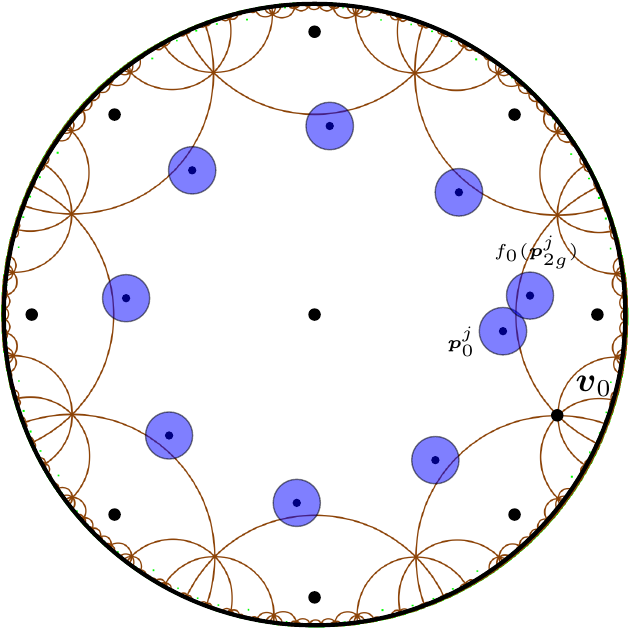}
			\caption{$d(O,\ph^j_k)\geq d_g$ and $d(\vh_0,\ph^j_k)\geq d_g$ for all $k\in\{0,\ldots,4g-1\}$ and $d(\ph^j_0,\partial D_g)<\tfrac{1}{8}\sysg$.}
			\label{fig:symmetric-algo-case3}
		\end{subfigure}%
		\caption{Schematic drawings of different cases. In the first two drawings, the minimum width of the corresponding annulus is marked in red. In the second drawing, only the disks with center in $D_g$ or sufficiently close to $\vh_0$ are drawn. In the third drawing, only the disks with center in $D_g$ together with the unique disk that overlaps $D(\ph^j_0,\tfrac{1}{8}\sysg)$ are drawn.}
		\label{fig:symmetric-algo-cases}
	\end{figure}
	
	Second, assume that $d(O,\ph^j_k)\geq d_g$ and $d(\vh_0,\ph^j_k)\geq d_g$ for all $k\in\{0,\ldots,4g-1\}$. If $d(\ph^j_0,\partial D_g)\geq\tfrac{1}{8}\sysg$, then $\Uh_j$ is completely contained in $D_g$. Because $d(O,\ph^j_k)\geq d_g$, it follows that $a_j\geq 2ga$ by definition of $d_g$, so $j\in\I^c$. Now, assume that $d(\ph^j_0,\partial D_g)<\tfrac{1}{8}\sysg$. If $\ph^j_0$ is close to the midpoint of a side of $D_g$, then $D(\ph^j_0,\tfrac{1}{8}\sysg)$ can only overlap with a translate of $D(\ph^j_{2g},\tfrac{1}{8}\sysg)$ (Figure~\ref{fig:symmetric-algo-case3}). Then, $U_j$ contains at least $2g$ pairwise disjoint disks, so $a_j\geq 2ga$. Therefore, $j\in\I^c$. Hence, the only way that $D(\ph^j_0,\tfrac{1}{8}\sysg)$ can overlap with multiple other disks is when $\ph^j_0$ is sufficiently close to a vertex of $D_g$. Consider again the circle $\Ch_j$ centered at $\vh_0$ and passing through a translate of $\vh^j_\ell$ for all $\ell\in\{0,\ldots,4g-1\}$. Because now $d(\vh_0,\ph^j_k)\geq d_g$, it follows that $a_j\geq 2ga$ by definition of $d_g$.  
	
	We conclude that $j\in\I$ if and only if there exists $k\in\{0,\ldots,4g-1\}$ such that either $d(O,\ph^j_k)< d_g$ or $d(\vh_0,\ph^j_k)< d_g$. We have also shown that if $d(O,\ph^j_k)< d_g$, then $\Uh_j$ is a topological annulus around the origin. If $d(\vh_0,\ph^j_k)< d_g$, then $\proj^{-1}(U_j)$ contains a topological annulus around $\vh_0$. In either case, the boundary of such an annulus consists of two connected components. Let the minimum width of an annulus be given by the distance between these connected components. Suppose, for a contradiction, that the minimum width of an annulus corresponding to $j\in\I$ can be arbitrarily close to 0. Then the disks in $U_j$ have arbitrarily small overlap, so $a_j$ is arbitrarily close to $4ga$. However, this is not possible, since $a_j<2ga$ for all $j\in\I$. Therefore, there exists $\varepsilon>0$ (independent of the output of the algorithm) such that the minimum width of an annulus corresponding to $j\in\I$ is at least $\varepsilon$. 
	
	To find an upper bound for $|\I|$, consider the line segment $[O,\vh_0]$ between the origin and $\vh_0$. By the above discussion, $[O,\vh_0]$ crosses the annulus corresponding to any $j\in\I$ exactly once. Because the annuli are pairwise disjoint and each annulus has minimum width $\varepsilon$, there are at most $\hlen{[O,\vh_0]}/\varepsilon$ annuli, where
	$$\hlen{[O,\vh]}=\arcosh\left(\cot^2\left(\tfrac{\pi}{4g}\right)\right).$$
	Therefore,
	$$ |\I|\leq \dfrac{\arcosh\left(\cot^2\left(\tfrac{\pi}{4g}\right)\right)}{\varepsilon}.$$
	Because $\cot^2(\tfrac{\pi}{4g})\sim \tfrac{16}{\pi^2}g^2$ for $g\rightarrow\infty$, it follows that $|\I|$ is of order $O(\log g)$.
	
	Now, consider $\I^c$. Because the disks of radius $\tfrac{1}{8}\sysg$ centered at points of $\dummy$ that correspond to different iterations of the while loop are disjoint, we see that
	\[ \area(\Mg)\geq\area\left(\cup_{j\in\I^c}U_j \right)=\sum_{j\in\I^c}\area(U_j)=\sum_{j\in\I^c}a_j\geq |\I^c|\cdot 2ga.
	\]
	Since $\area(\Mg)=4\pi(g-1)$ and $a$ is constant, $|\I^c|$ is of order $O(1)$. 	
	
	Because the number of iterations is given by $|\I|+|\I^c|$, the number of iterations is of order $O(\log g)$. Each iteration adds $4g$ points, so the resulting dummy point set has cardinality of order $O(g\log g)$.
	
	Secondly, we show that $|{\dummy}|$ is of order $\Omega(g\log g)$. As before, the points added to ${\dummy}$ in iteration $j$ of the while loop are denoted by $\proj(\ph^j_k)$ where $k=0,\ldots,4g-1$. Fix an arbitrary vertex $\vh$ of $D_g$. Let $P=\angular{O,\ph^{j_1}_{k_1},\ph^{j_2}_{k_2},\ldots,\ph^{j_n}_{k_n},\vh}$ be a shortest path from the origin to $\vh$ in the Delaunay graph of $\proj^{-1}(\dummy)$. We claim that all indices $j_h$ are distinct, i.e. $P$ contains at most one element of each of the sets $\{ \ph^j_k\;|\;k=0,\ldots,4g-1\}$ (see Figure~\ref{fig:structuredalgorithmlowerboundpoints}). 
	
	\begin{figure}[htbp]
		\centering
		\includegraphics[width=.9\textwidth]{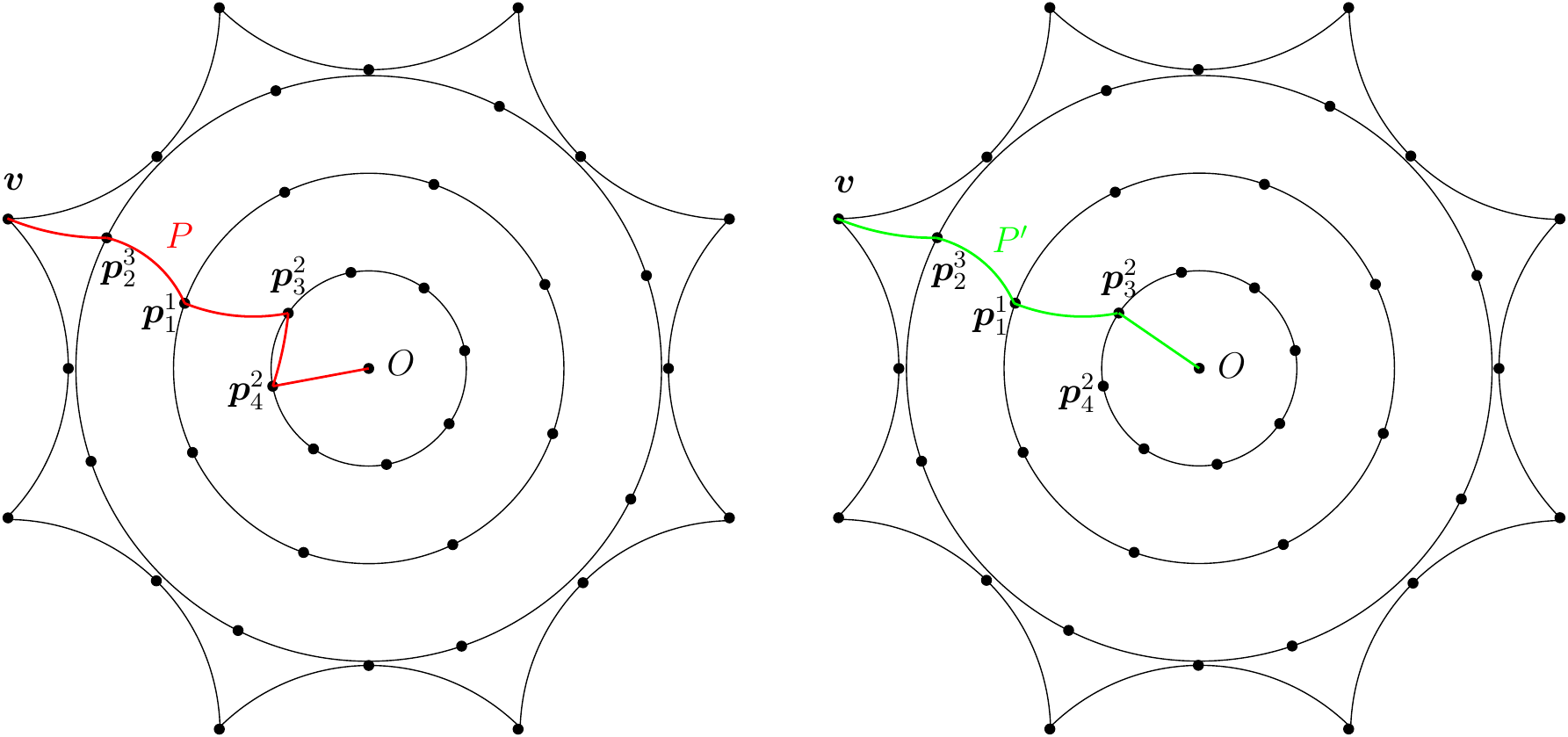}
		\caption{The left figure shows a path $P$ from the origin to $\vh$ that visits two vertices from the same iteration, namely $\ph^2_3$ and $\ph^2_4$. The right figure shows a shorter path from the origin to $\vh$. In this case, $j_1=j_2=2$ and $k_1=4$ and $k_2=3$. The subdivision of $P$ into three parts is given by $P_1=\angular{O,\ph^2_4},P_2=\angular{\ph^2_4,\ph^2_3}$ and $P_3=\angular{\ph^2_3,\ph_1^1,\ph^3_2,\vh}$. The path $P'$ is defined as $P'_1\cup P_3$, where $P'_1$ is obtained by rotating $P_1$ around the origin by angle $-\tfrac{\pi}{2g}$, i.e., $P'_1=\angular{O,\ph^2_3}$.}
		\label{fig:structuredalgorithmlowerboundpoints}
	\end{figure}
	
	Suppose, for a contradiction, that there exist $l$ and $m$ with $l<m$, such that $j_l=j_m$. We will construct a path $P'$ from $O$ to $\vh$ that is shorter than $P$. We know that $\ph^{j_l}_{k_l}\neq\ph^{j_m}_{k_m}$, because otherwise the shortest path would contain a cycle, so in particular $k_l\neq k_m$. Subdivide $P$ into three paths: the path $P_1$ from $O$ to $\ph^{j_l}_{k_l}$, the path $P_2$ from $\ph^{j_l}_{k_l}$ to $\ph^{j_m}_{k_m}$, and the path $P_3$ from $\ph^{j_m}_{k_m}$ to~$\vh$. Now, let $P'_1$ be the image of $P_1$ after rotation around $O$ by angle $(k_m-k_l)\cdot\tfrac{\pi}{2g}$. It is clear that $P'_1$ is a path from $O$ to $\ph^{j_m}_{k_m}$ of the same length of $P_1$. It follows that $P':=P'_1\cup P_3$ is a path from $O$ to $\vh$ that is shorter than $P$. This is a contradiction, so all indices $j_h$ are distinct. Therefore, the number of vertices of the graph that $P$ visits is smaller than the number of iterations of the while loop. Each edge in the path $P$ is the side of a triangle with circumdiameter smaller than $\tfrac{1}{2}\sysg$, so in particular the length of each edge is smaller than $\tfrac{1}{2}\sysg$. The length of $P$ is at least 
	$$\hlen{[O,\vh]}=\arcosh\left(\cot^2\left(\tfrac{\pi}{4g}\right)\right)\sim 2\log g.$$
	As $\tfrac{1}{2}\sysg$ is bounded as a function of $g$ (Theorem~\ref{thm:systolevalue}), the number of edges in $P$ is of order $\Omega(\log g)$. Then, the number of iterations of the while loop is of order $\Omega(\log g)$, so $|{\dummy}|$ has cardinality of order $\Omega(g \log g)$. The result follows by combining the lower and upper bounds.
\end{proof}

\subsection{Experimental results for small genus\label{sec:dummy-exp}}
The refinement algorithm and the symmetric algorithm have been implemented. The implementation uses the \expr number type \cite{core:library} to represent coordinates of points, which are algebraic numbers. 

For the Bolza surface (genus 2), both algorithms compute a set of 22 dummy
points. In Figure~\ref{fig:dummy-bolza} we have shown the dummy point set computed by the symmetric algorithm.
However, a smaller set, consisting of 14 dummy points, was proposed earlier~\cite{btv-dtosl-16}: in addition to the six Weierstrass points, it contains the eight midpoints of the segments $[O, \vh_k], \;k = 0, 1, \ldots 7$  (see
Figure~\ref{fig:dummy-bolza}). 

\begin{figure}[ht]
	\centering
	\includegraphics[width=.45\textwidth,valign=t]{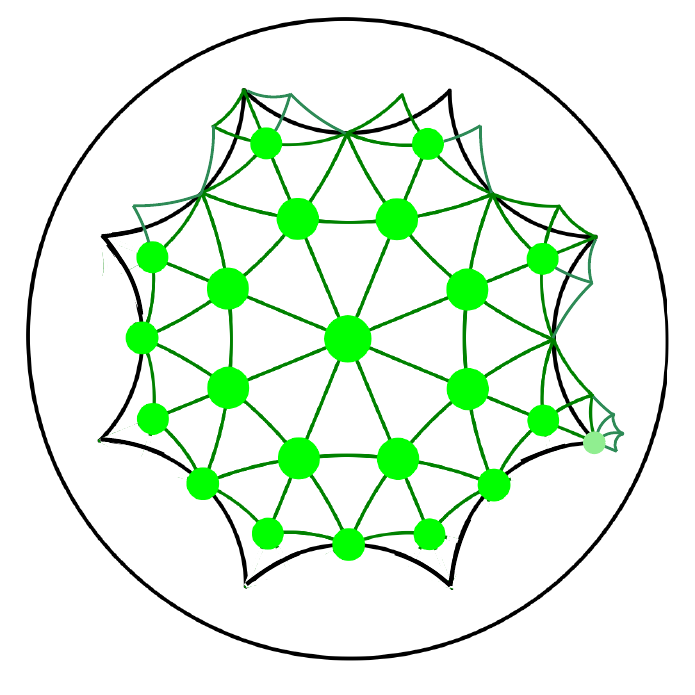}
	\includegraphics[width=.42\textwidth,valign=t]{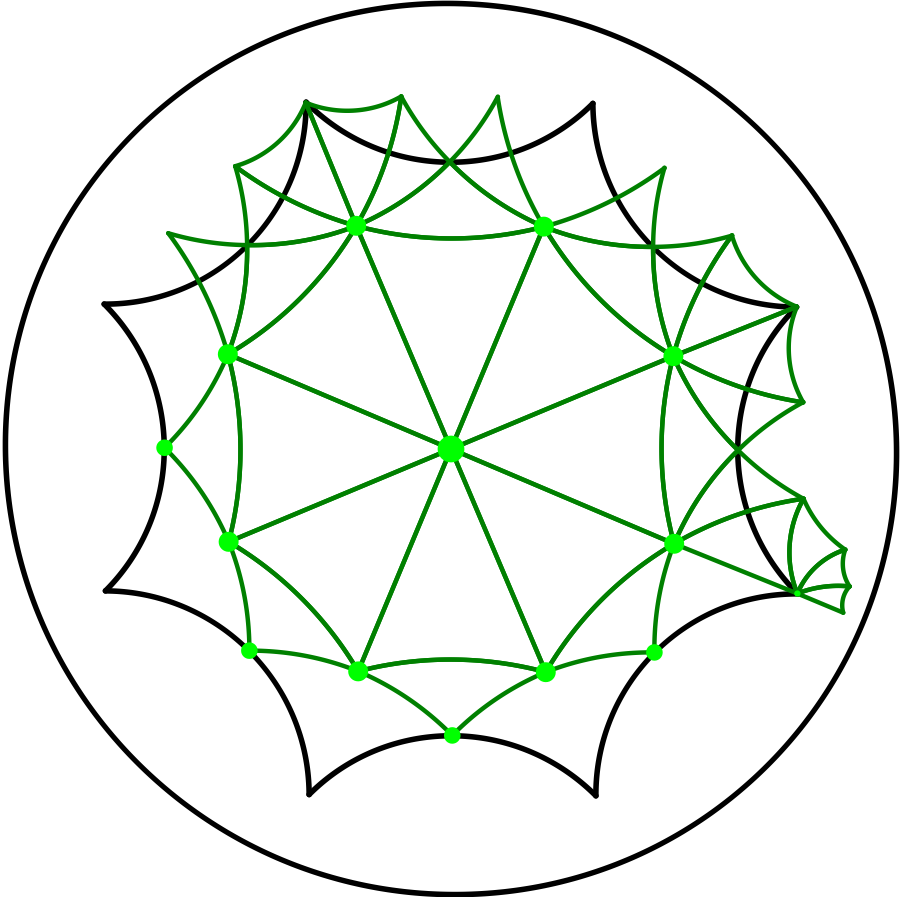}
	\caption{Set of 22 dummy points for the Bolza surface computed by the
		symmetric algorithm (left) and set of 14 dummy points
		constructed by hand~\cite{btv-dtosl-16} (right).}
	\label{fig:dummy-bolza}
\end{figure}

The computation does not terminate for higher genus after seven hours of computations when performing the computations exactly. To be able to obtain a result, we impose a finite precision to \expr. 

For genus~3, we obtain sets of dummy points with both strategies with precision $512\times g$ bits (chosen empirically). The refinement algorithm yields a set of 28 dummy points (Figure~\ref{fig:simplealgorithm}), while the symmetric algorithm leads to 32 dummy points (Figure~\ref{fig:symmetricalgorithm}). 
Computing dummy point sets for Bolza surfaces of higher genus poses a challenge regarding the evaluation of algebraic expressions. Our experiments show that we have to design a new strategy for arithmetic computations, which goes beyond the scope of this paper. 	

%%%%%%%%%%%%%%%%%%%%%%%%%%%%%%%%%%%%%%%%%%%%%%%%%%%%%%%%%%%%%

\section{Data structure, predicates, and implementation}
\label{sec:representation-of-delaunay-triangulations}

In this section, we detail two major aspects of Bowyer's algorithm for
generalized Bolza surfaces. On the one hand, the combinatorial aspect,
i.e., the data structure and the way it supports the algorithm, is
studied in Section~\ref{sec:data-structure}.  On the other hand, the
algebraic degree of the predicates based on which the decisions are
made by the algorithm is analyzed in Section~\ref{sec:predicates}.
Finally, we report on our implementation and experimental results in
Section~\ref{sec:implementation}.

Let us first define a unique canonical representative for each
triangle of a triangulation, which is a major ingredient for the data
structure.

\subsection{Canonical representatives}
\label{sec:canonical}

We have defined in Section~\ref{sec:generalized} the canonical
representative of a point on the surface $\Mg$. Let us now determine
a unique canonical representative for each orbit of a triangle in
$\mD$ under the action of $\Gg$.

We consider all the
neighboring regions, i.e., the images of $D_g$ by a translation in
$\N\setminus\{\eg\}$ (see Section~\ref{sec:generalized}, to be ordered counterclockwise around $0$,
starting with the Dirichlet region
\[ \prod_{j=0}^{2g-1}f_{j(2g+1)} (D_g) =f_0f_{2g+1}f_{2(2g+1)}\ldots f_{(2g+1)^2} (D_g) \]
(where indices are taken modulo $4g$)
incident to $\vh_0$, which gives an ordering of $\N\setminus\{\eg\}$. An illustration
for genus 2 is shown in
Figure~\ref{fig:neighboring-region}.

\begin{figure}[htb]
	\centering
	\includegraphics[height=5cm]{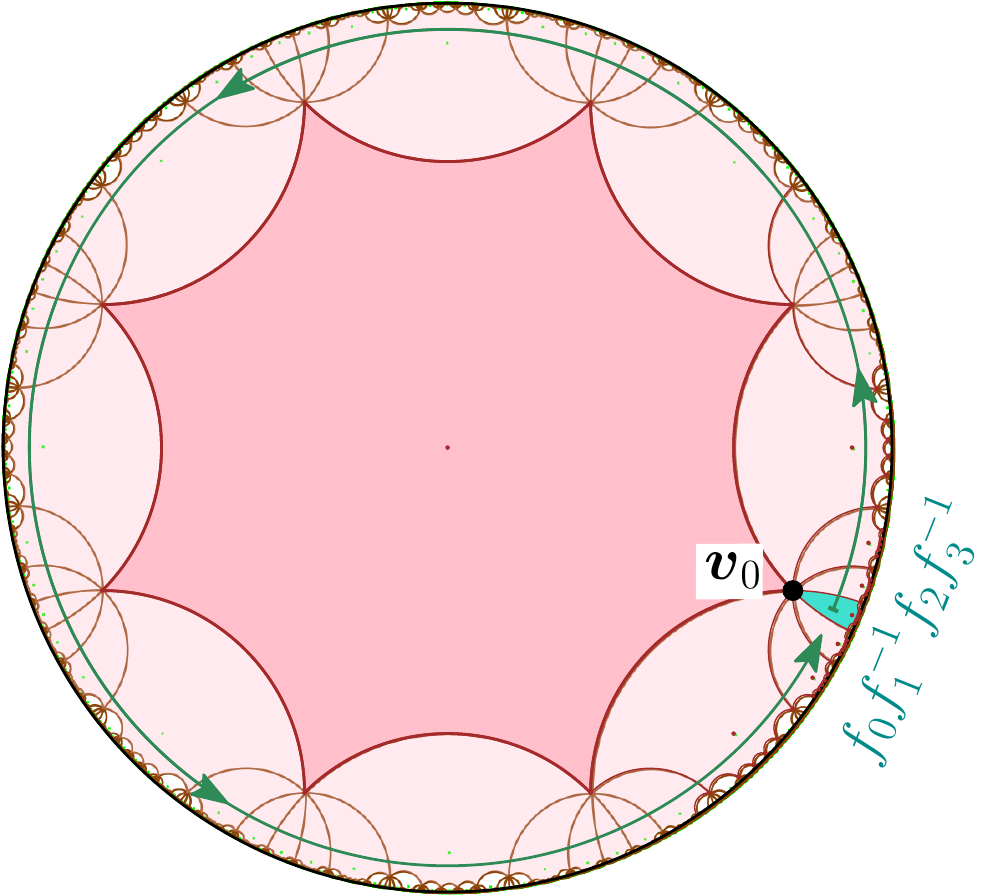}
	\caption{The ordering of $\Nb$ starts with
		$f_0f_{1}^{-1}f_2f_{3}^{-1}=f_0f_5f_2f_7$.}
	\label{fig:neighboring-region}
\end{figure}

We say that a triangle in $\mD$ is \emph{admissible} if its
circumdiameter is less than half the systole of $\Mg$. We can prove
the following property:

\begin{proposition}[Inclusion property]\label{prop:inclusionproperty}
	If at least one vertex of an admissible triangle is contained in $\Do$, then the whole triangle is contained in $D_{\N}$.
	\label{inclusion-property}
\end{proposition}

\begin{proof}
	It is sufficient to show that the distance between the boundary $\partial D_g$ of $D_g$ and the boundary $\partial D_{\N}$ of $D_{\N}$ is at least $\tfrac{1}{2}\sysg$. Consider points $\ph\in \partial D_g$ and $\qh\in\partial D_{\N}$. We will show that $d(\ph,\qh)\geq\tfrac{1}{2}\sysg$. By symmetry of $D_g$, we can assume without loss of generality that $\ph\in \sh_0$. In Section~\ref{sec:caseanalysis}, we gave a definition for a $k$-segment and a $k$-separated segment, where the segment is a hyperbolic line segment between sides of $D_g$. This definition extends naturally to line segments between sides of a translate of $D_g$. 
	
	\begin{figure}[ht]
		\centering
		\begin{subfigure}{0.3\textwidth}
			\centering
			\includegraphics[width=0.6\textwidth]{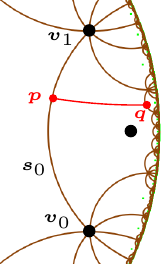}
			\caption{Case where $\qh\in f_0(D_g)$.}
			\label{fig:inclpropcase1}
		\end{subfigure}%
		\begin{subfigure}{0.3\textwidth}
			\centering
			\includegraphics[width=0.6\textwidth]{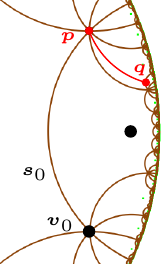}
			\caption{Case where $\ph=\vh_1$.}
			\label{fig:inclpropcase2}
		\end{subfigure}%
		\begin{subfigure}{0.3\textwidth}
			\centering
			\includegraphics[width=0.6\textwidth]{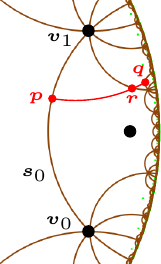}
			\caption{Case where $[\ph,\qh]$ intersects $f_0(D_g)\setminus\partial D_{\N}$.}
			\label{fig:inclpropcase3}
		\end{subfigure}%
		\caption{Cases in the proof of Proposition~\ref{prop:inclusionproperty}.}
		\label{fig:inclusion-property-proof}
	\end{figure}
	
	Recall that $f_0$ is the side-pairing transformation that maps $\sh_{2g}$ to $\sh_0$. First, assume that $\qh\in f_0(D_g)$. Because $\qh\in\partial D_{\N}$, $[\ph,\qh]$ is a segment of separation at least 2 (see Figure~\ref{fig:inclpropcase1}). By Part~\ref{item:sep2} of Lemma~\ref{lem:separationlowerbound}, $d(\ph,\qh)\geq\tfrac{1}{2}\sysg$. 
	
	Second, assume that $\qh\not\in f_0(D_g)$. Without loss of generality, we may assume that $\qh$ is contained in a translate of $D_g$ that contains either $\vh_0$ or $\vh_1$ as a vertex. If $\ph$ is either $\vh_0$ or $\vh_1$, then again $[\ph,\qh]$ is a segment of separation at least 2 (see Figure~\ref{fig:inclpropcase2}), so $d(\ph,\qh)\geq\tfrac{1}{2}\sysg$ by Part~\ref{item:sep2} of Lemma~\ref{lem:separationlowerbound}. If $\ph$ is not a vertex of $D_g$, then $[\ph,\qh]$ intersects one of the two sides in $f_0(D_g)\setminus\partial D_{\N}$, say in a point $\rh$ (see Figure~\ref{fig:inclpropcase3}). In particular, $[\ph,\rh]$ is a 1-separated segment. If $[\rh,\qh]$ is a segment of separation at least 2, then $d(\rh,\qh)\geq\tfrac{1}{2}\sysg$ by Part~\ref{item:sep2} of Lemma~\ref{lem:separationlowerbound}, so $d(\ph,\qh)\geq\tfrac{1}{2}\sysg$ as well. If $[\rh,\qh]$ is a 1-separated segment, then $d(\ph,\rh)+d(\rh,\qh)\geq\tfrac{1}{2}\sysg$ by Part~\ref{item:sep11} of Lemma~\ref{lem:separationlowerbound}.
	
	We have shown that in all cases $d(\ph,\qh)\geq\tfrac{1}{2}\sysg$, which finishes the proof.	
\end{proof}

Let now $\S\subset\Mg$ be a set of points satisfying the validity condition~\eqref{condition}. 
By definition, all triangles in the Delaunay triangulation $\dth{\proj^{-1}(\S)}$ are
admissible and thus satisfy the inclusion property. 
Let $t$ be a face in the Delaunay triangulation $\dtsg{\S}$.

By definition of $\Do$, each
vertex of $t$ has a unique preimage by $\proj$ in $\Do$, so,
the set 
\begin{equation}
	{\Sigma} = \left\{ \th\in \proj^{-1}(t) \; \mid \; \th
	\mbox{ has at least one vertex in } \Do \right\} 
	\label{equation:set-tau}
\end{equation}
contains at most three faces. See Figure~\ref{figure:canonical}. When
$\Sigma$ contains only one face, then this face is completely
included in $\Do$, and we naturally choose it to be the canonical
representative $\can{t}$ of $t$. Let us now assume that
$\Sigma$ contains two or three faces. From
Proposition~\ref{inclusion-property}, each face $\th\in \Sigma$ is
contained in $D_{\N}$. So, for each vertex $\uh$ of $\th$, there is a
unique translation $T(\uh,\th)$ in $\N$ such that $\uh$ lies in
$T(\uh,\th)(\Do)$. This translation is such that 
$$T(\uh,\th)(\can{u})=\uh.$$
Considering the triangles in $\mD$ to be oriented
counterclockwise, for $\th\in\Sigma$, we denote as $\first{\uh}_\th$
the first vertex of $\th$ that is not lying in
$\Do$. Using 
the ordering on $\N$ defined above, we can now choose $\can{t}$ as
the face of $\Sigma$ for which $T(\first{\uh}_{\can{t}},\can{t})$ 
is closest to $f_0f_{2g+1}f_{2(2g+1)}\ldots
f_{(2g+1)^2}$ 
for the counterclockwise order on $\N$.
\begin{figure}[htb]
	\centering
	\includegraphics[width=0.32\textwidth]{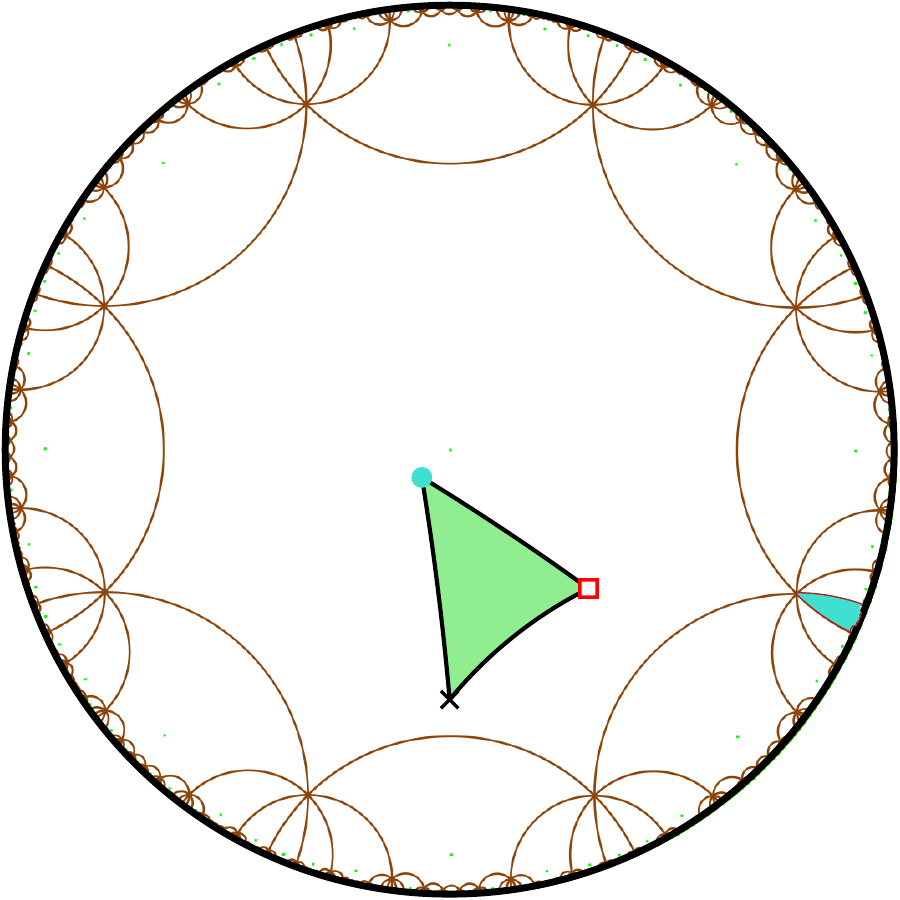}
	\includegraphics[width=0.32\textwidth]{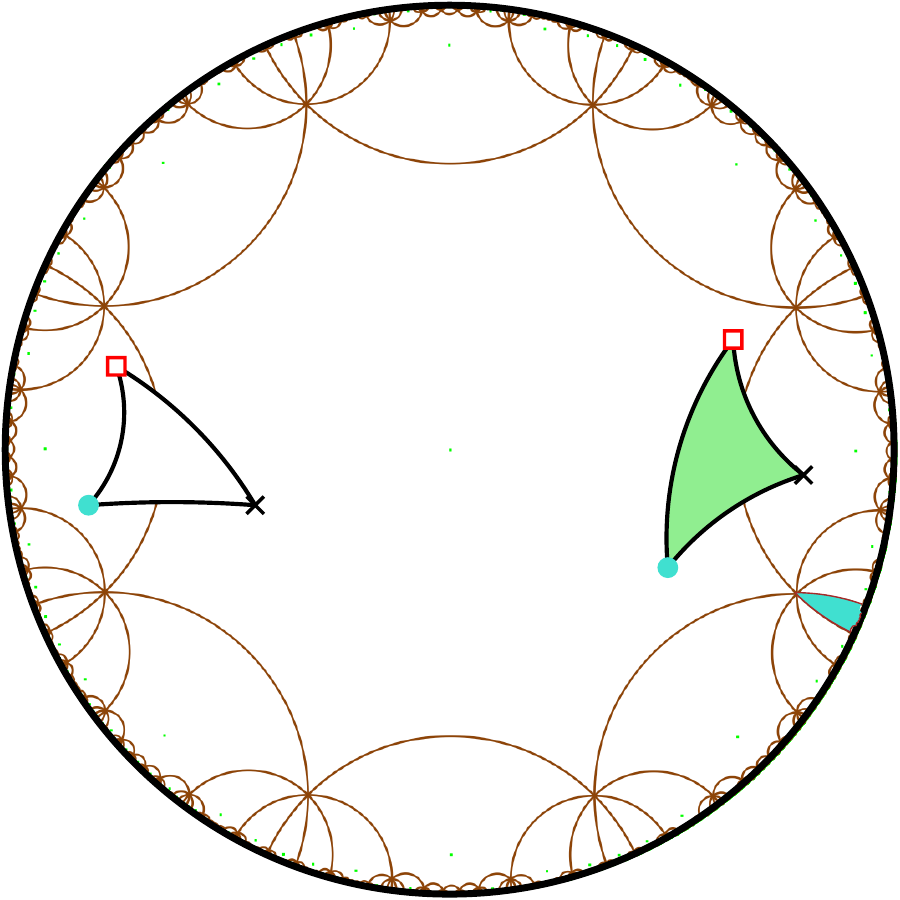}
	\includegraphics[width=0.32\textwidth]{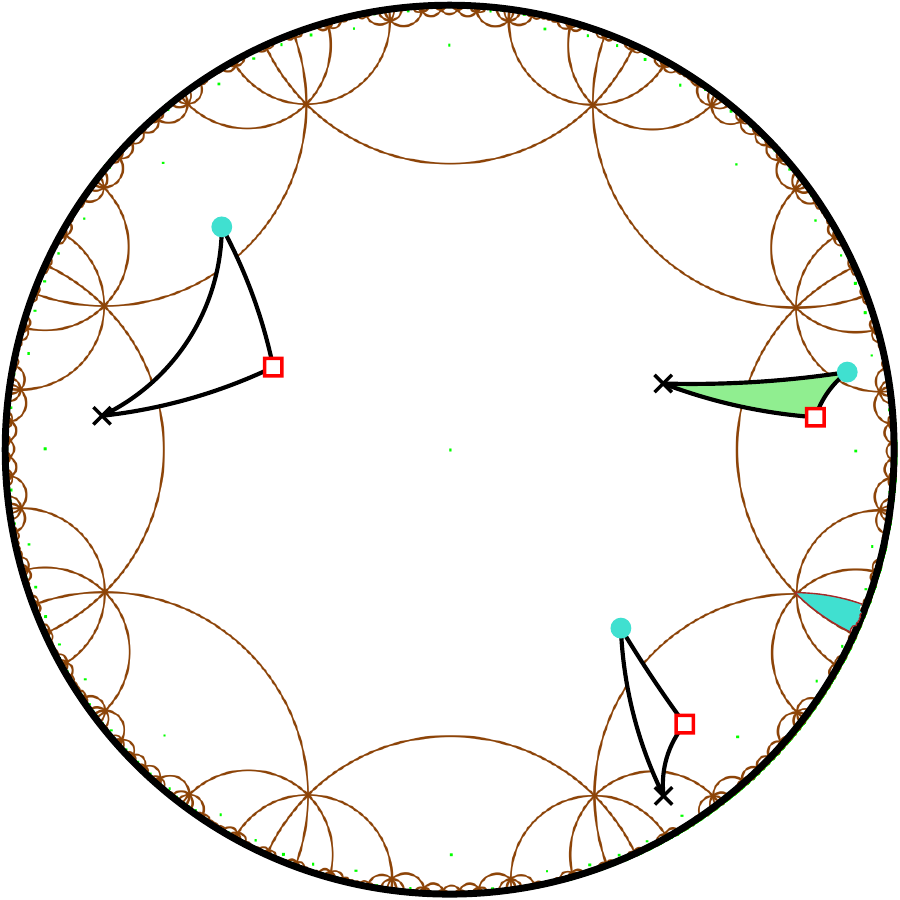}
	\caption[Illustration of candidate faces and canonical
	representatives]{Examples (for $g=2$) of faces of $\dth{\proj^{-1}(\S)}$
		with one (left), two (middle) and three (right) vertices in
		$\Do$ that project to the same face on $\Mg$. Their respective
		vertices drawn as a dot project to the same vertex on $\Mg$
		(same for cross and square). The canonical representative is the
		shaded face.}
	\label{figure:canonical}
\end{figure}

To summarize, we have shown that:

\begin{proposition}
	Let $\S \subset \Mg$ be a set of points satisfying the validity
	condition~\eqref{condition}. For any face $t$ in $\dtsg{\S}$,
	there exists a unique canonical representative
	$\can{t}\subset D_{\N}$ in $\dth{\proj^{-1}(\S)}$.
	\label{prop:canonical-representative-existence}
\end{proposition}

Using a slight abuse of vocabulary, for a
triangle $\th$ in $\mD$, we will sometimes refer to the canonical
representative $\can{t}$ of its projection $t=\proj(\th)$ as the canonical
representative of $\th$. 

\subsection{Data structure}\label{sec:data-structure}
Proposition~\ref{prop:canonical-representative-existence} allows us to
propose a data structure to represent Delaunay triangulations of
generalized Bolza surfaces. 

A triangulation of a point set $\S\subset\Mg$ is represented via its vertices and triangular
faces. Each vertex $u$ stores its canonical representative $\can{u}$
in $\Do$ and gives access to one of its incident triangles. Each
triangle $t$ is actually storing information to construct its canonical representative
$\can{t}$: it gives access to its three incident vertices $u_0,u_1,$
and $u_2$ and its three adjacent faces; it also stores the three
translations $T(u_j,t) := T\br{\uh_j,\can{t}}, j=0,1,2$ in $\N$ as
defined in Section~\ref{sec:canonical}, so that applying
each translation to the corresponding canonical point yields
the canonical representative $\can{t}$ of $t$, i.e.,
\[ \can{t} = \left( \;
T\br{u_0,t}(\uh_0^\h{c}),\;
T\br{u_1,t}(\uh_1^\h{c}),\;
T\br{u_2,t}(\uh_2^\h{c})\;
\right). \]

In the rest of this section, we show how this data structure 
supports the algorithm that was briefly sketched in
Section~\ref{sec:bower-surf}. 

\paragraph{Finding conflicts.}
The notion of conflict defined in section~\ref{sec:bower-surf} can now be made more explicit: a
triangle $t\in\dtsg{\S}$ is in conflict with a point $p\in\Mg$ if the
circumscribing disk of one of the (at most three) triangles in $\Sigma$ is in
conflict with $\can{p}$, where $\Sigma$ is the set defined by
relation~(\ref{equation:set-tau}).

By the correspondence between Euclidean circles and hyperbolic circles in
the Poincar\'e disk model, the
triangle in the Delaunay triangulation in $\mD$ whose
associated Euclidean triangle contains the point $\can{p}$ is in
conflict with this point; 
these Euclidean and hyperbolic triangles will both be denoted as $\th_p$, which should not
introduce any confusion. To find this triangle, we adapt the so-called
\emph{visibility walk}~\cite{dsm-wt-02}: the walk starts from an
arbitrary face, then, for each visited face, it visits one of its
neighbors, until the face whose associated Euclidean triangle contains
$\can{p}$ is found. This walk will be detailed below. 

We first need some notation. Let $t, t'$ be two adjacent faces in
$\dtsg{\S}$. 
We define the \emph{neighbor translation} $\nbt(\can{\th'},\can{\th})$ from
$\can{\th'}$ to $\can{\th}$ as the translation of $\Gg$ such that 
$\nbt(\can{\th'},\can{\th})(\can{\th'})$ is adjacent to $\can{\th}$
in $\dth{\pi^{-1}(\S)}$. See Figure~\ref{figure:neighbor-offset}. 
Let $u$ be a vertex common to $t$ and $t'$, and let
$\uh_j$ and $\uh_{j'}$ be the vertices of $\can{\th}$ and $\can{\th'}$ that
project on $u$ by $\proj$. We can compute the neighbor translation from $\can{\th'}$ to $\can{\th}$ as
$\nbt(\can{\th'},\can{\th})=T(u_{j}, t) (T(u_{j'}, t'))^{-1}.$
It can be easily seen that 
$\nbt(\can{\th'},\can{\th}) = T(u_j,t) (T(u_{j'},t'))^{-1} = 
\left( T(u_{j'},t')(T(u_j,t))^{-1} \right)^{-1} = (\nbt(\can{\th},\can{\th'}))^{-1}$.

\begin{figure}[htb]
	\centering
	\includegraphics{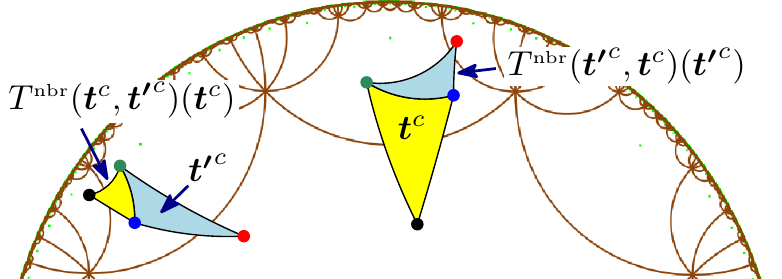}
	\caption{Translating $\can{t'}$ by $\nbt(\can{t'}, \can{t})$ gives a face adjacent to $\can{t}$.}
	\label{figure:neighbor-offset}
\end{figure}

Finally, we define the \emph{location translation} $\ltr$ as the translation that moves the canonical
face $\canindex{t}{p}$ to $\th_p$. This translation is computed during the walk.
The walk starts from a face containing the
origin. As this face is necessarily canonical, $\ltr$ is initialized
to $\eg$. Then, for each visited face $\th$ of $\dth{\pi^{-1}(\S)}$, we consider
the Euclidean edge $\eh$ defined by two of the vertices of $\th$. We
check whether the Euclidean line supporting $\eh$ separates $\can{p}$
from the vertex of $\th$ opposite to $\eh$. If this is the case, the
next visited face is the neighbor $\th'$ of $\th$ through
$\eh$; the location translation is updated:
$\ltr:=\ltr\nbt(\can{\th'},\can{\th})$. The walk stops when it finds
the Euclidean triangle $\th_p$ containing $\can{p}$. Then the
canonical face $\canindex{t}{p}$ in conflict with $\can{p}$ is
$(\ltr)^{-1}(\th_p)$.
See Figure~\ref{fig:bolza-conflict-screenshot} for an
example. Here the walk first visits canonical faces and reaches
the face $\th_{D}\subset D_g$; up to that stage, $\ltr$ is unchanged. Then the walk
visits the non-canonical neighbor $\th'$ of $\th_{D}$, and $\ltr$ is updated to
$\nbt(\can{\th'},\th_D)$. The next face visited by the walk is
$\th_p$, which contains $\can{\ph}$; as $\can{\th_p}$ and $\can{\th'}$ are
adjacent, $\ltr$ is left unchanged.

Let us now present the computation of the set $\h{C}_p$ of faces of
$\dth{\pi^{-1}(\S)}$ in conflict with~$\can{p}$.  Starting from
$\th_p$, for each face of $\dth{\pi^{-1}(\S)}$ in conflict with
$\can{p}$ we recursively examine each neighbor (obtained with a
neighbor translation) that has not yet been visited, checking it for
conflict with $\can{p}$. When a face is found to be in conflict, we
temporarily store directly in each of its vertices the translation
that moves its corresponding canonical point to it (we cannot store
such translations in the face itself, since this face will be deleted
by the insertion). Since the union of the faces of $\h{C}_p$ is a
topological disk by definition, the resulting translation for a given
vertex is the same for all faces of $\h{C}_p$ incident to it, so this
translation is well defined for each vertex. The temporary
translations will be used during the insertion stage described
below. We store the set $\canindex{C}{p}$ of canonical faces
corresponding to faces of $\h{C}_p$. Note that $\canindex{C}{p}$ is
not necessarily a connected region in $\mD$, as illustrated in
Figure~\ref{fig:bolza-conflict-screenshot}(Right). 

\begin{figure}[htb]
	\centering
	\includegraphics[height=6cm]{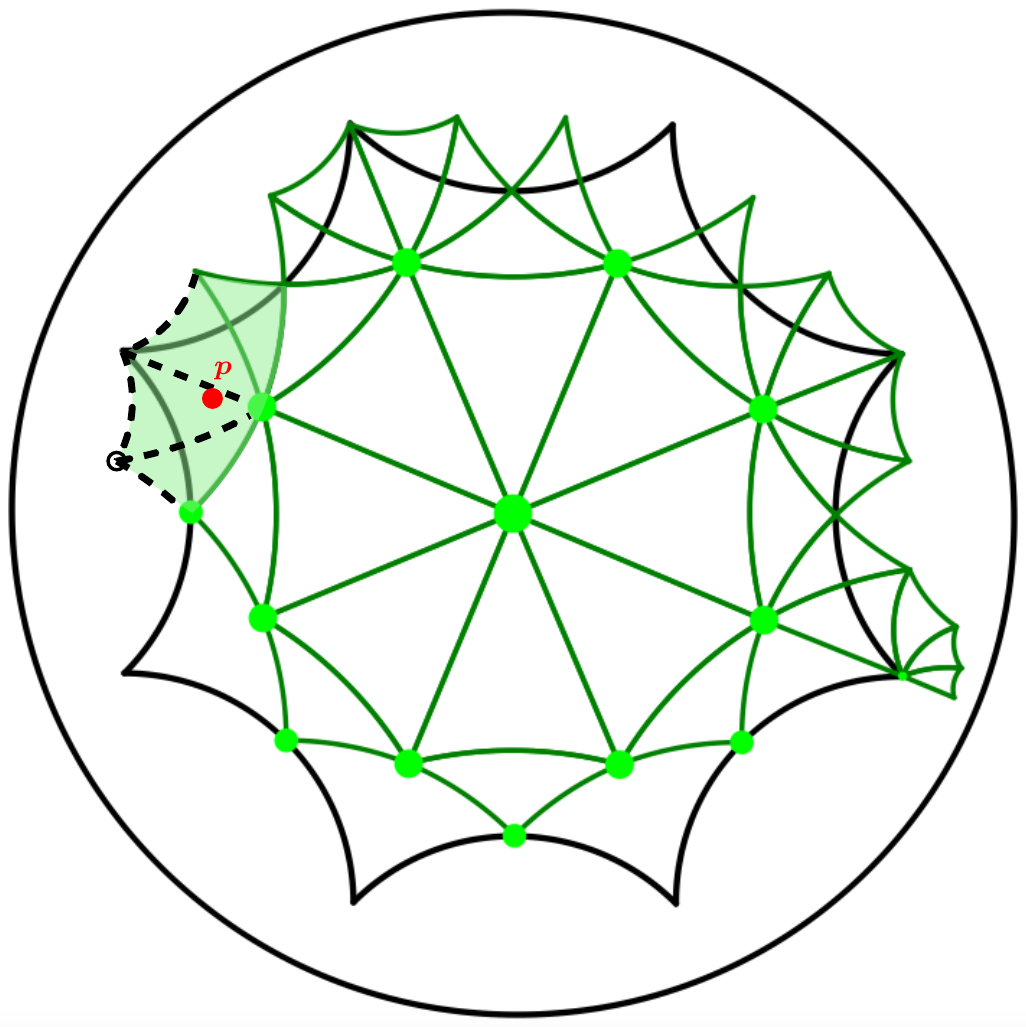}
	\hspace*{.5cm}
	\includegraphics[height=6cm]{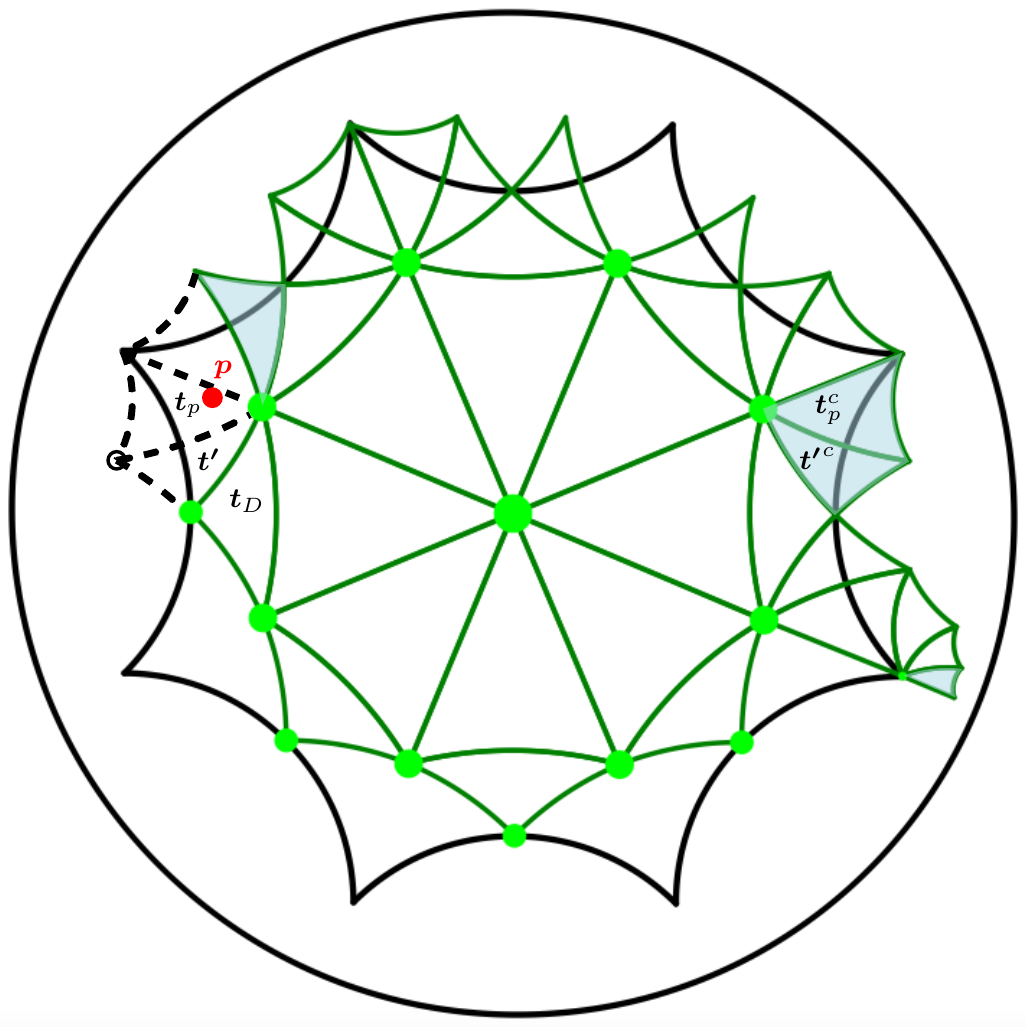} 
	\caption{Left: The shaded faces are the (not necessarily
		canonical) faces in $\h{C}_p$, i.e., faces in
		conflict with the red point $\can{p}$. Their union is a
		topological disk. Right: The region $\canindex{C}{p}$ of the
		(shaded) corresponding canonical
		faces is not connected in~$\mD$.} 
	\label{fig:bolza-conflict-screenshot}
\end{figure}

\paragraph{Inserting a point.} To actually insert the new point $p$ on $\Mg$,
we first create a new vertex storing $\can{p}$.  We store $\eg$
as the temporary translation in the new vertex.

For each edge $\eh$ on the boundary of $\h{C}_p$, we create a new face
$t_\eh$ on $\Mg$ corresponding to the triangle $\th_\eh$ in $\mD$
formed by the new vertex and the edge $\eh$. The neighbor of $t_\eh$
through $\eh$ is the neighbor through $\eh$ of the face in
$\canindex{C}{p}$ that is incident to $\eh$. Two new faces consecutive
along the boundary of $\h{C}_p$ are adjacent. We now delete all faces
in $\h{C}_p$. The triangle $\th_\eh$ is not necessarily the canonical
representative of $t_\eh$; we must now compute the three translations
to be stored in $t_\eh$ to get $\canindex{t}{\eh}$.  To this aim, we
first retrieve the translations temporarily stored in its vertices
$u_j, j=0,1,2$ and we respectively initialize the translations
$T(u_j,t_\eh)$ in $t_\eh$ to them.  If all translations are equal
to~$\eg$, then the face is already canonical and there is nothing more
to do.  Otherwise, the translations stored in the face are updated
following Section~\ref{sec:canonical}: $T(u_j,t_\eh) :=
(T(u_k,t_\eh))^{-1} T(u_j,t_\eh), j=0,1,2$, where $k$ is the index in
$\{0,1,2\}$ for which $\uh_k=\first{\uh}_{\canindex{t}{e}}$.

%%%%%%%%% MT: PLEASE DO NOT ERASE THIS COMMENTED TEXT %%%%%%%%%%%
% Otherwise, one of the vertices of the newly created face
% $\th_\eh$ is the new vertex; without loss of
% generality it is its vertex $\uh_0$, i.e., $T(\uh_0,\th_\eh)=\eg$.
% For $j = 0,1,2$ we can easily compute the
% index $d_j$ of $T(\first{\uh}_{\th_j}, \th_j)$ in the ordering on
% $\N$, where $\th_j$ is the image of $\th_\eh$ under
% $(T(\uh_j, \th_\eh))^{-1}$ and $\first{\uh}_{\th_j}$ is the
% first vertex of $\th_j$ such that
% $T(\first{\uh}_{\th_j}, \th_j)\neq\eg$, as defined in
% Section~\ref{sec:canonical}.
% Finally, the translations stored in $t_\eh$ are updated:
% $T(\uh_j, \th_\eh) := (T(\uh_k,\th_\eh))^{-1} T(\uh_j, \th_\eh), j=0,1,2$, where $k$ is the index for which $d_k$ is
% minimal. The data structure now gives access to the canonical
% representative of~$t_\eh$. 
%%%%%%%%%%%%%%%%%%%%%%%%%%%%%%%%%%%%%%%%%%%%%%%%%%%%%%
Once this is done for all new faces, temporary translations stored in
vertices can be removed. 

\subsection{Degree of predicates \label{sec:predicates}}

Following the celebrated exact geometric computation
paradigm~\cite{yd-ecp-95}, the correctness of the combinatorial
structure of the Delaunay triangulation relies on the exact evaluation
of predicates. The main two predicates are 
\begin{itemize}
	\item \orientation, which checks whether an input point $\ph$ in $\Do$
	lies on the right side, the left side, or on an oriented Euclidean
	segment. 
	\item \incircle, which checks whether an input point $\ph$ in $\Do$ lies
	inside, outside, or on the boundary of the disk circumscribing an
	oriented triangle.
\end{itemize} 
Input points, which lie in $\Do$, are tested against canonical
triangles of the triangulation, whose vertices are images of input
points by translations in $\N$. If points are assumed to have
rational coordinates,
then evaluating the predicates boils down to determining the sign of
polynomial expressions whose coefficients are lying in some extension
field of the rationals. Proposition~\ref{prop:degreepredicates-highergenus} gives an upper
bound on the degree of these polynomial expressions. For the special case of the
Bolza surface ($g=2$), it improves the previously known upper bound from
72~\cite[Proposition~1]{it-idtbs-17}, which was proved using symbolic computations in $\textsc{Maple}$, to~48. 

\begin{proposition}\label{prop:degreepredicates-highergenus}
	For the generalized Bolza surface of genus $g\geq 2$, the predicates can be evaluated by determining the sign of rational polynomial expressions of total degree at most $12\varphi(4g)\leq 24g$ in the coordinates of the input points, where $\varphi$ is the Euler totient function. 
\end{proposition}

Recall that the Euler totient function $\varphi(n)$ counts the number
of integers up to a given integer $n$ that are relatively prime to
$n$~\cite{long1972}.

\begin{proof}
	We will only consider the \incircle predicate; the strategy
	for determining the maximum degree for the \orientation
	predicate is similar and the resulting maximum degree is
	lower. The \incircle predicate is given by the sign of
	\begin{align*}
		\incircle(\ph_1,\ph_2,\ph_3,\ph_4) &= \begin{vmatrix}
			x_1 & y_1 & x_1^2+y_1^2 & 1 \\
			x_2 & y_2 & x_2^2+y_2^2 & 1 \\
			x_3 & y_3 & x_3^2+y_3^2 & 1 \\
			x_4 & y_4 & x_4^2+y_4^2 & 1 
		\end{vmatrix}\\
		&= x_3y_4(x_1^2+y_1^2)-x_3y_4(x_2^2+y_2^2)+x_2y_4(x_3^2+y_3^2)-x_2y_3(x_4^2+y_4^2)\ldots,
	\end{align*} 
	for points $\ph_j=x_j+y_ji, j=1,\ldots,4,$ in $\mD$. In the second equality,
	we have just written 4 of the 24 terms to illustrate what the
	terms look like. This will be used later in the proof to
	determine their maximum degree.	Since at least one of the four
	points is contained in $\Do$, we will assume without loss of
	generality that $x_1,y_1\in\mQ$. The other points are images
	of points with rational coordinates under some elements of
	$\Gamma_g$. We know that $\Gamma_g$ is generated by $f_k$ for
	$k=0,\ldots,4g-1$. The translation $f_k$ can be represented by
	the matrix $A_k$ (see Equation~\eqref{eq:generator-matrix} in Section~\ref{sec:generalized}). 
	The entries of $A_k$ are contained in the extension field 
	\[
	L=\mQ\left[\zeta_{4g},\sqrt{\cot^2(\tfrac{\pi}{4g})-1} \;\right],
	\]  
	where $\zeta_{4g}=\exp(\tfrac{\pi i}{2g})$ is a primitive $4g$-th root of unity. The field $L$ is an extension field of degree 2 of the cyclotomic field $\mQ[\zeta_{4g}]$, which is an extension field of degree $\varphi(4g)$ of $\mQ$ so the total degree of $L$ as an extension field of $\mQ$ is $2\varphi(4g)$. Later in the proof, we will actually look at the degree of the field $L\cap \R$ over $\mQ$. Because $L\cap\R$ is the fixed field of $L$ under complex conjugation, $L$ is a quadratic extension of $L\cap \R$. Therefore, the degree of $L\cap\R$ as an extension field of $\mQ$ is $\varphi(4g)$. Since each translation in $\Gamma$ can be represented by a product of matrices $A_k$, it follows that for $j=2,3,4$ we can write
	\[ x_j+y_ji=\dfrac{\alpha_j(x_j^c+y_j^ci)+\beta}{\bar{\beta}(x_j^c+y_j^c)+\bar{\alpha}},  \]
	where $x_j^c$ and $y_j^c$ are the (rational) coordinates of the canonical representative of $x_j+y_ji$ and where $\alpha_j$ and $\beta_j$ are elements of $L$. As usual, we can get rid of the $i$ in the denominator by multiplying numerator and denominator by the complex conjugate of the denominator. Because both the numerator and denominator are linear as function of $x_j^c$ and $y_j^c$, we obtain
	\[ x_j+y_ji=\dfrac{P_j(x_j^c,y_j^c)+Q_j(x_j^c,y_j^c)i}{R_j(x_j^c,y_j^c)},  \]
	where $P_j,Q_j$ and $R_j$ are polynomials in $x_j^c$ and
	$y_j^c$ of total degree at most 2 with coefficients in
	$L\cap\R$. Note that we indeed know that the coefficients are
	real numbers, since by construction we have already split the
	real and imaginary parts. Hence, suppressing the dependencies
	on $x_j^c$ and $y_j^c$, we see that
	\begin{eqnarray*}
		\lefteqn{\incircle(\ph_1,\ph_2,\ph_3,\ph_4)}\\
		&=& \dfrac{P_3Q_4(x_1^2+y_1^2)}{R_3R_4}-\dfrac{P_3Q_4(P_2^2+Q_2^2)}{R_2^2R_3R_4}+\dfrac{P_2Q_4(P_3^2+Q_3^2)}{R_2R_3^2R_4}-\dfrac{P_2Q_3(P_4^2+Q_4^2)}{R_2R_3R_4^2}+\ldots,\\
		&=& \dfrac{R_2^2P_3R_3Q_4R_4(x_1^2+y_1^2)-P_3R_3Q_4R_4(P_2^2+Q_2^2)+\ldots}{R_2^2R_3^2R_4^2}
	\end{eqnarray*}
	Now, testing whether $\incircle(\ph_1,\ph_2,\ph_3,\ph_4)>0$ amounts to testing whether \[R_2^2P_3R_3Q_4R_4(x_1^2+y_1^2)-P_3R_3Q_4R_4(P_2^2+Q_2^2)+\ldots\;>\;R_2^2R_3^2R_4^2.\]
	Since all $P_j,Q_j$ and $R_j$ are polynomials in $x_j^c$ and $y_j^c$ of degree at most 2, this reduces to evaluating a polynomial of total degree at most 12 in the coordinates of the input points, with coefficients in $L\cap\R$. Because $L\cap\R$ is an extension field of $\mQ$ of degree $\varphi(4g)$, we conclude that evaluating $\incircle(\ph_1,\ph_2,\ph_3,\ph_4)$ amounts to determining the sign of a polynomial of total degree at most $12\varphi(4g)$ with rational coefficients. To prove $12\varphi(4g)\leq 24g$, we write $g=2^k g'$ where $g'$ is odd. Then
	\[ \varphi(4g)=\varphi(2^{k+2}g')=\varphi(2^{k+2})\varphi(g')=(2^{k+2}-2^{k+1})\varphi(g')=2^{k+1}\varphi(g').\]
	If $g'=1$, then $\varphi(g')=1$, so $\varphi(4g)=2^{k+1}=2g$. If $g'> 1$, then $\varphi(g')\leq g'-1$, so $\varphi(4g)\leq 2^{k+1}(g'-1)\leq 2(g-1)$. Hence, in both cases $\varphi(4g)\leq 2g$. This finishes the proof.  
\end{proof}

\subsection{Implementation and experimental results \label{sec:implementation}}

The algorithm presented in Section~\ref{sec:computingDT} was
implemented in C++, with the data structure described in Section~\ref{sec:data-structure}.  The preprocessing step consists in computing dummy points
that serve for the initialization of the data structure, following the
two options presented in Section~\ref{sec:combinatoricsofsmallDT}. The
implementation also uses the value of the systole given by
Theorem~\ref{thm:systolevalue}.

Let us continue the discussion on
predicates. 
In practice, the implementation relies on the \expr\ number type
\cite{core:library}, which provides us with exact and filtered
computations. As for the computation of dummy points
(Section~\ref{sec:dummy-exp}), the evaluation exceeds the
capabilities of \core\ for genus bigger than~2, due to the barriers
raised by their very high algebraic degree, so, only a non-robust
implementation of the algorithm can be obtained.

The rest of this section is devoted to the
implementation for the Bolza surface, for which a fully robust
implementation has been
integrated in \cgal~\cite{cgal-it-p4ht2-perm}.  All details can be found in Iordanov's PhD
thesis~\cite{iordanov-tel-02072155}. We only mention a few key points
here. 

To avoid
increasing further the algebraic degree of predicates, the coordinates
of dummy points are rounded to rationals (see
Table~\ref{table:dummy-points}). We have checked that the validity
condition~\eqref{condition} still holds for the rounded points, and
that the combinatorics of the Delaunay triangulations of exact and
rounded points are identical.
\begin{table}[htb]
	\begin{center}
		\caption{Exact and rational expressions for the dummy
			points for the Bolza surface. The midpoint of side~$\h{s}_j$ of the
			fundamental domain is denoted as~$\h{m}_j$. The midpoint of
			segment~$[0,\vh_j]$ is denoted as~$\ph_j$.}
		\begin{tabular}{l|l|l}
			\textbf{Point}  & \textbf{Expression} & \textbf{Rational approximation} \\ \hline
			$\vh_0$           & $\br{ \frac{{2}^{3/4}\sqrt {2+\sqrt {2}}}{4}, -\frac{{2}^{3/4}\sqrt {2-\sqrt {2}}}{4} } $ & 
			$\br{ 97/125, -26/81 }$  \\ 
			$\h{m}_4$           & $\br{-\sqrt{\sqrt{2}-1}, 0 }$ &
			$\br{-9/14, 0}$  \\
			$\h{m}_5$           & $\br{-\frac{\sqrt{2}\sqrt{\sqrt{2}-1}}{2}, -\frac{\sqrt{2}\sqrt{\sqrt{2}-1}}{2} }$ &
			$\br{-5/11, -5/11}$  \\
			$\h{m}_6$           & $\br{ 0, -\sqrt{\sqrt{2}-1} }$ &
			$\br{ 0, -9/14 }$  \\
			$\h{m}_7$           & $\br{ \frac{\sqrt{2}\sqrt{\sqrt{2}-1}}{2}, -\frac{\sqrt{2}\sqrt{\sqrt{2}-1}}{2} }$ &
			$\br{ 5/11, -5/11}$  \\
			$\ph_0$           & $\br{ \frac{2^{1/4} \sqrt{2+\sqrt{2}}}{2\sqrt{2} + 2\sqrt{2-\sqrt{2}}}, -\frac{2^{1/4} \sqrt{2-\sqrt{2}}}{2\sqrt{2} + 2\sqrt{2-\sqrt{2}}} }$ &
			$\br{ 1/2, -4/19 }$ \\
			$\ph_1$           & $\br{ \frac{2^{3/4} \br{\sqrt{2+\sqrt{2}} + \sqrt{2-\sqrt{2}}}}{4\sqrt{2} + 4\sqrt{2-\sqrt{2}}}, \frac{2^{3/4} \br{\sqrt{2+\sqrt{2}} - \sqrt{2-\sqrt{2}}}}{4\sqrt{2} + 4\sqrt{2-\sqrt{2}}} }$ &
			$\br{ 1/2, 4/19 }$  \\
			$\ph_2$           & $\br{ \frac{2^{1/4} \sqrt{2-\sqrt{2}}}{2\sqrt{2} + 2\sqrt{2-\sqrt{2}}},  \frac{2^{1/4} \sqrt{2+\sqrt{2}}}{2\sqrt{2} + 2\sqrt{2-\sqrt{2}}} }$ &
			$\br{ 4/19, 1/2 }$ \\
			$\ph_3$           & $\br{ \frac{2^{3/4} \br{\sqrt{2-\sqrt{2}} - \sqrt{2+\sqrt{2}}}}{4\sqrt{2} + 4\sqrt{2-\sqrt{2}}}, \frac{2^{3/4} \br{\sqrt{2+\sqrt{2}} + \sqrt{2-\sqrt{2}}}}{4\sqrt{2} + 4\sqrt{2-\sqrt{2}}} }$ &
			$\br{-4/19, 1/2 }$ \\
			$\ph_4$           & $\br{-\frac{2^{1/4} \sqrt{2+\sqrt{2}}}{2\sqrt{2} + 2\sqrt{2-\sqrt{2}}},  \frac{2^{1/4} \sqrt{2-\sqrt{2}}}{2\sqrt{2} + 2\sqrt{2-\sqrt{2}}} }$ &
			$\br{-1/2, 4/19 }$ \\
			$\ph_5$           & $\br{-\frac{2^{3/4} \br{\sqrt{2+\sqrt{2}} + \sqrt{2-\sqrt{2}}}}{4\sqrt{2} + 4\sqrt{2-\sqrt{2}}}, \frac{2^{3/4} \br{\sqrt{2-\sqrt{2}} - \sqrt{2+\sqrt{2}}}}{4\sqrt{2} + 4\sqrt{2-\sqrt{2}}} }$ &
			$\br{-1/2, -4/19 }$ \\
			$\ph_6$           & $\br{-\frac{2^{1/4} \sqrt{2-\sqrt{2}}}{2\sqrt{2} + 2\sqrt{2-\sqrt{2}}},  -\frac{2^{1/4} \sqrt{2+\sqrt{2}}}{2\sqrt{2} + 2\sqrt{2-\sqrt{2}}} }$ &
			$\br{-4/19, -1/2 }$ \\
			$\ph_7$           & $\br{ \frac{2^{3/4} \br{\sqrt{2+\sqrt{2}} - \sqrt{2-\sqrt{2}}}}{4\sqrt{2} + 4\sqrt{2-\sqrt{2}}},-\frac{2^{3/4} \br{\sqrt{2-\sqrt{2}} + \sqrt{2+\sqrt{2}}}}{4\sqrt{2} + 4\sqrt{2-\sqrt{2}}} }$ &
			$\br{ 4/19, -1/2 }$
		\end{tabular}
		\label{table:dummy-points}
	\end{center}
\end{table}

Attention has also been paid to the manipulation of translations.  As
seen in Section~\ref{sec:data-structure}, translations are composed
during the execution of the algorithm. To avoid performing the same
multiplications of matrices several times, we actually represent a
translation as a word on the elements of $\Z_8$, where $\Z_8$ is
considered as an alphabet and each element corresponds to a generator
of $\Gb$. The composition of two translations corresponds to the
concatenation of their two corresponding words.
Section~\ref{sec:data-structure} showed
that only the finitely many translations in $\Nb$ must be stored in
the data structure. Moreover, words that appear during the various
steps of the algorithm can be reduced by Dehn's
algorithm~\cite{Dehn1912,g-dawp-60}, yielding a finite number of words
to be stored, so, a map can be used to associate a matrix to each word.
Dehn's algorithm terminates in a finite number of steps and its time
complexity is polynomial in the length of the input word. From
Sections~\ref{sec:canonical} and~\ref{sec:data-structure}, words to be reduced
are formed by the concatenation of two or three words corresponding to
elements of $\Nb$, whose length is not more than four, so, the longest
words to be reduced have length~12.

Running times have been measured on a MacBook
Pro (2015) with processor Intel Core i5, 2.9 GHz, 16 GB and 1867 MHz
RAM, running MacOS X (10.10.5). The code was compiled with
clang-700.1.81. 
We generate 1~million points in the half-open octagon $\Db$ and construct four triangulations: 
\begin{itemize}\cramped
	\item a \cgal\ Euclidean Delaunay triangulation with \texttt{double} as number type. 
	\item a \cgal\ Euclidean Delaunay triangulation with \expr as number type, 
	\item our Delaunay triangulation of the Bolza with \texttt{double} as number type,
	\item our Delaunay triangulation of the Bolza surface with \expr as number type,
\end{itemize}
Note that the implementations using \texttt{double} are not robust and
are only considered for the purpose of this experimentation. The
insertion times are averaged over~10 executions. The results are
reported in Table~\ref{tbl:runtimes}.
\begin{table}[htb]
	\centering
	\begin{tabular}{l|l}
		& Runtime (in seconds) \\ \hline
		Euclidean DT (\texttt{double}) &  1    \\
		Euclidean DT (\expr)           &  24 \\
		Bolza DT (\texttt{double})    &  16 \\
		Bolza DT (\expr)              &  55 \\
	\end{tabular}
	\caption{Runtimes for the computation of Delaunay triangulations of 1
		million random points in the half-open octagon $\Db$.}
	\label{tbl:runtimes}
\end{table}

The experiments confirm the influence of the algebraic demand for the
Bolza surface: almost two thirds of the runnning time is spent in predicate
evaluations. Also, it was observed that only 0.76\% calls to
predicates involve translations in $\Nb$, but these calls account for
36\% of the total time spent in predicates. 

Note also that the triangulation can quickly be cleared of dummy points:
in most runs, all dummy points are removed from the triangulation
after the insertion of~30 to~70 points. 

%%%%%%%%%%%%%%%%%%%%%%%%%%%%%%%%%%%%%%%%%%%%%%%%%%%%%%%%%%%%%

\section{Conclusion and open problems}
We have extended Bowyer's algorithm to the computation of Delaunay triangulations of point
sets on generalized Bolza surfaces, a particular type of hyperbolic surfaces.
A challenging open problem is the
generalization of our algorithm to arbitrary hyperbolic surfaces. 

One of the main ingredients of our extension of Bowyer's algorithm is the validity condition~\eqref{condition}, and to be able to say whether it holds or not we need to know the value of the systole of the hyperbolic surface. For general hyperbolic surfaces an explicit value, or a `reasonable' lower bound of
the systole, is not known, and there are no efficient
algorithms to compute or approximate it. The effective procedure presented
in~\cite{Akrout2006} is based on the construction of a pants decomposition of a hyperbolic
surface, and computes the systole from the Fenchel-Nielsen coordinates associated with this
decomposition. However, the complexity of this algorithm does not seem to be known, and it is not
clear how to turn this method into an efficient and robust algorithm.

If the systole is known, then it seems that we can use the refinement algorithm presented in Section~\ref{sec:simplealgorithm} to compute a dummy point set satisfying the validity condition. However, in the case of generalized Bolza surfaces it is sufficient to consider only the translates of vertices in $D_{\N}$ by Proposition~\ref{prop:dummy-points-inclusion}, whereas it is not clear how many translates are needed for an arbitrary hyperbolic surface.

A more modest attempt towards generalization could focus on hyperbolic surfaces
represented by a `nice' fundamental polygon. Hyperelliptic surfaces have a point-symmetric
fundamental polygon (See~\cite{SchmutzSchaller1999}), so these surfaces are obvious
candidates for future work.

%%%%%%%%%%%%%%%%%%%%%%%%%%%%%%%%%%%%%%%%%%%%%%%%%%%%%%%%%%%%%

%%%%%%%%%%%%%%%%%%%%%%%%%%%%%%%%%%%%%%%%%%%%%%%%%%%%%%%%%%%%%

\appendix

\section{Statement and proof of Lemma~\ref{lem:triangleareabound}}\label{sec:appendixtriangulations}

\begin{lemma}\label{lem:triangleareabound}
	Let $T$ be a hyperbolic triangle with a circumscribed disk of radius $R$. Then
	\[\area(T)\leq \pi-6\arcot(\sqrt{3}\cosh(R)).\]
\end{lemma}

Lemma~\ref{lem:triangleareabound} is the special case $m=3$ of the following lemma. A proof was given in Ebbens's master's thesis~\cite{ebbens2017}, but for completeness we have included it here as well. 

\begin{lemma}\label{maxarea}
	Let $P$ be a convex hyperbolic $m$-gon for $m\geq 3$ with all vertices on a circle with radius $R$. Then the area of $P$ attains its maximal value $A(R)$ if and only if $P$ is regular and in this case
	\[ \cosh R=\cot\left(\dfrac{\pi}{m} \right)\cot\left( \dfrac{(m-2)\pi-A(R)}{2m}\right).\]
\end{lemma}

\begin{proof}
	A lower bound for the circumradius of a polygon given the area of the polygon is given in the literature~\cite{Naatanen1982}. We use the same approach to prove Lemma~\ref{maxarea}. 
	
	Consider $m=3$. Divide $P$ into three pairs of right-angled triangles with angles $\theta_j$ at the center of the circumscribed circle, angles $\alpha_j$ at the vertices and right angles at the edges of $P$ (see Figure~\ref{trianglecirlce}). 
	
	\begin{figure}[ht]
		\centering
		\includegraphics[width=0.5\textwidth]{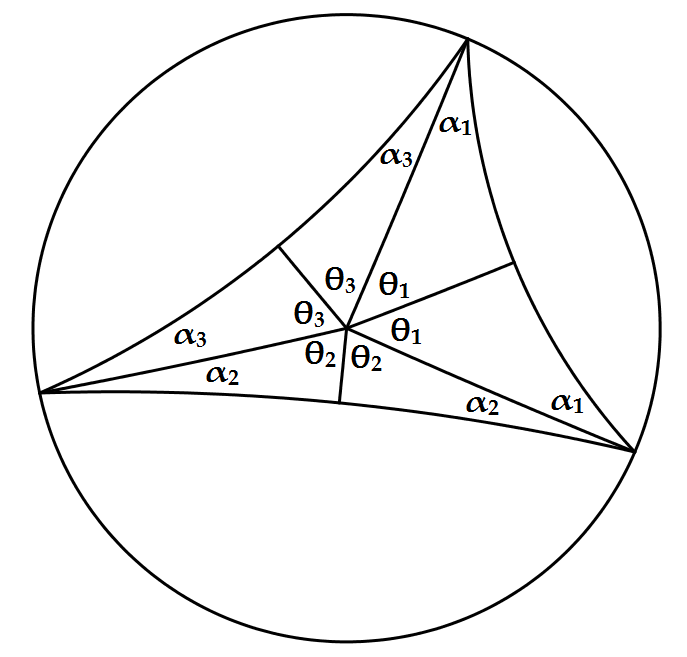}
		\caption{Division of $P$ into three pairs of right-angled triangles}
		\label{trianglecirlce}
	\end{figure}
	
	By the second hyperbolic cosine rule
	$$ \cosh R=\cot\theta_j\cot\alpha_j$$
	for $j=1,2,3$. Furthermore, $\sum_{j=1}^{3}\theta_j=\pi$ and $A=\pi-2\sum_{j=1}^{3}\alpha_j$. Therefore, maximizing $A$ reduces to minimizing
	\begin{equation}\label{minangles}
		f(\theta_1,\theta_2,\theta_3)=\sum_{j=1}^{3}\arcot(\cosh R\tan\theta_j)
	\end{equation}
	subject to the constraints $\sum_{j=1}^{3}\theta_j=\pi$ and $0\leq\theta_j<\pi$, i.e., minimizing~\eqref{minangles} over the triangle in $\R^3$ with vertices $(\pi,0,0),(0,\pi,0),(0,0,\pi)$. We parametrize this triangle as follows
	$$ \theta_1=s+t,\theta_2=s-t,\theta_3=\pi-2s$$
	for $0<s<\pi/2$ and $|t|\leq s$. By~\eqref{minangles}, we can view $f$ as a function of $s$ and $t$. First, we fix $s$ and minimize over $t$. Then
	\begin{align*}
		\dfrac{\partial}{\partial t}f(\theta_1(s,t),\theta_2(s,t),\theta_3(s,t))&=\sum_{j=1}^{3}\dfrac{-\sec^2\theta_j}{1+\cosh^2R\tan^2\theta_j}\dfrac{\partial \theta_j}{\partial t},\\
		&=\dfrac{\sec^2\theta_2}{1+\cosh^2R\tan^2\theta_2}-\dfrac{\sec^2\theta_1}{1+\cosh^2R\tan\theta_1},\\
		&=\dfrac{1}{1+(\cosh^2R-1)\sin^2\theta_2}-\dfrac{1}{1+(\cosh^2R-1)\sin^2\theta_1}.
	\end{align*} 
	Therefore, a minimum is obtained if and only if $\theta_1=\theta_2$, i.e., if and only if $t=0$. In a similar way, we minimize over $s$.
	\begin{align*}
		\dfrac{\partial}{\partial s}f(\theta_1(s,t),\theta_2(s,t),\theta_3(s,t))&=\sum_{j=1}^{3}\dfrac{-\sec^2\theta_j}{1+\cosh^2R\tan^2\theta_j}\dfrac{\partial \theta_j}{\partial s},\\
		&=\dfrac{2}{1+(\cosh^2R-1)\sin^2\theta_3}-\dfrac{2}{1+(\cosh^2R-1)\sin^2\theta_1},
	\end{align*}
	and it follows that a minimum is obtained for $\theta_1=\theta_3$. Therefore, the area of $P$ obtains its maximal value $A(R)$ if and only if $\theta_1=\theta_2=\theta_3=\pi/3$, i.e., if and only if $P$ is a regular triangle. In this case,
	$$\alpha_1=\alpha_2=\alpha_3=\dfrac{\pi-A(R)}{6},$$
	so 
	$$\cosh(R)=\cot\theta_j\cot\alpha_j=\cot\left(\dfrac{\pi}{3}\right)\cot\left(\dfrac{\pi-A(R)}{6}\right).$$
	For arbitrary $m\geq 3$, the proof that maximal area is obtained for a regular polygon is the same but with more parameters. In this case $\theta_j=\pi/m$ and 
	$$ A(R)=(m-2)\pi-2m\alpha_j,$$
	so the area $A(R)$ of the regular polygon is given by 
	$$ \cosh(R)=\cot\theta_j\cot\alpha_j=\cot\left(\dfrac{\pi}{m}\right)\cot\left(\dfrac{(m-2)\pi-A(R)}{2m}\right).$$
\end{proof}

\section{Proof of Lemma~\ref{lem:separationlowerbound}}
\label{sec:appendixsymmetrichyperbolicsurfaces}

We prove the different properties in the same order as the statement of the lemma.	
\begin{enumerate}
	\item Consider a segment $\gammah_j$ of separation $k\geq 4$. By symmetry of $D_g$, we can assume that $\gammah_j$ is a segment between $\sh_0$ and $\sh_k$. The length of $\gammah_j$ is greater than or equal to the distance between $\sh_0$ and $\sh_k$, which is given as the length of the common orthogonal line segment $\gammah_j^\perp$ between $\sh_0$ and $\sh_k$ (see Figure~\ref{fig:separationlowerbound}).    
	
	\begin{figure}[htbp]
		\centering
		\includegraphics[width=0.9\textwidth]{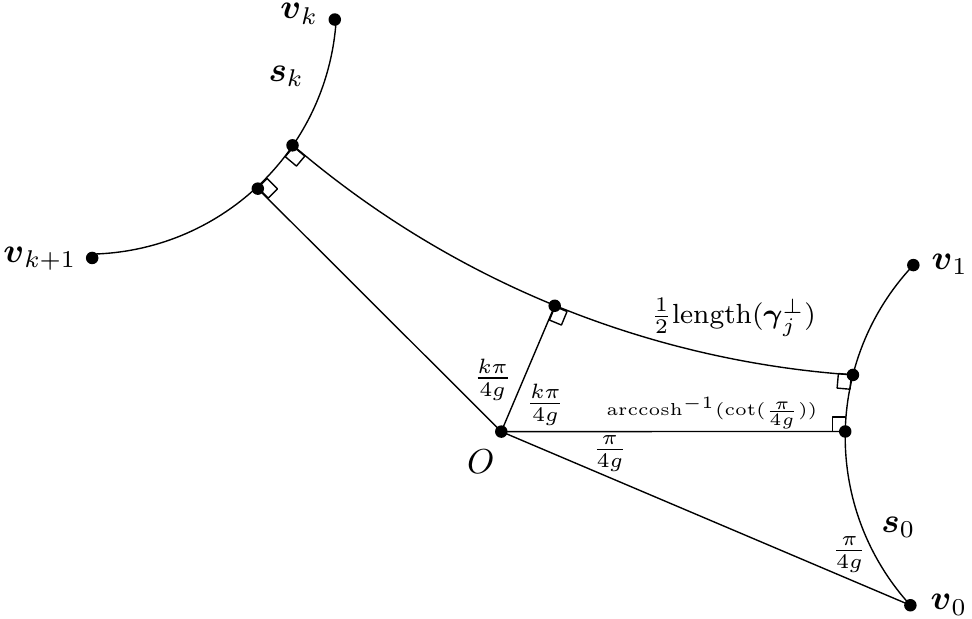}
		\caption{Computing the length of $\gammah_j^\perp$.}
		\label{fig:separationlowerbound}
	\end{figure}
	
	To find an expression for $\hlen{\gammah_j^\perp}$, we draw line segments between the origin $O$ and $\sh_0$ and between $O$ and $\sh_k$. In this way, we obtain a hyperbolic pentagon with four right-angles and remaining angle $\tfrac{k\pi}{2g}$. The line segment from $O$ to $\sh_0$ is a non-hypotenuse side of an isosceles triangle with angles $\tfrac{\pi}{4g},\tfrac{\pi}{4g},\tfrac{\pi}{2}$, as shown in Figure~\ref{fig:separationlowerbound}. Therefore, \cite[Theorem 7.11.3(i)]{Beardon1983} 
	$$ \cosh(d(O,\sh_0))=\dfrac{\cos(\tfrac{\pi}{4g})}{\sin(\tfrac{\pi}{4g})}=\cot(\tfrac{\pi}{4g}).$$
	Drawing a line segment from $O$ orthogonal to $\gammah_j^\perp$, we obtain two quadrilaterals, each of which has three right angles and remaining angle $\tfrac{k\pi}{4g}$. It follows that \cite[Theorem 7.17.1(ii)]{Beardon1983}
	\begin{equation}\label{eq:lengthperp}
		\cosh(\hlen{\gammah_j^\perp})=\cosh(d(O,\sh_0))\sin(\tfrac{k\pi}{4g})=\cot(\tfrac{\pi}{4g})\sin(\tfrac{k\pi}{4g})
	\end{equation}
	The lower bound for the length of a segment of separation at least 4 follows from $\sin(\tfrac{k\pi}{4g})\geq \sin(\tfrac{\pi}{g})$ and a direct computation using properties of trigonometric functions. 
	
	\item Now, consider a segment $\gammah_j$ of separation $k\geq 2$. Using the same argument as in Part~\ref{item:sep4}, wee see that formula~\eqref{eq:lengthperp} still holds. The lower bound for the length of $\gammah_j$ follows from $\sin(\tfrac{k\pi}{4g})\geq
	\sin(\tfrac{\pi}{2g})$ and a direct computation using properties of trigonometric functions.
	
	\item Consider a pair $\gammah_1,\gammah_2$ of consecutive 1-separated segments. By symmetry, we can assume without loss of generality that $\gammah_1$ is a 1-segment between $\sh_0$ and $\sh_1$ (Figure~\ref{fig:proof-1-segments-alternate}). 
	
	\begin{figure}[ht]
		\centering
		\includegraphics[width=.7\textwidth]{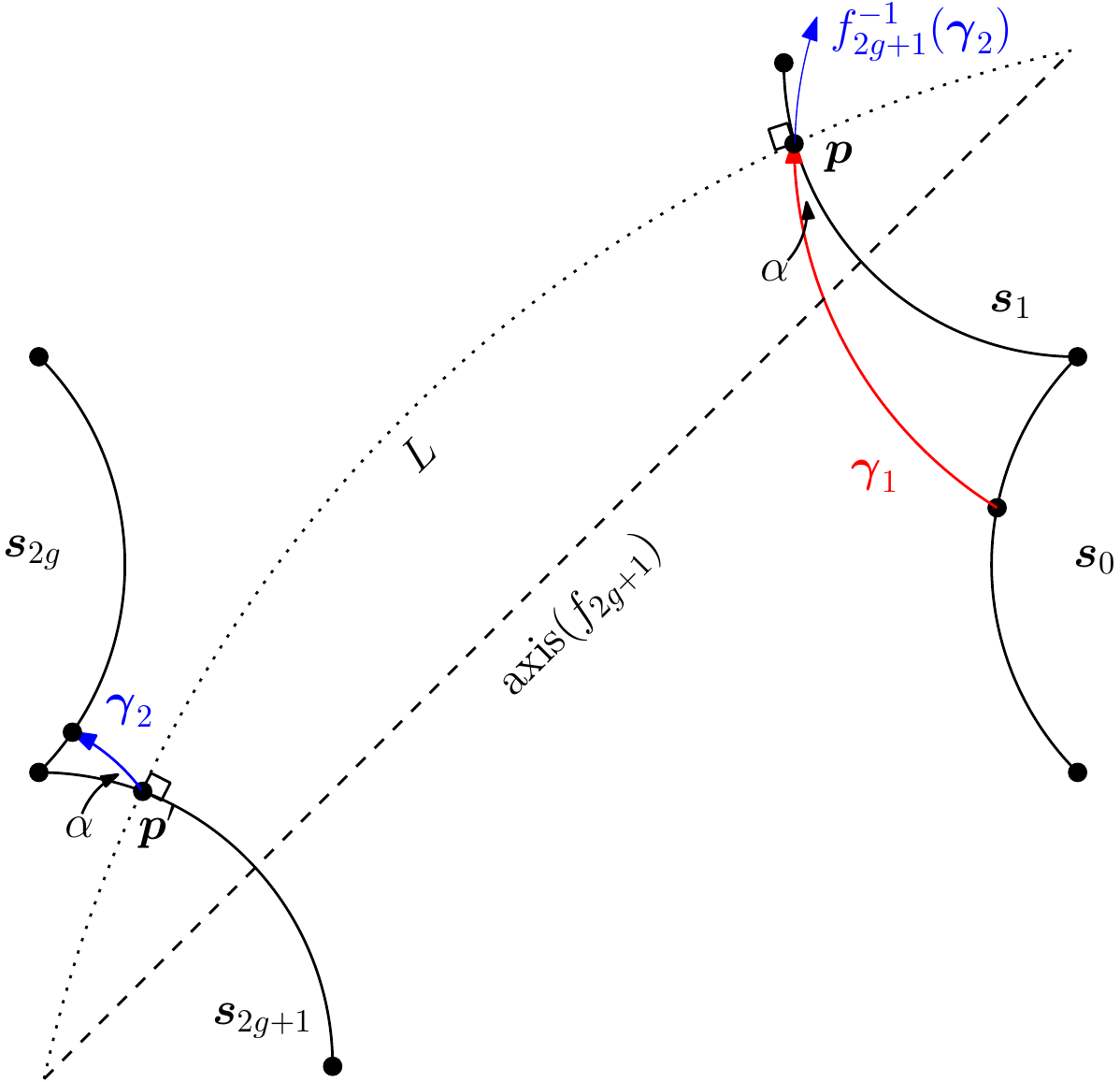}
		\caption{Construction in the proof of Part~\ref{item:sep1alternate} of Lemma~\ref{lem:separationlowerbound}.}
		\label{fig:proof-1-segments-alternate}
	\end{figure}
	
	The side-pairing transformation $f_{2g+1}$ maps the endpoint $\ph$ of $\gammah_1$ to the starting point $\ph'$ of $\gammah_2$. Both $\ph$ and $\ph'$ lie on the same equidistant curve $L$ of the axis of $f_{2g+1}$. The curve $L$ intersects all geodesics that are perpendicular to $\operatorname{axis}(f_{2g+1})$ orthogonally. It can be seen that the angle $\alpha$ between $\gammah_1$ and $\sh_1$ is acute and that $\gammah_1$ and $f_{2g+1}^{-1}(\gammah_2)$ lie on opposite sides of $L$. Moreover, the parts of $\mD$ separated by $L$ are $f_{2g+1}$-invariant (see also Figure~\ref{fig:hyp-transformation} - Right). Hence, $\gammah_1$ and $\gammah_2$ lie on opposite sides of $L$ as well. We conclude that the endpoint of $\gammah_2$ lies on $\sh_{2g}$ so $\gammah_2$ is a $(4g-1)$-segment. 
	
	\item As in Part~\ref{item:sep1alternate}, denote the two segments by $\gammah_1$ and $\gammah_2$. By Part~\ref{item:sep1alternate}, we know that one of $\gammah_1$ and $\gammah_2$ is a 1-segment and the other is a $(4g-1)$-segment. Hence, we can assume without loss of generality that $\gammah_1$ is a $1$-segment between $\sh_0$ and $\sh_1$ and $\gammah_2$ is a $(4g-1)$-segment between $\sh_{2g+1}$ and $\sh_{2g}$ (see Figure~\ref{fig:doublesegmentlowerbound}). Let $x$ be the distance between $\ph_1$ and $\vh_1$ and let $\alpha_1$ be the angle between $\gammah_1$ and $\sh_0$. The distance between $\ph'_1$ and $\vh_{2g+1}$ is $\ell-x$, where the length $\ell$ of the sides satisfies $\cosh(\tfrac{1}{2}\ell)=\cot(\tfrac{\pi}{4g})$. Let $\alpha_2$ be the angle between $\gammah_2$ and $\sh_{2g}$. 
	
	\begin{figure}[htbp]
		\centering
		\includegraphics[width=0.7\textwidth]{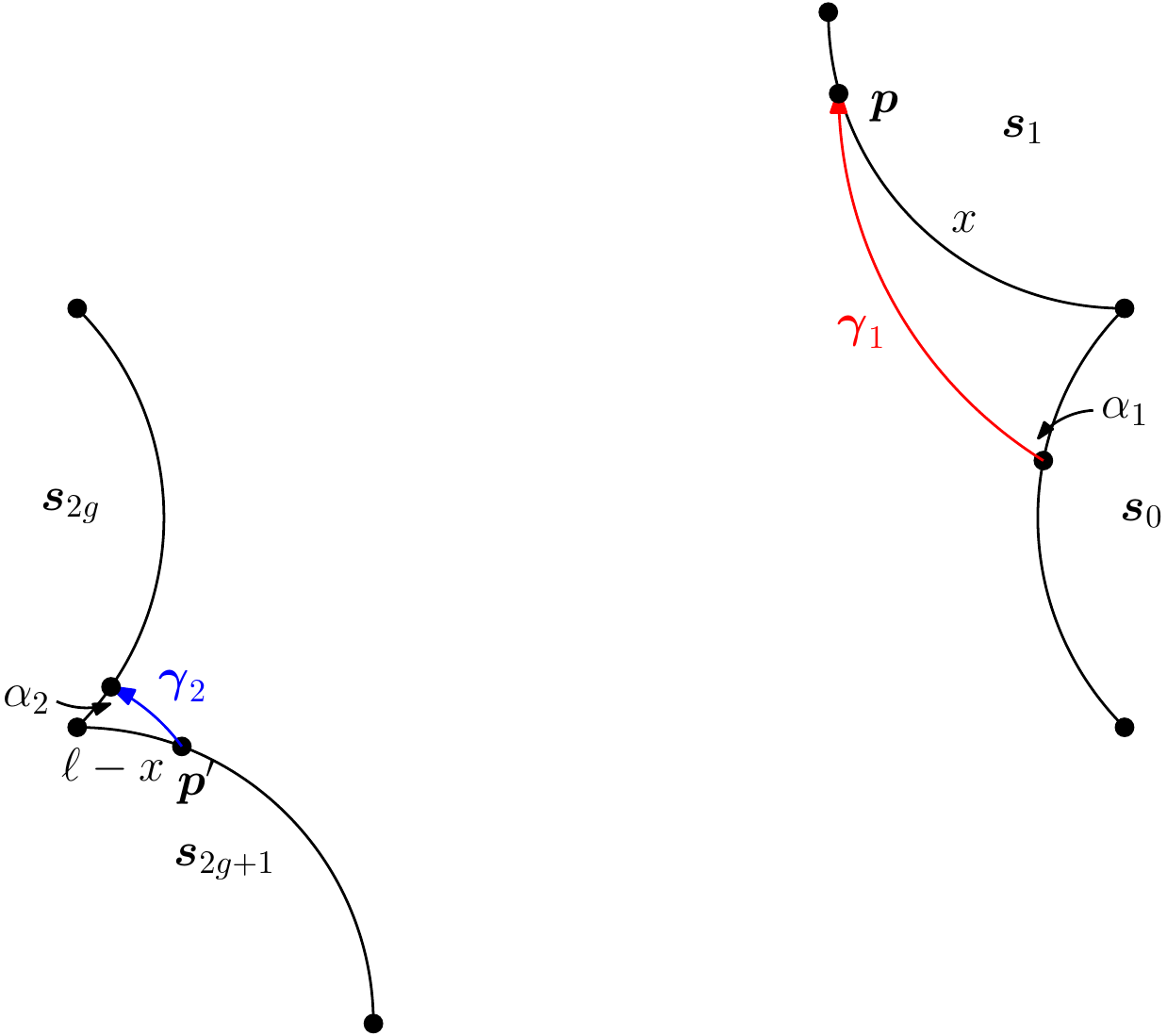}
		\caption{Construction in the proof of Part~\ref{item:sep11} of  Lemma~\ref{lem:separationlowerbound}.}
		\label{fig:doublesegmentlowerbound}
	\end{figure}
	
	By the hyperbolic sine rule,
	$$\dfrac{\sinh(\hlen{\gammah_1})}{\sin(\tfrac{\pi}{2g})}=\dfrac{\sinh(x)}{\sin(\alpha_1)}, $$
	so
	$$ \sinh(\hlen{\gammah_1})\geq \sinh(x)\sin(\tfrac{\pi}{2g}).$$
	In a similar way, we obtain
	$$ \sinh(\hlen{\gammah_2})\geq \sinh(\ell-x)\sin(\tfrac{\pi}{2g}).$$
	We minimize 
	$$ f(x):=\arsinh(\sinh(x)\sin(\tfrac{\pi}{2g}))+\arsinh(\sinh (\ell-x) \sin(\tfrac{\pi}{2g}))$$
	subject to $0< x< \ell$, as this provides a lower bound for $\hlen{\gammah_1\cup\gammah_2}$. Because
	$$ \dfrac{d^2}{dx^2}\;\arsinh(\sinh(x)\sin(\tfrac{\pi}{2g}))=\dfrac{\sin(\tfrac{\pi}{2g})\cos^2(\tfrac{\pi}{2g})\sinh(x)}{(\sin^2(\tfrac{\pi}{2g})\sinh^2(x)+1)^{3/2}}>0$$
	for all $x>0$, the function $x\mapsto \arsinh(\sinh(x)\sin(\tfrac{\pi}{2g}))$ is strictly convex. It follows that $f$ is also strictly convex, so it has a unique global minimum. The derivative of $f$ is given by
	$$ f'(x)=\dfrac{\sin(\tfrac{\pi}{2g})\cosh(x)}{(\sin^2(\tfrac{\pi}{2g})\sinh^2(x)+1)^{1/2}}-\dfrac{\sin(\tfrac{\pi}{2g})\cosh(\ell-x)}{(\sin^2(\tfrac{\pi}{2g})\sinh^2(\ell-x)+1)^{1/2}}.$$
	It is clear that $f'(\tfrac{1}{2}\ell)=0$, so $x=\tfrac{1}{2}\ell$ is the unique minimizer with minimum value $f(\tfrac{1}{2}\ell)=2\arsinh(\sinh(\tfrac{1}{2}\ell)\sin(\tfrac{\pi}{2g}))$. By the discussion above, this implies that
	$$\sinh (\tfrac{1}{2}\hlen{\gammah_1\cup\gammah_2})\geq \sinh(\tfrac{1}{2}\ell)\sin(\tfrac{\pi}{2g}).$$
	Then
	\begin{align*}
		\cosh(\hlen{\gammah_1\cup\gammah_2})&=2\sinh^2(\tfrac{1}{2}\hlen{\gammah_1\cup\gammah_2})+1,\\
		&\geq 2\sinh^2(\tfrac{1}{2}\ell)\sin^2(\tfrac{\pi}{2g})+1,\\
		&= 2(\cot^2(\tfrac{\pi}{4g})-1)\sin^2(\tfrac{\pi}{2g})+1,\\
		&=(1+2\cos(\tfrac{\pi}{2g}))^2,
	\end{align*}
	from which we conclude that $\hlen{\gammah_1\cup\gammah_2}> \tfrac{1}{2}\varsigma_g$.
	
	\item Now, using the notation from Part~\ref{item:sep11}, assume that $\gammah=\gammah_1\cup\gammah_2$. By the argument in Part~\ref{item:sep11}, the length of $\gammah$ is minimal when $\ph_1$ is the midpoint of $\sh_1$ and $\ph'_1$ is the midpoint of $\sh_{2g+1}$, given any location of the starting point of $\gammah_1$ and the endpoint of $\gammah_2$. By symmetry of the argument, it follows that $\hlen{\gammah}$ is minimal when the starting point of $\gammah_1$ is the midpoint of $\sh_0$ and the endpoint of $\gammah_2$ is the midpoint of $\sh_{2g}$. It can be seen that the resulting curve is the curve constructed in the proof of Lemma~\ref{lem:systole}, the length of which is $\varsigma_g$. This finishes the proof.
\end{enumerate}

\section{Proofs omitted in Section~\ref{sec:combinatoricsofsmallDT}}
\label{sec:appendixstructuredalgorithm}

First, we give the proof of Proposition~\ref{prop:dummy-points-inclusion}, which states that for any finite set of points $\dummy$ containing $\W$, each face in $\dth{\proj^{-1}(\dummy)}$ with at least one vertex in $\Do$ is contained in $D_{\N}$.

The proof will use the following lemma.

\begin{lemma}
	Let $C_{\qh}$ be a Euclidean disk centered at $O$ and passing through a point ${\qh}$ (Figure~\ref{fig:proof-gert-fig}). Let $H_1$ and $H_2$ denote the two open half-planes bounded by the Euclidean line through $O$ and ${\qh}$. Let $H_0$ be a half-plane that contains ${\qh}$, bounded by another Euclidean line passing through $O$ but not through ${\qh}$. Let $R_j = \br{H_0 \bigcap H_j}\setminus C_{\qh}, j = 1,2$, and let $\ph_j\in R_j, i=1,2$. The disk $C(\ph_1, \ph_2)$ through $O, \ph_1,$ and $\ph_2$ contains ${\qh}$.
	\label{thm:lemma-gert}
\end{lemma}

\begin{figure}[htbp]
	\centering
	\includegraphics[width=0.46\textwidth]{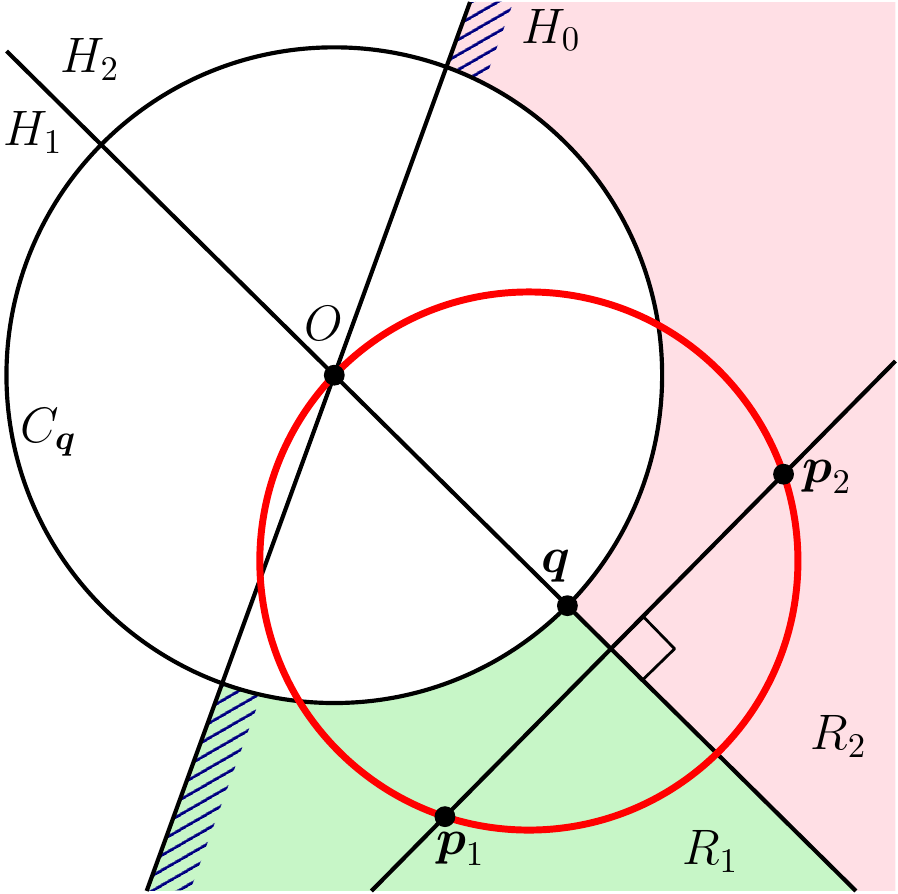} \qquad
	\includegraphics[width=0.46\textwidth]{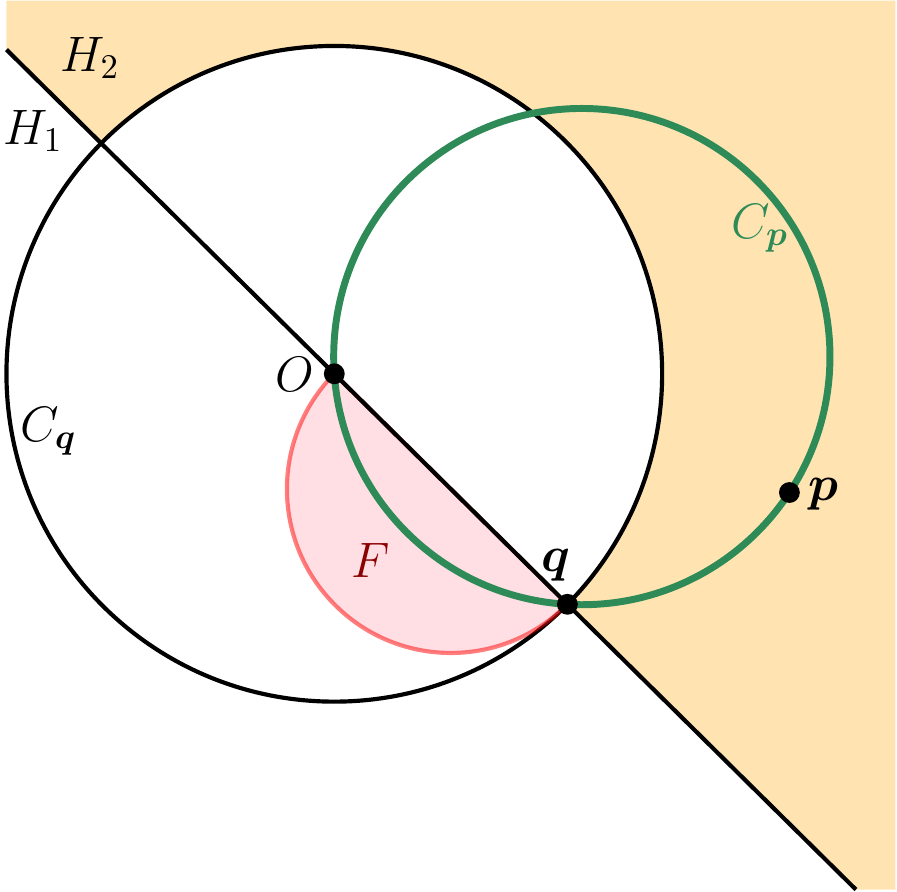}
	\caption[Illustrations for the proof of Lemma~\ref{thm:lemma-gert}]{Illustrations for the proof of Lemma~\ref{thm:lemma-gert}.}
	\label{fig:proof-gert-fig}
\end{figure}

\begin{proof}
	It is easy to verify that there exist pairs of points $(\ph_1, \ph_2)\in R_1\times R_2$ for which the point ${\qh}$ lies inside the disk $C(\ph_1, \ph_2)$. For instance, consider a line perpendicular to the line through $O$ and ${\qh}$ so that ${\qh}$ is closer to $O$ than to their intersection point, as shown in Figure~\ref{fig:proof-gert-fig} - Left. If $\ph_1$ lies on this perpendicular line and $\ph_2$ is the reflection of $\ph_1$ in the line through $O$ and ${\qh}$, then the disk $C(\ph_1, \ph_2)$ contains ${\qh}$. Since this disk varies continuously when $(\ph_1, \ph_2)$ ranges over $R_1\times R_2$, it is sufficient to prove that there are no pairs $(\ph_1^*, \ph_2^*)\in R_1\times R_2$ for which ${\qh}$ lies on the boundary of $C(\ph_1^*, \ph_2^*)$. 
	
	Suppose, for a contradiction, that there exists a pair $(\ph_1^*, \ph_2^*)\in R_1\times R_2$ for which $C(\ph_1^*, \ph_2^*)$ is a disk with ${\qh}$ and $O$ on its boundary. Consider the disk $C_{\qh}$ centered at $O$ and passing through ${\qh}$. Let $F$ be the intersection of the disk with diameter $\seg{O,{\qh}}$ with the half-plane $H_1$, as shown in Figure~\ref{fig:proof-gert-fig} - Right. For any point $\ph\in H_2\setminus C_{\qh}$, the circle $C_{\ph}$ through $O,{\qh},$ and $\ph$ has a non-empty intersection with $H_1$, which is completely included in $F$, so in particular $C(\ph_1^*, \ph_2^*)$ intersects $H_1$ inside the disk with diameter $\seg{O,{\qh}}$. By a symmetric observation, $C(\ph_1^*, \ph_2^*)$ also intersects $H_2$ inside the same disk. Therefore, $C(\ph_1^*, \ph_2^*)$ is the disk with diameter $\seg{O,{\qh}}$. This implies that both $\ph_1$ and $\ph_2$ lie in the disk $C_{\qh}$, which is a contradiction. Therefore, there exists no pair $(\ph_1^*, \ph_2^*)\in R_1\times R_2$ for which $C(\ph_1^*, \ph_2^*)$ has ${\qh}$ on its boundary. This finishes the proof. 
\end{proof}

Note that Lemma~\ref{thm:lemma-gert} can be directly used in the Poincar\'e disk because hyperbolic circles are represented as Euclidean circles, and hyperbolic geodesics through the origin $O$ are supported by Euclidean lines.

We can now proceed with the proof of Proposition~\ref{prop:dummy-points-inclusion}.

\begin{proof}
	\begin{figure}[htp]
		\centering
		\includegraphics[width=0.48\textwidth]{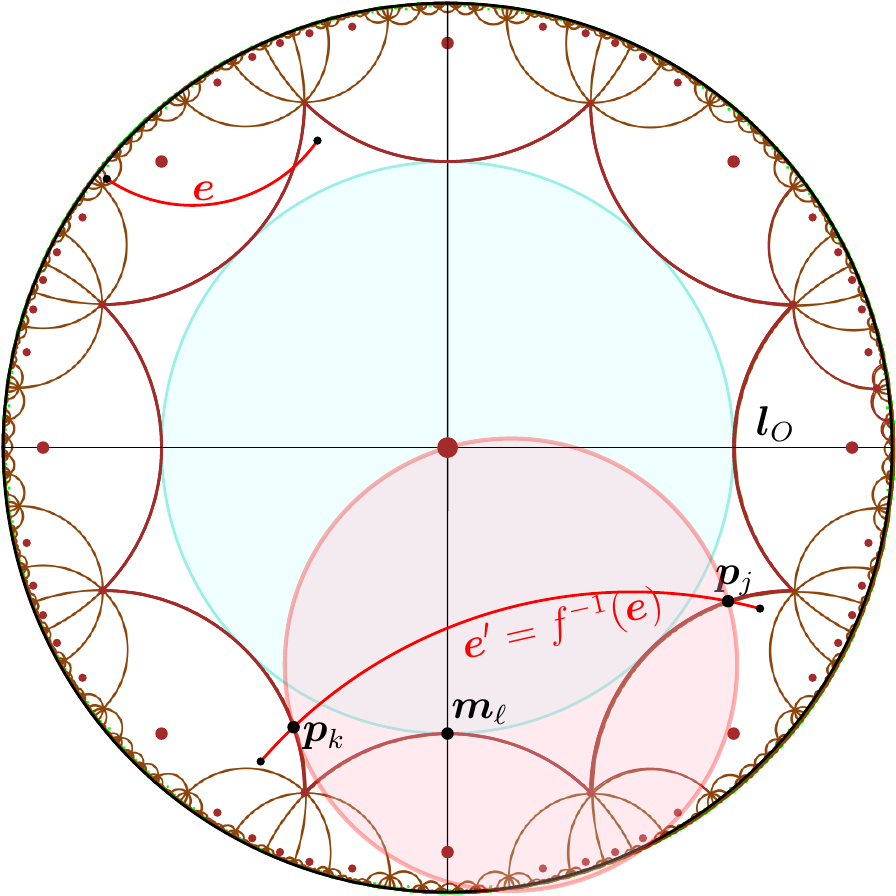}
		\includegraphics[width=0.48\textwidth]{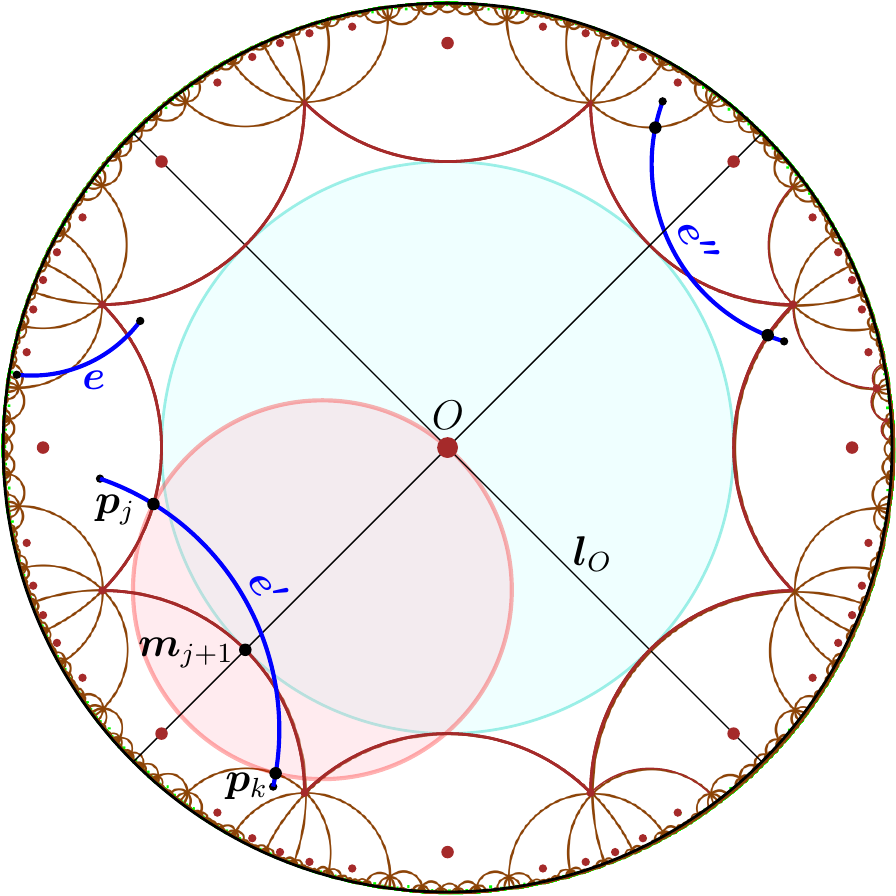}
		\caption[Illustration for the proof of Proposition~\ref{prop:dummy-points-inclusion}]{
			Illustration for the proof of Proposition~\ref{prop:dummy-points-inclusion} for $g = 2$.
		}
		\label{fig:proof-dummy-pts}
	\end{figure}
	
	We show that each edge in $\dth{\proj^{-1}(\dummy)}$ with one endpoint in $\Do$ has its other endpoint inside $D_{\N}$. 
	
	Let $\eh$ be a segment with an endpoint in $\Do$ and an endpoint outside $D_{\N}$. We will prove that every disk passing through the endpoints of $\eh$ contains a point in $\W$. There are two cases to consider: $\eh$ either crosses only one image of $\Do$ under an element of $\N\setminus\{\eg\}$, or it crosses several of its images. We examine each case separately.

	\paragraph*{Case A: The edge $\eh$ crosses only one image of $\Do$ before leaving $D_{\N}$.}
	This case is illustrated in Figure~\ref{fig:proof-dummy-pts} - Left. Let $f(\Do), f\in\N\setminus\{\eg\}$ be the Dirichlet region that $\eh$ crosses. The image $\eh' = \inv{f}(\eh)$ of $\eh$ then crosses $\Do$, intersecting two of its non-adjacent sides $\sh_j$ and $\sh_k$ in the points $\ph_j$ and $\ph_k$, respectively. We can assume without loss of generality that the hyperbolic segment $[\ph_j,\ph_k]$ does not contain the origin, since in that case any disk through $\ph_j$ and $\ph_k$ clearly contains the origin. Then, there exists a line $\lh_O$ through $O$ such that $\ph_j$ and $\ph_k$ are contained in the same half-space. Let $\mh_\ell$ be the midpoint of a side between $\sh_j$ and $\sh_k$ in the same half-space as $\ph_j$ and$\ph_k$. Consider the disk centered at $O$ that passes through $\mh_\ell$ (and, of course, through all the other midpoints $\mh_k$ as well), and consider also the line through $O$ and $\mh_\ell$. By Lemma~\ref{thm:lemma-gert}, the disk $C(\ph_j, \ph_k)$ passing through $O, \ph_j,$ and $\ph_k$ contains $\mh_\ell$. Since $O$ and $\mh_\ell$ are on both sides of the segment $\seg{\ph_j, \ph_k}$, any disk through $\ph_j$ and $\ph_k$ contains either $\mh_\ell$ or $O$, therefore there is no empty disk that passes through $\ph_j$ and $\ph_k$. Assume now that there is an empty disk that passes through the endpoints of $\eh'$. This empty disk can then be shrunk continuously so that it passes through $\ph_j$ and $\ph_k$. The shrunk version of the disk must be also empty, which is a contradiction. Therefore, there is no empty disk passing through the endpoints of $\eh'$, which implies that $\eh'$ (and, by consequence, $\eh$) cannot be an edge in $\dth{\proj^{-1}(\dummy)}$.
	
	\paragraph*{Case B: The edge $\eh$ crosses several images of $\Do$ before leaving $D_{\N}$.}
	This case is illustrated in Figure~\ref{fig:proof-dummy-pts} - Right. There exist multiple images of $\eh$ in $\dth{\proj^{-1}(\dummyq)}$ that intersect $\Do$, in fact as many as the number of Dirichlet regions it intersects. Each one of these images intersects two adjacent sides of $\Do$. Let $\eh'$ be an image of $\eh$ that intersects two adjacent sides $\sh_j$ and $\sh_{j+1}$ of $\Do$ so that the hyperbolic line supporting $\eh'$ separates $O$ and the midpoint $\mh_{j+1}$. Note that such an image of $\eh$ exists always: $\eh$ either separates $O$ and the midpoint $\mh_{j+1}$, or it separates an image of $O$ under some translation $f$ of $\Gg$ and $\mh_{j+1}$; in the second case, $\inv{f}(\eh)$ separates $O$ and the midpoint $\inv{f}(\mh_{j+1})$.  The edge $\eh'$ intersects also the side $\sh_k$ adjacent to $\sh_{j+1}$ in the Dirichlet region that shares the side $\sh_{j+1}$ with $\Do$ (see Figure~\ref{fig:proof-dummy-pts} - Right). Let $\ph_j$ and $\ph_k$ be the intersection points of $\eh'$ with $\sh_j$ and $\sh_k$, respectively. Consider the circle centered at the origin that passes through $\mh_{j+1}$. Consider also the line through $O$ and $\mh_{j+1}$ and the line $\lh_0$ through $O$ perpendicular to it. By Lemma~\ref{thm:lemma-gert}, the disk $C(\ph_j, \ph_k)$ passing through $O, \ph_j,$ and $\ph_k$ contains $\mh_{j+1}$. Since $O$ and $\mh_{j+1}$ are on both sides of the segment $\seg{\ph_j, \ph_k}$, any disk through $\ph_j$ and $\ph_k$ contains either $\mh_{j+1}$ or $O$, therefore there is no empty disk that passes through $\ph_j$ and $\ph_k$. By the same reasoning as in Case A, there is no empty disk passing through the endpoints of $\eh'$ either, which implies that $\eh'$ (and, by consequence, $\eh$) cannot be an edge in $\dth{\proj^{-1}(\dummy)}$.
	
	\medskip
	
	In conclusion, no edge of $\dth{\proj^{-1}(\dummy)}$ can have an endpoint in $\Do$ and an endpoint outside $D_{\N}$, therefore all faces with at least one vertex in $\Do$ are included in $D_{\N}$.
\end{proof}

Second, we compute the distance between any pair of Weierstrass points, as mentioned in the proof of Theorem~\ref{thm:simplealgorithm}.

\begin{lemma}\label{lem:distanceweierstrasspoints}
	The distance between any pair of distinct Weierstrass points of $\Mg$ is at least $\tfrac{1}{4}\sysg$.
\end{lemma}

\begin{proof}
	Recall that the Weierstrass points of $\Mg$ are represented by the origin, the vertices and the midpoints of the sides of $D_g$. The distance between midpoints of non-consecutive sides is clearly larger than the distance between midpoints of consecutive sides. Hence, by symmetry, it is sufficient to consider the region in Figure~\ref{fig:hyperbolictrigonometry}.
	
	\begin{figure}[htp]
		\centering
		\includegraphics[width=.5\textwidth]{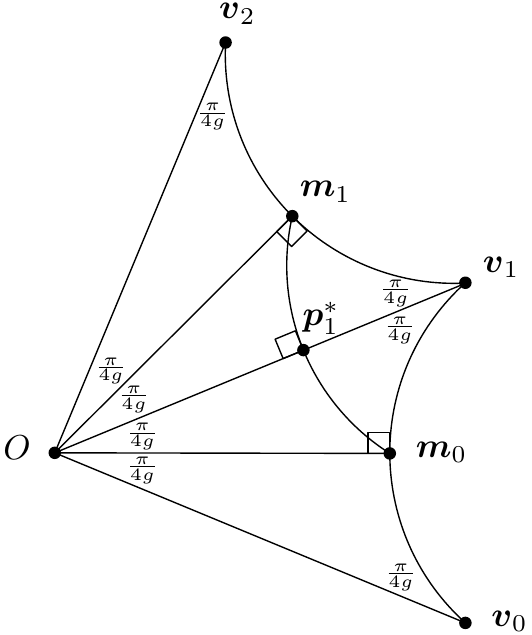}
		\caption{Computing the distances between Weierstrass points of $\Mg$.}
		\label{fig:hyperbolictrigonometry}
	\end{figure}
	
	Let $\mh_k$ be the midpoint of side $\sh_k$ and let $\ph_k^*$ be the midpoint of $\mh_{k-1}$ and $\mh_k$. First, in the right-angled triangle $[O,\mh_0,\vh_0]$ we know that~\cite[Theorem 7.11.3]{Beardon1983}
	\begin{align*}
		\cosh(d(O,\mh_0))=\cosh(d(\mh_0,\vh_0))&=\dfrac{\cos(\tfrac{\pi}{4g})}{\sin(\tfrac{\pi}{4g})}=\cot(\tfrac{\pi}{4g}),\\
		\cosh(d(O,\vh_0))&=\cot^2(\tfrac{\pi}{4g}).
	\end{align*}
	Second, because $d(O,\mh_0)=d(\mh_0,\vh_1)$, $[O,\mh_0,\ph_1^*]$ and $[\vh_1,\mh_0,\ph_1^*]$ are congruent triangles. In particular, $\angle(O\mh_0\ph_1^*)=\angle(\vh_1\mh_0\ph_1^*)=\tfrac{\pi}{4}$. It follows that
	\begin{equation}\label{eq:midpointsdistance}
		\cosh(\tfrac{1}{2}d(\mh_0,\mh_1))=\cosh(d(\mh_0,\ph_1^*))=\dfrac{\cos(\tfrac{\pi}{4g})}{\sin(\tfrac{\pi}{4})}= \sqrt{2}\cos(\tfrac{\pi}{4g}).
	\end{equation}
	The above formulas yield expressions for $d(O,\mh_0),d(O,\vh_0)$ and $d(\mh_0,\mh_1)$ and comparing these to the expression for $\sysg$ (see Theorem~\ref{thm:systolevalue}) yields the result.
\end{proof}

\section{Structured algorithm}\label{sec:structuredalgo}\label{sec:appendixstructuredalgorithmdescription}

Like the symmetric algorithm, this algorithm respects the $4g$-fold symmetry of the Dirichlet region of $\Mg$. Before we give the algorithm in pseudocode, we first explain the idea and the notation. See Figure \ref{fig:pizzaslice} for an illustration of the dummy points within one slice of the $4g$-gon.

\begin{figure}[htbp]
	\centering
	\includegraphics[width=\textwidth]{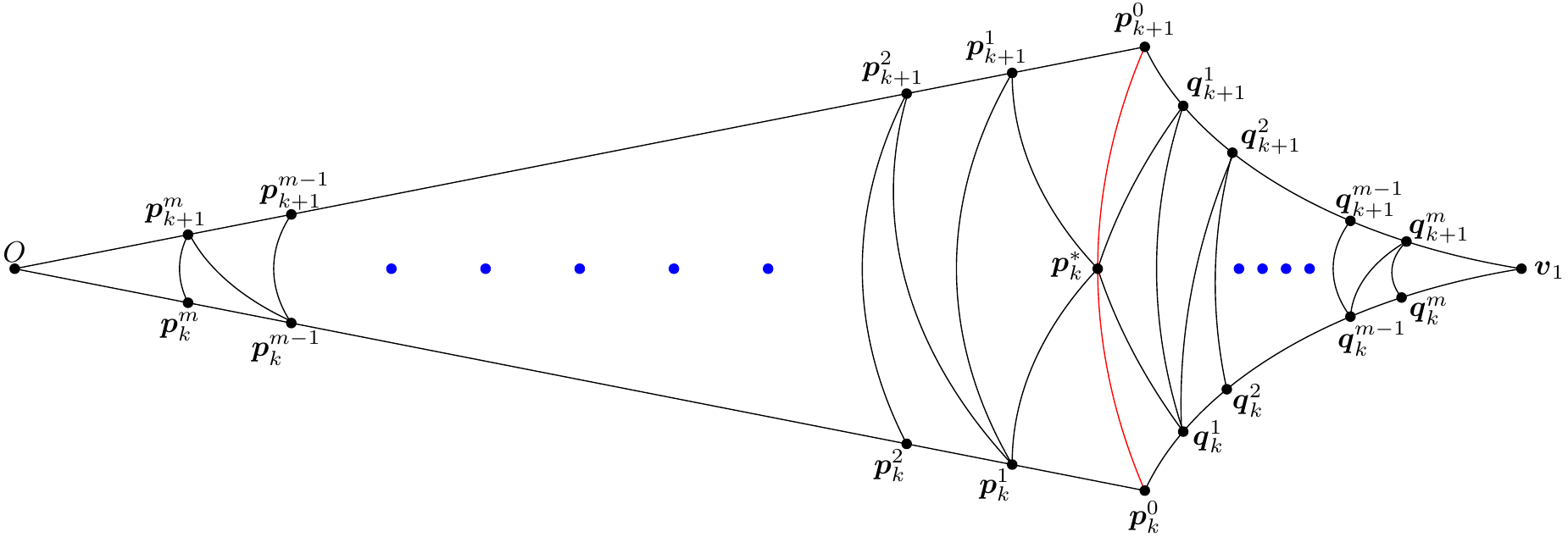}
	\caption{Dummy points within one slice of the $4g$-gon.}
	\label{fig:pizzaslice}
\end{figure}

\begin{enumerate}
	\item As in the other algorithms, initially $\dummy$ consists of the Weierstrass points of $\Mg$: the origin, the vertex and the midpoints of the sides. As usual, the origin and the vertices of $D_g$ are denoted by $O$ and $\vh_k$ respectively. The midpoint of side $\sh_k$ is denoted by $\ph^0_k$. Because the sides are paired to obtain $\Mg$, we only consider $k=0,\ldots,2g-1$ to avoid that several points are actually the same on the surface. Just as the side $\sh_k$ is obtained from $\sh_0$ by rotating it around the origin by angle $\tfrac{k\pi}{2g}$, so $\ph^0_k$ is obtained from $\ph^0_0$ by rotating it in the same way. Hence, we use the lower index as `rotation' index in the definition of the other points as well.
	\item Secondly, the projections $\proj(\ph^*_k)$ of the midpoints $\ph^*_k$ of the geodesic segment in $\H^2$ connecting consecutive pairs $(\ph^0_k,\ph^0_{k+1})$ of midpoints are added to $\dummy$. These points are unique in the sense that they will be the only points in $\dummy$ that have their pre-image in $\Do$ not on some line segment $[O,\ph^0_k]$. This is the reason why they have a star as superscript.
	\item Thirdly, points $\ph^j_0$ are consecutively added on $[O,\ph^{j-1}_0]$ in such a way that the distance between consecutive points $\ph^{j-1}_0$ and $\ph^j_0$ is given by $d(\ph^j_0,\ph^{j-1}_0)=\tfrac{1}{4}\sysg$, until $d(\ph^j_0,O)\leq \tfrac{1}{4}\sysg$. By rotating the points in the same way as before we obtain the points $\ph^j_k$. The projections $\proj(\ph^j_k)$ are added to $\dummy$. Here, the upper index denotes the `iteration' index. Notice that the midpoints of the sides were denoted $\ph^0_k$ to initialize this process. Since $d(O,\ph^0_0)=\arcosh(\cot(\tfrac{\pi}{4g}))$, this step consists of 
	$$ m=\left\lceil \dfrac{\arcosh(\cot(\tfrac{\pi}{4g}))}{\tfrac{1}{4}\sysg}\right\rceil-1$$ 
	iterations.
	\item Finally, observe that the triangles $[O,\ph^0_k,\ph^0_{k+1}]$ and $[\vh,\ph^0_k,\ph^0_{k+1}]$ are congruent under reflection in $[\ph^0_k,\ph^0_{k+1}]$. To establish the same congruence in $\dummy$, we want to apply this reflection to all points in $[O,\ph^0_k,\ph^0_{k+1}]$ that are currently in $\proj^{-1}(\dummy)\cap D_g$. However, if we do so directly, we obtain several pairs of points projecting to the same point on $\Mg$. To avoid this, we reflect the points $\ph^j_k$ in one half of the fundamental polygon across the line through $\ph^0_k$ and $\ph^0_{k+1}$ and the points $\ph^j_k$ in the other half of the fundamental polygon across the line through $\ph^0_{k-1}$ and $\ph^0_k$. In each case, the image of the $\ph^j_k$ after reflection is denoted by $\qh^j_k$ and its projection $\proj(\qh^j_k)$ is added to $\dummy$. 
\end{enumerate}

One of the major advantages of this algorithm is that the combinatorics of the resulting Delaunay triangulation can be explicitly described. Below we describe these combinatorics. The proof that this is the Delaunay triangulation of the dummy point set will be given in Lemma \ref{lem:structuredalgoemptydisks}. Again, see Figure \ref{fig:pizzaslice} for an illustration of the dummy points and the triangulation $\mathcal{T}$ within one slice of the $4g$-gon. 

\begin{definition}
	Define the infinite triangulation $\mathcal{T}$ of $\proj^{-1}(\dummy)$ as follows:
	\begin{itemize}
		\item As vertices take $\proj^{-1}(\dummy)$.
		\item The edges completely contained in the hyperbolic triangle $[O,\ph^0_0,\ph^0_1]$ are given by the following list.
		
		\begin{tabular}{l l l}
			$(\ph^j_k,\ph^{j+1}_k)$ & $j=0,\ldots,m-1,$ & $k=0,1,$ \\
			$(\ph^j_0,\ph^{j+1}_1)$ & $j=1,\ldots,m-1,$ & \\
			$(\ph^m_k,O)$ & & $k=0,1,$ \\ 
			$(\ph^j_0,\ph^j_1)$ & $j=1,\ldots,m,$ & \\
			$(\ph^j_k,\ph^*_0)$ & $j=0,1,$ & $k=0,1.$\\   
		\end{tabular}
		
		The other edges can be obtained as the images of the edges in the list above under the following maps:
		\begin{itemize}
			\item rotation around the origin by angle $\tfrac{k\pi}{2g}$,
			\item reflection in the line through $\ph^0_0,\ph^0_1$,
			\item reflection in the line through $\ph^0_0,\ph^0_1$, followed by rotation around the origin by angle $\tfrac{k\pi}{2g}$,
			\item any one of the above maps, followed by an element of $\Gamma_g$.
		\end{itemize}
	\end{itemize}
\end{definition}

Algorithm~\ref{structuredalgo2} shows the algorithm in pseudocode. We refer to this algorithm as the structured algorithm.  

\begin{algorithm}
	\SetKwInOut{Input}{Input}
	\SetKwInOut{Output}{Output}
	\DontPrintSemicolon
	\Input{hyperbolic surface $\Mg$}
	\Output{finite point set $\dummy\subset\Mg$ such that $\diam{\dummy}<\tfrac{1}{2}\sysg$}
	\BlankLine
	Initialize: let $\dummy$ be the set $\W$ of Weierstrass points of $\Mg$.\;
	Label the vertex by $v$ and the origin by $O$.\;
	For all $k=0,\ldots,4g-1$, label the midpoint of side $\sh_k$ by $\ph^0_k$.\;
	For all $k=0,\ldots,4g-1$, label the midpoint of $\ph^0_k$ and $\ph^0_{k+1}$ by $\ph^*_k$ and add $\proj(\ph_k^*)$ to $\dummy$.\;
	$m\leftarrow \lceil4\arcosh(\cot(\tfrac{\pi}{4g}))/\sysg\rceil-1$.\;
	\For {$i=1,\ldots,m$}{
		Let $\ph^j_0$ be the point on $[\ph^{j-1}_0,O]$ with $d(\ph^j_0,\ph^{j-1}_0)=\tfrac{1}{4}\sysg$ and add $\proj(\ph^j_0)$ to $\dummy$.\; 
		\For{$k=1,\ldots,2g-1$}{
			Let $\ph^j_k$ be $\ph^j_0$ rotated clockwise around the origin by angle $\tfrac{k\pi}{2g}$.\;
			Let $\qh^j_k$ be the reflection of $\ph^j_k$ in the line through $\ph^0_k,\ph^0_{k+1}$.\;
			Add $\proj(\ph^j_k)$ and $\proj(\qh^j_k)$ to $\dummy$.
		}
		\For{$k=2g,\ldots,4g-1$}{
			Let $\ph^j_k$ be $\ph^j_0$ rotated clockwise around the origin by angle $\tfrac{k\pi}{2g}$.\;
			Let $\qh^j_k$ be the reflection of $\ph^j_k$ in the line through $\ph^0_{k-1},\ph^0_{k}$.\;
			Add $\proj(\ph^j_k)$ and $\proj(\qh^j_k)$ to $\dummy$.
		}
	}
	\caption{Structured algorithm}\label{structuredalgo2}
\end{algorithm}

The main difference between the structured algorithm and the other two algorithms is that there is no {\bf while} loop in the structured algorithm, only {\bf for} loops. As a result, the cardinality of the resulting dummy point set is known precisely (see Theorem~\ref{thm:structuredalgorithm}). On the other hand, for the cardinality of the dummy point sets for the refinement or symmetric algorithm we can only give an estimate and the exact number of points depends on the implementation.

The following two lemmas show that the circumdiameters of triangles $\mathcal{T}$ are smaller than $\tfrac{1}{2}\sysg$ and that $\mathcal{T}$ is a Delaunay triangulation. As the proofs are rather technical, the reader may want to skip them at a first reading.

\begin{lemma}\label{lem:structuredalgoboundedcircumdiam}
	The circumdiameters of triangles in $\mathcal{T}$ are smaller than $\tfrac{1}{2}\sysg$.
\end{lemma}

\begin{proof}
	See again Figure~\ref{fig:pizzaslice}. By symmetry it is sufficient to consider the circumscribed disks of the triangles $$[\ph^0_0,\ph^1_0,\ph^*_0],[\ph^1_0,\ph^1_1,\ph^*_0],[\ph^m_0,\ph^m_1,O],[\ph^j_0,\ph^j_1,\ph^{j+1}_1],$$
	for $1\leq j\leq m-1$. For easy reference, the used lengths and angles satisfy the following relations. Here we denote $\hlen{[\ph,\qh]}$ by abuse of notation by $[\ph,\qh]$.
	\begin{align}
		[\ph^0_0,\ph^*_0]&=[\ph^0_0,\ph^1_0]=\tfrac{1}{4}\sysg,\label{eq:quartersystole}\\
		\cosh(\tfrac{1}{2}\sysg)&=1+2\cos(\tfrac{\pi}{2g}),\label{eq:halfsystole}\\
		\angle(O\ph^0_0\ph^*_0)&=\tfrac{\pi}{4},\label{eq:baseangle}\\
		%		\cosh([O,p^*_0])&=\dfrac{\cot(\tfrac{\pi}{4g})}{\cosh(\tfrac{1}{4}\sysg)},\label{eq:midpointsystole}\\
		%		\cos\angle(Op^0_0p^*_0)&=\dfrac{\tanh(\tfrac{1}{4}\sysg)}{\tanh[O,p^0_0]},\label{eq:cosbaseangle}\\
		%		\sin\angle(Op^0_0p^*_0)&=\dfrac{\sinh([O,p^*_0])}{\sqrt{\cot^2(\tfrac{\pi}{4g})-1}},\label{eq:sinbaseangle}\\
		\sinh(\tfrac{1}{2}[\ph^1_0,\ph^*_0])&=\sinh(\tfrac{1}{4}\sysg)\sin(\tfrac{1}{2}\angle(O\ph^0_0\ph^*_0)),\label{eq:firstdiagonal}\\
		\sin\angle(\ph^1_0\ph^*_0\ph^0_0)&=\dfrac{\sinh(\tfrac{1}{4}\sysg)\sin\angle(O\ph^0_0\ph^*_0)}{\sinh([\ph^1_0,\ph^*_0])},\label{eq:interiorangle}\\ 
		\angle(\ph^1_0\ph^*_0\ph^1_1)&=\pi-2\angle(\ph^1_0\ph^*_0\ph^0_0)\label{eq:topangle}.
		%		\sinh(\tfrac{1}{2}[p^1_0,p^1_1])&=\sinh([p^1_0,p^*_0])\sin(\tfrac{1}{2}\angle(p^1_0p^*_0p^1_1))\label{eq:firsthorizontal}. 
	\end{align}
	Equations~\eqref{eq:quartersystole} and~\eqref{eq:halfsystole} follow from the construction in the proof of Lemma~\ref{lem:systole} and from Theorem~\ref{thm:systolevalue}, respectively. Equation~\eqref{eq:baseangle} holds because $\angle(O\ph_0^0\vh_1)=\tfrac{\pi}{2}$ and the triangles $[O,\ph_0^0,\ph_1^0]$ and $[\vh_1,\ph_0^0,\ph_1^0]$ are congruent. Equation~\eqref{eq:firstdiagonal} follows from the application of~\cite[Theorem 7.11.2(ii)]{Beardon1983} to the triangle with vertices $\ph_0^1,\ph_0^0$ and the midpoint of $\ph_0^1$ and $\ph_0^*$: this is a right-angled triangle, because $[\ph_0^1,\ph_0^0,\ph_0^*]$ is isosceles. Equation~\eqref{eq:interiorangle} is the sine rule~\cite[Chapter 7.12]{Beardon1983} in triangle $[\ph_0^1,\ph_0^0,\ph_0^*]$. Finally, Equation~\eqref{eq:topangle} holds by symmetry.
	
	The circumradius $R$ of an isosceles triangle with legs of length $b$ and vertex angle $\alpha$ (see Figure~\ref{fig:circumradiusisoscelestriangle} satisfies
	
	\begin{equation}\label{eq:circumradius}
		\tanh(R)=\dfrac{\tanh(\tfrac{1}{2}c)}{\cos(\tfrac{1}{2}\alpha)},
	\end{equation}
	
	which can be derived by applying~\cite[Theorem 7.1..2(iii)]{Beardon1983} to the interior triangle with edges of length $R$ and $\tfrac{1}{2}b$.
	
	\begin{figure}
		\centering
		\includegraphics[width=0.5\textwidth]{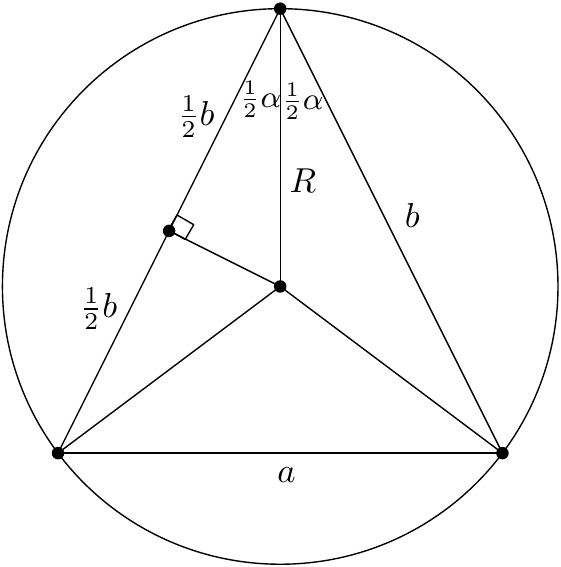}
		\caption{Deriving a formula for the circumradius of an isosceles triangle.}
		\label{fig:circumradiusisoscelestriangle}
	\end{figure}
	
	First, consider the circumradius $R(\ph^0_0,\ph^1_0,\ph^*_0)$ of $[\ph^0_0,\ph^1_0,\ph^*_0]$. By the above observation,
	\begin{align*}
		\tanh(R(\ph^0_0,\ph^1_0,\ph^*_0))&=\dfrac{\tanh(\tfrac{1}{8}\sysg)}{\cos(\tfrac{1}{2}\angle(O\ph^0_0\ph^*_0))},\\
		&=\sqrt{4-2\sqrt{2}}\tanh(\tfrac{1}{8}\sysg).
	\end{align*}
	We want to prove that $R(\ph^0_0,\ph^1_0,\ph^*_0)<\tfrac{1}{4}\sysg$, or, equivalently, that $$\tanh(R(\ph^0_0,\ph^1_0,\ph^*_0))<\tanh(\tfrac{1}{4}\sysg).$$ 
	Because $\tanh(\tfrac{1}{8}\sysg)$ is strictly increasing as function of $g$, we know that 
	$$ \sqrt{4-2\sqrt{2}}\tanh(\tfrac{1}{8}\sysg)<\sqrt{4-2\sqrt{2}}\lim_{g\rightarrow\infty}\tanh(\tfrac{1}{8}\sysg)\approx 0.448. $$
	For the same reason, 
	$$ \tanh(\tfrac{1}{4}\sysg)\geq \tanh(\tfrac{1}{4}\sysg)\big|_{g=2}\approx 0.643.$$
	It follows that $\tanh(R(\ph^0_0,\ph^1_0,\ph^*_0))<\tanh(\tfrac{1}{4}\sysg)$ holds. This proves that the circumdiameter of $[\ph^0_0,\ph^1_0,\ph^*_0]$ is smaller than $\tfrac{1}{2}\sysg$.\\
	Second, consider the circumradius $R(\ph^1_0,\ph^1_1,\ph^*_0)$ of $[\ph^1_0,\ph^1_1,\ph^*_0]$. By a similar computation as above, we see that
	\begin{align*}
		\tanh(R(\ph^1_0,\ph^1_1,\ph^*_0))&=\dfrac{\tanh(\tfrac{1}{2}[\ph^1_0,\ph^*_0])}{\cos(\tfrac{1}{2}\angle(\ph^1_0\ph^*_0\ph^1_1))},\\
		&\overset{\eqref{eq:topangle}}{=}\dfrac{\tanh(\tfrac{1}{2}[\ph^1_0,\ph^*_0])}{\sin\angle(\ph^1_0\ph^*_0\ph^0_0)}.
	\end{align*}
	By using
	\begin{align*}
		\sin\angle(\ph^1_0\ph^*_0\ph^0_0)&\overset{\eqref{eq:interiorangle}}{=}\dfrac{\sinh(\tfrac{1}{4}\sysg)\sin\angle(O\ph^0_0\ph^*_0)}{\sinh([\ph^1_0,\ph^*_0])},\\
		&=\dfrac{\sinh(\tfrac{1}{4}\sysg)\cdot 2\sin(\tfrac{1}{2}\angle(O\ph^0_0\ph^*_0))\cos(\tfrac{1}{2}\angle(O\ph^0_0\ph^*_0))}{\sinh([\ph^1_0,\ph^*_0])},\\
		&\overset{\eqref{eq:firstdiagonal}}{=}\dfrac{\sinh(\tfrac{1}{2}[\ph^1_0,\ph^*_0])\cdot 2\cos(\tfrac{1}{2}\angle(O\ph^0_0\ph^*_0))}{2\sinh(\tfrac{1}{2}[\ph^1_0,\ph^*_0])\cosh(\tfrac{1}{2}[\ph^1_0,\ph^*_0])},\\
		&=\dfrac{\cos(\tfrac{1}{2}\angle(O\ph^0_0\ph^*_0))}{\cosh(\tfrac{1}{2}[\ph^1_0,\ph^*_0])},
	\end{align*}
	we can rewrite this as
	\begin{align*}
		\tanh(R(\ph^1_0,\ph^1_1,\ph^*_0))&=\dfrac{\tanh(\tfrac{1}{2}[\ph^1_0,\ph^*_0])\cosh(\tfrac{1}{2}[\ph^1_0,\ph^*_0])}{\cos(\tfrac{1}{2}\angle(O\ph^0_0\ph^*_0))},\\
		&=\dfrac{\sinh(\tfrac{1}{2}[\ph^1_0,\ph^*_0])}{\cos(\tfrac{1}{2}\angle(O\ph^0_0\ph^*_0))},\\
		&\overset{\eqref{eq:firstdiagonal}}{=}\sinh(\tfrac{1}{4}\sysg)\tan(\tfrac{1}{2}\angle(O\ph^0_0\ph^*_0)),\\
		&=(\sqrt{2}-1)\sinh(\tfrac{1}{4}\sysg).
	\end{align*}
	By the same reasoning as before,
	$$ \tanh(R(\ph^1_0,\ph^1_1,\ph^*_0))<(\sqrt{2}-1)\lim_{g\rightarrow\infty}\sinh(\tfrac{1}{4}\sysg)\approx 0.414, $$
	which shows that $\tanh(R(\ph^1_0,\ph^1_1,\ph^*_0))<\tanh(\tfrac{1}{4}\sysg)$. This proves that the circumdiameter of $[\ph^1_0,\ph^1_1,\ph^*_0]$ is smaller than $\tfrac{1}{2}\sysg$.\\
	Third, consider the circumradius $R(\ph^m_0,\ph^m_1,O)$ of $[\ph^m_0,\ph^m_1,O]$. We know that $[O,\ph^m_0]=x\leq \tfrac{1}{4}\sysg$. By a similar computation as above, we see that
	\begin{align*}
		\tanh(R(\ph^m_0,\ph^m_1,O))&=\dfrac{\tanh(\tfrac{1}{2}x)}{\cos(\tfrac{\pi}{4g})},\\
		&\leq \dfrac{\tanh(\tfrac{1}{8}\sysg)}{\cos(\tfrac{\pi}{8})},\\
		&=\tanh(R(\ph^0_0,\ph^1_0,\ph^*_0)),\\
		&<\tanh(\tfrac{1}{4}\sysg),
	\end{align*}
	from which we conclude that $R(\ph^m_0,\ph^m_1,O)<\tfrac{1}{4}\sysg$. This proves that the circumdiameter of $[\ph^m_0,\ph^m_1,O]$ is smaller than $\tfrac{1}{2}\sysg$.\\
	Finally, let $1\leq j\leq m-1$ and consider the circumradius $R(\ph^j_0,\ph^j_1,\ph^{j+1}_1)$ of $[\ph^j_0,\ph^j_1,\ph^{j+1}_1]$. In fact, the points $\ph^j_0,\ph^j_1,\ph^{j+1}_0,\ph^{j+1}_1$ are concircular due to symmetry, so the center of the circumscribed circle of $[\ph^j_0,\ph^j_1,\ph^{j+1}_1]$ lies on the line segment $[O,\ph^*_0]$. It follows that the circumradius of this disk decreases when the distance of $\ph^j_0$ to $[O,\ph^*_0]$ decreases, or equivalently, when the distance $[O,\ph^j_0]$ decreases. Hence, for all $1\leq j\leq m-1$, $R(\ph^j_0,\ph^j_1,\ph^{j+1}_1)<R(\ph^1_0,\ph^1_1,\ph^{2}_1)$. Therefore, it is sufficient to show that $R(\ph^1_0,\ph^1_1,\ph^{2}_1)<\tfrac{1}{4}\sysg$.\\
	In this case, we cannot use equation~\eqref{eq:circumradius}, since $[\ph^1_0,\ph^1_1,\ph^{2}_1]$ is not an isosceles triangle. There exists a more general expression for the circumradius of an arbitrary (not necessarily isosceles) triangle, but this will lead to unnecessarily long expressions. Instead, consider the circumcenter $\ch$ of $[\ph^1_0,\ph^1_1,\ph^{2}_1]$. See also Figure~\ref{fig:center} for a more detailed view of the relevant triangles.
	\begin{figure}[htbp]
		\centering
		\includegraphics[width=0.8\textwidth]{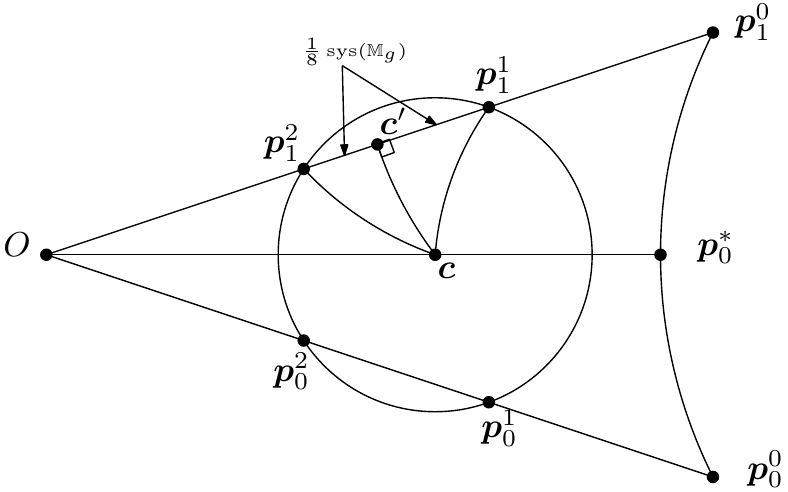}
		\caption{Close-up of situation at $[\ph^1_0,\ph^1_1,\ph^{2}_1]$}
		\label{fig:center}
	\end{figure}
	By the discussion above, we know that $\ch\in [O,\ph^*_0]$. Let $\ch'$ be the orthogonal projection of $\ch$ onto $[O,\ph^0_1]$. Since $[\ch,\ph^1_1,\ph^2_1]$ is isosceles, we know that 
	$$  [O,\ch']=[O,\ph^2_1]+\tfrac{1}{8}\sysg =[O,\ph^0_1]-\tfrac{3}{8}\sysg=\arcosh(\cot(\tfrac{\pi}{4g}))-\tfrac{3}{8}\sysg,$$
	from which it can be seen that $[O,\ch']$ is strictly increasing as function of $g$. 
	Furthermore,~\cite[Theorem 7.11.2(i)]{Beardon1983}
	$$ \tanh([\ch,\ch'])=\sinh([O,\ch'])\tan(\tfrac{\pi}{4g}),$$
	which after substitution of our expression for $[O,\ch']$ can be rewritten as
	$$  \tanh([\ch,\ch'])=\sqrt{1-\tan^2(\tfrac{\pi}{4g})}\cosh(\tfrac{3}{8}\sysg)-\sinh(\tfrac{3}{8}\sysg). $$
	By the Pythagorean law in $[\ch,\ch',\ph^2_1]$ we know 
	$$ \cosh(R(\ph^1_0,\ph^1_1,\ph^{2}_1))=\cosh(\tfrac{1}{8}\sysg)\cosh([\ch,\ch']).$$
	Using this expression of $R(\ph^1_0,\ph^1_1,\ph^{2}_1)$, it can be seen that $R(\ph^1_0,\ph^1_1,\ph^{2}_1)$ is strictly increasing as function of $g$. Therefore,
	$$ \cosh(R(\ph^1_0,\ph^1_1,\ph^{2}_1))<\lim_{g\rightarrow\infty}\cosh(R(\ph^1_0,\ph^1_1,\ph^{2}_1))\approx 1.140,$$
	which can be obtained by computing the corresponding limits of $\cosh(\tfrac{1}{8}\sysg)$ and $\cosh([\ch,\ch'])$. It follows that $\tanh(R(\ph^1_0,\ph^1_1,\ph^{2}_1))<0.480$, so $\tanh(R(\ph^1_0,\ph^1_1,\ph^{2}_1))<\tanh(\tfrac{1}{4}\sysg)$. This proves that the circumdiameter of $[\ph^1_0,\ph^1_1,\ph^{2}_1]$ is smaller than $\tfrac{1}{2}\sysg$. Since the circumdiameters of all four triangles are smaller than $\tfrac{1}{2}\sysg$ and since by symmetry every triangle is congruent to one of these four, this concludes the proof.   		
\end{proof}

\begin{lemma}\label{lem:structuredalgoemptydisks}
	The triangulation $\mathcal{T}$ is a Delaunay triangulation of $\proj^{-1}(\dummy)$. 
\end{lemma}

\begin{proof}
	We show that the circumdisks of triangles in $\mathcal{T}$ do not have vertices of $\mathcal{T}$ in their interior. Denote the circumscribed circle and \emph{open} circumscribed disk of a triangle $[\ph,\qh,\rh]$ by $C(\ph,\qh,\rh)$ and $D(\ph,\qh,\rh)$ respectively, and similarly when more than three points are concircular. Observe in particular that $C(\ph,\qh,\rh)\cap D(\ph,\qh,\rh)=\emptyset$. Denote the hyperbolic line through points $\ph,\qh$ by $L(\ph,\qh)$ and the open and closed line segments connecting $\ph,\qh$ by $(\ph,\qh),[\ph,\qh]$ respectively. By symmetry it is sufficient to consider only $D(\ph^0_0,\ph^1_0,\ph^*_0),D(\ph^1_0,\ph^1_1,\ph^*_0),D(\ph^m_0,\ph^m_1,O),D(\ph^j_0,\ph^j_1,\ph^{j+1}_1)$, for $1\leq j\leq m-1$. For convenience, we treat the cases for each circumdisk in a fixed order, namely
	\begin{enumerate}
		\item $O,\vh$
		\item $\ph^*_k,k=0,\ldots,4g-1$,
		\item $\ph^j_k,k=0,\ldots,4g-1,j=0,\ldots,m$,
		\item $\qh^j_k,k=0,\ldots,4g-1,j=1,\ldots,m$.
	\end{enumerate}
	First, consider $D(\ph^0_0,\ph^1_0,\ph^*_0)$.
	\begin{enumerate}
		\item Clearly, $O$ and $\vh$ are too far away from $\ph^*_0$ to be inside $D(\ph^0_0,\ph^1_0,\ph^*_0)$.
		\item Since $\ph^*_0\in C(\ph^0_0,\ph^1_0,\ph^*_0)$, we know that $\ph^*_0\not\in D(\ph^0_0,\ph^1_0,\ph^*_0)$. Since the center of $D(\ph^0_0,\ph^1_0,\ph^*_0)$ lies inside $D(\ph^0_0,\ph^1_0,\ph^*_0)$ on the bisector of angle $\angle(\ph^1_0\ph^0_0\ph^*_0)$ we see that $\ph^*_{4g-1}$ is farther away from this center than $\ph^*_0$. Since $\ph^*_0\not\in D(\ph^0_0,\ph^1_0,\ph^*_0)$, it follows that $\ph^*_{4g-1}\not\in D(\ph^0_0,\ph^1_0,\ph^*_0)$ as well. Since $D(\ph^0_0,\ph^1_0,\ph^*_0)\cap C(\ph^*_0,\ldots,\ph^*_{4g-1})$ is contained in the shortest open chord of $C(\ph^*_0,\ldots,\ph^*_{4g-1})$ between $\ph^*_0$ and $\ph^*_{4g-1}$, we see that $\ph^*_k\not\in D(\ph^0_0,\ph^1_0,\ph^*_0)$ for $k=0,\ldots,4g-1$. 
		\item Since the center of $D(\ph^0_0,\ph^1_0,\ph^*_0)$ is in the interior of $[\ph^0_0,\ph^1_0,\ph^*_0]\subset[O,\ph^0_0,\ph^*_0]$, we know that $\ph^j_0$ is closer to the center of $D(\ph^0_0,\ph^1_0,\ph^*_0)$ than $\ph^j_k$ for $j=0,\ldots,m$ and $k\neq 0$.
		Because $L(O,\ph^0_0)\cap D(\ph^0_0,\ph^1_0,\ph^*_0)=(\ph^0_0,\ph^1_0)$, we know that $\ph^j_0\not\in D(\ph^0_0,\ph^1_0,\ph^*_0)$ for all $j=0,\ldots,m$. Therefore, $\ph^j_k\not\in D(\ph^0_0,\ph^1_0,\ph^*_0)$ for $j=0,\ldots,m$ and $k=0,\ldots,4g-1$. 
		\item By a reasoning similar to above, $\ph^j_k$ is closer to the center of $D(\ph^0_0,\ph^1_0,\ph^*_0)$ than $\qh^j_k$. Since by the previous step $\ph^j_k\not\in D(\ph^0_0,\ph^1_0,\ph^*_0)$ for $j=1,\ldots,m$ and $k=0,\ldots,4g-1$, it follows that $\qh^j_k\not\in D(\ph^0_0,\ph^1_0,\ph^*_0)$ for $j=1,\ldots,m$ and $k=0,\ldots,4g-1$ as well.  
	\end{enumerate}
	Second, consider $D(\ph^1_0,\ph^1_1,\ph^*_0)$. 
	\begin{enumerate}
		\item Clearly, $O$ and $\vh$ are too far away from $\ph^*_0$ to be inside $D(\ph^1_0,\ph^1_1,\ph^*_0)$.
		\item The circle $C(\ph^1_0,\ph^1_1,\ph^*_0)$ is tangent to $C(\ph^*_0,\ldots,\ph^*_{4g-1})$, because both circles have their center on $[O,\ph_0^*]$ and pass through $\ph_0^*$. Since the radius of $C(\ph^1_0,\ph^1_1,\ph^*_0)$ is smaller than the radius of $C(\ph^*_0,\ldots,\ph^*_{4g-1})$, this means that $D(\ph^1_0,\ph^1_1,\ph^*_0)\subseteq D(\ph^*_0,\ldots,\ph^*_{4g-1})$. Therefore, $\ph^*_k\not\in D(\ph^1_0,\ph^1_1,\ph^*_0)$ for all $k=0,\ldots,4g-1$.
		\item First, to prove that $\ph^j_0\not\in D(\ph^1_0,\ph^1_1,\ph^*_0)$ for all $j=0,\ldots,m$, we now show that $D(\ph^1_0,\ph^1_1,\ph^*_0)\cap(O,\ph^1_0)=\emptyset$. First observe that the line $L(O,\ph^0_0)$ intersects $C(\ph^1_0,\ph^1_1,\ph^*_0)$ in one or two points. If the intersection consists of one point, then it has to be $\ph^1_0$ and we are done. If the intersection consists of two points, then it is sufficient to show that $\ph^1_0$ is the closest of these two to $O$. Let $\ph^M_0$ denote the midpoint of $\ph^1_0$ and $\ph^1_1$. It is sufficient to show that $[\ph^*_0,\ph^M_0]\geq R(\ph^1_0,\ph^1_1,\ph^*_0)$, since then the center of $D(\ph^1_0,\ph^1_1,\ph^*_0)$ is contained in $[\ph^1_0,\ph^1_1,\ph^*_0]$. It is known that~\cite[Theorem 7.11.2(iii)]{Beardon1983}
		$$\tanh([\ph^*_0,\ph^M_0])=\tanh([\ph^1_0,\ph^*_0])\cos\tfrac{1}{2}\angle(\ph^1_0\ph^*_0\ph^1_1)=\tanh([\ph^1_0,\ph^*_0])\sin\angle(\ph^1_0\ph^*_0\ph^0_0),$$
		where the second equality follows from Equation~\eqref{eq:topangle} in the proof of Lemma~\ref{lem:structuredalgoboundedcircumdiam}. Therefore, $[\ph^*_0,\ph^M_0]\geq R(\ph^1_0,\ph^1_1,\ph^*_0)$ is equivalent with the following sequence of inequalities:
		\begin{align*}
			\tanh([\ph^*_0,\ph^M_0])&\geq \tanh(R(\ph^1_0,\ph^1_1,\ph^*_0)),\\
			\tanh([\ph^1_0,\ph^*_0])\sin\angle(\ph^1_0\ph^*_0\ph^0_0)&\geq \dfrac{\tanh(\tfrac{1}{2}[\ph^1_0,\ph^*_0])}{\sin\angle(\ph^1_0\ph^*_0\ph^0_0)},\\
			\sin^2\angle(\ph^1_0\ph^*_0\ph^0_0)&\geq \dfrac{\tanh(\tfrac{1}{2}[\ph^1_0,\ph^*_0])}{\tanh([\ph^1_0,\ph^*_0])}.
		\end{align*}
		Since 
		$$ \dfrac{\tanh(\tfrac{1}{2}[\ph^1_0,\ph^*_0])}{\tanh([\ph^1_0,\ph^*_0])}=\tfrac{1}{2}\tanh^2(\tfrac{1}{2}[\ph^1_0,\ph^*_0])+\tfrac{1}{2},$$
		and since $[\ph^1_0,\ph^*_0]$ is strictly increasing by \eqref{eq:firstdiagonal} in the proof of Lemma~\ref{lem:structuredalgoboundedcircumdiam}, we find that
		$$ \dfrac{\tanh(\tfrac{1}{2}[\ph^1_0,\ph^*_0])}{\tanh([\ph^1_0,\ph^*_0])}\leq \lim_{g\rightarrow\infty}\dfrac{\tanh(\tfrac{1}{2}[\ph^1_0,\ph^*_0])}{\tanh([\ph^1_0,\ph^*_0])}\approx 0.542.$$
		Furthermore, since
		\begin{align*}
			\sin^2\angle(\ph^1_0\ph^*_0\ph^0_0)&\overset{\eqref{eq:interiorangle}}{=}\dfrac{\sinh^2(\tfrac{1}{4}\sysg)\sin^2\angle(O\ph^0_0\ph^*_0)}{\sinh^2([\ph^1_0,\ph^*_0])},\\
			&\overset{\eqref{eq:firstdiagonal}}{=}\dfrac{\sinh^2(\tfrac{1}{2}[\ph^1_0,\ph^*_0])\sin^2\angle(O\ph^0_0\ph^*_0)}{\sin^2(\tfrac{1}{2}\angle(O\ph^0_0\ph^*_0))\sinh^2([\ph^1_0,\ph^*_0])},\\
			&=\dfrac{\cos^2(\tfrac{1}{2}\angle(O\ph^0_0\ph^*_0))}{\cosh^2(\tfrac{1}{2}[\ph^1_0,\ph^*_0])},
		\end{align*}
		and since $\angle(O\ph^0_0\ph^*_0)$ is constant and $[\ph^1_0,\ph^*_0]$ strictly increasing, we see that $\sin^2\angle(\ph^1_0\ph^*_0\ph^0_0)$ is strictly decreasing, so
		$$ \sin^2\angle(\ph^1_0\ph^*_0\ph^0_0)\geq \lim_{g\rightarrow\infty} \sin^2\angle(\ph^1_0\ph^*_0\ph^0_0)\approx 0.744. $$
		From this we can conclude that 
		$$ \sin^2\angle(\ph^1_0\ph^*_0\ph^0_0)\geq \dfrac{\tanh(\tfrac{1}{2}[\ph^1_0,\ph^*_0])}{\tanh([\ph^1_0,\ph^*_0])}$$
		holds, which by the chain of equivalent inequalities means that $[\ph^*_0,\ph^M_0]\geq R(\ph^1_0,\ph^1_1,\ph^*_0)$. It follows that if $L(O,\ph^0_0)\cap C(\ph^1_0,\ph^1_1,\ph^*_0)$ consists of two points, then $\ph^1_0$ is the closest of these two. This implies that $D(\ph^1_0,\ph^1_1,\ph^*_0)\cap(O,\ph^1_0)=\emptyset$. We conclude that $\ph^j_0\not\in D(\ph^1_0,\ph^1_1,\ph^*_0)$ for all $j=0,\ldots,m$. By symmetry, we see that $\ph^j_1\not\in D(\ph^1_0,\ph^1_1,\ph^*_0)$ for all $j=0,\ldots,m$.\\
		Second, by the reasoning above we see that $D(\ph^1_0,\ph^1_1,\ph^*_0)$ is contained in the union of the triangle $[O,\ph^0_0,\ph^0_1]$ and the (open) annulus $$D(\ph^*_0,\ldots,\ph^*_{4g-1})\setminus (D(\ph^1_0,\ldots,\ph^1_{4g-1})\cup C(\ph^1_0,\ldots,\ph^1_{4g-1}))$$
		centered at $O$, with boundary passing through $\ph^*_0$ on one side and through $\ph^1_0$ on the other side. Combining this with $\ph^j_k\not\in D(\ph^1_0,\ph^1_1,\ph^*_0)$ for $j=0,\ldots,m$ and $k=0,1$, we can immediately conclude that $\ph^j_k\not\in D(\ph^1_0,\ph^1_1,\ph^*_0)$ for $j=0,\ldots,m$ and $k=0,\ldots,4g-1$.
		\item As we have seen before, $C(\ph^1_0,\ph^1_1,\ph^*_0)$ is tangent to $C(\ph^*_0,\ldots,\ph^*_{4g-1})$, which means that $D(\ph^1_0,\ph^1_1,\ph^*_0)$ is contained in the interior of the $4g$-gon $[\ph^0_0,\ldots,\ph^0_k,\ldots,\ph^0_{4g-1}]$. Therefore, $\qh^j_k\not\in D(\ph^1_0,\ph^1_1,\ph^*_0)$ for $j=1,\ldots,m$ and $k=0,\ldots,4g-1$. 
	\end{enumerate}
	Third, consider $D(\ph^m_0,\ph^m_1,O)$.
	\begin{enumerate}
		\item Since $O\in C(\ph^m_0,\ph^m_1,O)$, we know that $O\not\in D(\ph^m_0,\ph^m_1,O)$. Clearly, $\vh$ is too far away from $O$ to be inside $D(\ph^m_0,\ph^m_1,O)$. 
		\item Clearly, the points $\ph^*_k$ for $k=0,\ldots,4g-1$ are too far away from $O$ to be inside $D(\ph^m_0,\ph^m_1,O)$.
		\item Since $D(\ph^m_0,\ph^m_1,O)$ is contained in the union of the disk $D(\ph^m_0,\ph^m_1,\ldots,\ph^m_{4g-1})$ and the (closed) triangle $[O,\ph^0_0,\ph^0_1]$, we immediately see that $\ph^j_k\not\in D(\ph^m_0,\ph^m_1,O)$ for all $j=0,\ldots,m$ and $k\neq 0,1$. Since $L(O,\ph^m_0)\cap D(\ph^m_0,\ph^m_1,O)=(O,\ph^m_0)$, it follows that $\ph^j_0\not\in D(\ph^m_0,\ph^m_1,O)$ for $j=0,\ldots,m$. Similarly, $\ph^j_1\not\in D(\ph^m_0,\ph^m_1,O)$ for $j=0,\ldots,m$. Therefore, $\ph^j_k\not\in D(\ph^m_0,\ph^m_1,O)$ for all $j=0,\ldots,m$ and $k=0,\ldots,4g-1$. 
		\item Clearly, the points $\qh^j_k$ for $j=1,\ldots,m$ and $k=0,\ldots,4g-1$ are too far away from $O$ to be inside $D(\ph^m_0,\ph^m_1,O)$. 
	\end{enumerate}
	Finally, let $1\leq j\leq n-1$ and consider $D(\ph^j_0,\ph^j_1,\ph^{j+1}_1)$.
	\begin{enumerate}
		\item Clearly, $\vh$ is too far away from the center of $D(\ph^j_0,\ph^j_1,\ph^{j+1}_1)$ to be inside $D(\ph^j_0,\ph^j_1,\ph^{j+1}_1)$. Moreover, $O\not\in D(\ph^j_0,\ph^j_1,\ph^{j+1}_1)$, since $L(O,\ph^0_1)\cap D(\ph^j_0,\ph^j_1,\ph^{j+1}_1)=(\ph^j_1,\ph^{j+1}_1)$. 
		\item Of the set of disks $\{D(\ph^j_0,\ph^j_1,\ph^{j+1}_1),j=1,\ldots,m-1\}$, the one that is closest to $\ph^*_0$ is $D(\ph^1_0,\ph^1_1,\ph^{2}_1)$, i.e., if $\ph^*_0\not\in D(\ph^1_0,\ph^1_1,\ph^{2}_1)$, then $\ph^*_0\not\in D(\ph^j_0,\ph^j_1,\ph^{j+1}_1)$ for $j=1,\ldots,m-1$. Observe that $C(\ph^1_0,\ph^1_1,\ph^{2}_1)$ and $C(\ph^*_0,\ph^1_0,\ph^1_1)$ intersect in the points $\ph^1_0,\ph^1_1$. Since the center of $C(\ph^1_0,\ph^1_1,\ph^{2}_1)$ is closer to $O$ than the center of $C(\ph^*_0,\ph^1_0,\ph^1_1)$, we can conclude that $\ph^*_0\not\in D(\ph^1_0,\ph^1_1,\ph^{2}_1)$. Therefore, $\ph^*_0\not\in D(\ph^j_0,\ph^j_1,\ph^{j+1}_1)$ for $j=1,\ldots,m-1$. It follows that $D(\ph^j_0,\ph^j_1,\ph^{j+1}_1)\subseteq D(\ph^*_0,\ldots,\ph^*_{4g-1})$ for $j=1,\ldots,m-1$, which implies that $\ph^*_k\not\in D(\ph^j_0,\ph^j_1,\ph^{j+1}_1)$ for $j=1,\ldots,m-1$ and $k=0,\ldots,4g-1$. 
		\item Since $\ph^j_0,\ph^j_1,\ph^{j+1}_0,\ph^{j+1}_1$ are concircular, we see that $D(\ph^j_0,\ph^j_1,\ph^{j+1}_1)$ is contained in the union of the (closed) triangle $[O,\ph^0_0,\ph^0_1]$ and the (open) annulus
		$$ D(\ph^j_0,\ldots,\ph^j_{4g-1})\setminus (D(\ph^{j+1}_0,\ldots,\ph^{j+1}_{4g-1})\cup C(\ph^{j+1}_0,\ldots,\ph^{j+1}_{4g-1}))$$
		centered at $O$, with boundary passing through $\ph^j_0$ on one side and through $\ph^{j+1}_0$ on the other side. Therefore, $\ph^j_k\not\in D(\ph^j_0,\ph^j_1,\ph^{j+1}_1)$ for $j=0,\ldots,m$ and $k\neq 0,1$. Furthermore, since $L(O,\ph^0_0)\cap D(\ph^j_0,\ph^j_1,\ph^{j+1}_1)=(\ph^j_0,\ph^{j+1}_0)$, we see that $\ph^j_0\not\in D(\ph^j_0,\ph^j_1,\ph^{j+1}_1)$ for $j=0,\ldots,m$. Similarly, $\ph^j_1\not\in D(\ph^j_0,\ph^j_1,\ph^{j+1}_1)$ for $j=0,\ldots,m$. We conclude that $\ph^j_k\not\in D(\ph^j_0,\ph^j_1,\ph^{j+1}_1)$ for $j=0,\ldots,m$ and $k=0,\ldots,4g-1$. 
		\item Clearly, $D(\ph^j_0,\ph^j_1,\ph^{j+1}_1)$ is contained in the $4g$-gon $[\ph^0_0,\ldots,\ph^0_k,\ldots,\ph^0_{4g-1}]$, which means that $\qh^j_k\not\in D(\ph^j_0,\ph^j_1,\ph^{j+1}_1)$ for $j=1,\ldots,m$ and $k=0,\ldots,4g-1$.
	\end{enumerate}
	Since each triangle of the infinite triangulation $\mathcal{T}$ is congruent to one of the triangles above and since the circumdisk of each of the above triangles is empty, it follows that $\mathcal{T}$ is a Delaunay triangulation. 
\end{proof}

From these two lemmas the main statement of this subsection follows directly. 

\begin{theorem}\label{thm:structuredalgorithm}
	The structured algorithm terminates. The resulting dummy point $\dummy$ set satisfies $\diam{\dummy}<\tfrac{1}{2}\sysg$ and its cardinality $|\dummy|$ is equal to
	$$ 6g+2+8g(\lceil4\arcosh(\cot(\tfrac{\pi}{4g}))/\sysg\rceil-1).$$
	A Delaunay triangulation $\dth{\proj^{-1}(\dummy)}$ is given by $\mathcal{T}$. 
\end{theorem}

\begin{proof}
	Termination of the algorithm is trivial. By Lemma~\ref{lem:structuredalgoboundedcircumdiam}, the resulting dummy point set satisfies $\diam{\dummy}<\tfrac{1}{2}\sysg$. The cardinality of $\dummy$ can be computed as follows. In line 1, $\dummy$ contains the $2g+2$ Weierstrass points of $\Mg$. In line 4 we add $4g$ points to $\dummy$. There are 
	$$ m=\lceil4\arcosh(\cot(\tfrac{\pi}{4g}))/\sysg\rceil-1$$
	iterations of the {\bf for} loop in line 6, each adding $8g$ points $\dummy$. The cardinality of $\dummy$ is obtained by adding these expressions. By Lemma~\ref{lem:structuredalgoemptydisks}, $\mathcal{T}$ is a Delaunay triangulation. This finishes the proof.
\end{proof}

%%%%%%%%%%%%%%%%%%%%%%%%%%%%%%%%%%%%%%%%%%%%%%%%%%%%%%%%%%%%%

\bigskip 

\bibliography{refs}
\bibliographystyle{plainurl}

\end{document}